%% file: main.tex
\documentclass
[11pt]
{article}

\usepackage{preamble}

\usepackage{verbatim}
\usepackage{xspace}
\usepackage{enumitem}
\usepackage[dvipsnames]{xcolor}
\usepackage[T1]{fontenc}
\usepackage[full]{textcomp}
\usepackage[american]{babel}
\usepackage{mathtools}

\usepackage{textcomp} %
\usepackage[scr=rsfso]{mathalfa}%
\usepackage{cancel} %
\usepackage{microtype}

\let\originalleft\left
\let\originalright\right
\renewcommand{\left}{\mathopen{}\mathclose\bgroup\originalleft}
\renewcommand{\right}{\aftergroup\egroup\originalright}
\usepackage{turnstile}
\usepackage{mdframed}
\usepackage{tikz}
\usetikzlibrary{shapes,arrows}
\usetikzlibrary{positioning}
\usetikzlibrary{decorations.pathreplacing}
\usetikzlibrary{decorations.text}
\usetikzlibrary{decorations.pathmorphing}
\usepackage{caption}
\DeclareCaptionType{Algorithm}
\usepackage{newfloat}
\usepackage{array}
\usepackage{subfig}
\usepackage{xparse}

\makeatletter
\let\latexparagraph\paragraph
\RenewDocumentCommand{\paragraph}{som}{%
  \IfBooleanTF{#1}
    {\latexparagraph*{#3}}
    {\IfNoValueTF{#2}
       {\latexparagraph{\maybe@addperiod{#3}}}
       {\latexparagraph[#2]{\maybe@addperiod{#3}}}%
  }%
}
\newcommand{\maybe@addperiod}[1]{%
  #1\@addpunct{.}%
}
\makeatother

\newcommand{\Authornote}[2]{{\sffamily\small\color{red}{}}}

\newcommand{\Authorcomment}[2]{{\sffamily\small\color{blue}{[#1: #2]}}}

\newif\ifdraft
\draftfalse

\ifdraft
    \newcommand{\Jnote}{\Authornote{J}}
    \newcommand{\Snote}{\Authorcomment{sitan}}
    \newcommand{\Wnote}{\Authornote{W}}
    \newcommand{\Mnote}{\Authornote{M}}
    \newcommand{\sitan}{\Snote}
    \newcommand{\walt}[1]{{\color{green}[walt: #1]}}
    \newcommand{\mahbod}[1]{{\color{red}[Mahbod: #1]}}
\else
    \newcommand{\Jnote}[1]{}
    \newcommand{\Snote}[1]{}
    \newcommand{\Wnote}[1]{}
    \newcommand{\Mnote}[1]{}
    \newcommand{\sitan}[1]{}
    \newcommand{\walt}[1]{}
    \newcommand{\mahbod}[1]{}
\fi

\usepackage{boxedminipage}
\newcommand{\paren}[1]{(#1)}
\newcommand{\Paren}[1]{\Bigl(#1\Bigr)}

\newcommand{\Brac}[1]{\left[#1\right]}

\newcommand{\Abs}[1]{\left\lvert#1\right\rvert}

\newcommand{\Bigabs}[1]{\Big\lvert#1\Big\rvert}

\renewcommand{\set}[1]{\{#1\}}
\renewcommand{\Set}[1]{\left\{#1\right\}}

\newcommand{\Bigset}[1]{\Big\{#1\Big\}}
\renewcommand{\norm}[1]{\lVert#1\rVert}
\newcommand{\Norm}[1]{\left\lVert#1\right\rVert}

\newcommand{\normt}[1]{\norm{#1}_2}

\newcommand{\Snormt}[1]{\Norm{#1}^2_2}

\renewcommand{\iprod}[1]{\langle#1\rangle}

\newcommand{\Esymb}{\mathbb{E}}
\newcommand{\Psymb}{\mathbb{P}}

\DeclareMathOperator*{\E}{\Esymb}

\DeclareMathOperator*{\ProbOp}{\Psymb}
\renewcommand{\Pr}{\ProbOp}

\newcommand{\defeq}{\stackrel{\mathrm{def}}=}
\newcommand{\seteq}{\mathrel{\mathop:}=}
\newcommand{\from}{\colon}

\newcommand{\mper}{\,.}
\newcommand{\mcom}{\,,}

\newcommand\bdot\bullet

\DeclareMathOperator{\Tr}{Tr}

\DeclareMathOperator{\poly}{poly}

\DeclareMathOperator{\supp}{supp}

\newcommand{\cA}{\mathcal A}
\newcommand{\cB}{\mathcal B}

\newcommand{\cD}{\mathcal{D}}

\newcommand{\cL}{\mathcal L}
\newcommand{\cM}{\mathcal M}
\newcommand{\cN}{\mathcal N}
\newcommand{\cO}{\mathcal O}

\newcommand{\cR}{\mathcal R}
\newcommand{\cS}{\mathcal S}

\newcommand{\cX}{\mathcal X}

\renewcommand{\leq}{\leqslant}
\renewcommand{\le}{\leqslant}
\renewcommand{\geq}{\geqslant}
\renewcommand{\ge}{\geqslant}
\let\epsilon=\varepsilon
\newcommand\MYcurrentlabel{xxx}
\newcommand{\MYstore}[2]{%
  \global\expandafter \def \csname MYMEMORY #1 \endcsname{#2}%
}
\newcommand{\MYload}[1]{%
  \csname MYMEMORY #1 \endcsname%
}
\newcommand{\MYnewlabel}[1]{%
  \renewcommand\MYcurrentlabel{#1}%
  \MYoldlabel{#1}%
}
\newcommand{\MYdummylabel}[1]{}
\newcommand{\torestate}[1]{%
  \let\MYoldlabel\label%
  \let\label\MYnewlabel%
  #1%
  \MYstore{\MYcurrentlabel}{#1}%
  \let\label\MYoldlabel%
}
\newcommand{\restatetheorem}[1]{%
  \let\MYoldlabel\label
  \let\label\MYdummylabel
  \begin{theorem*}[Restatement of \cref{#1}]
    \MYload{#1}
  \end{theorem*}
  \let\label\MYoldlabel
}
\newcommand{\restatelemma}[1]{%
  \let\MYoldlabel\label
  \let\label\MYdummylabel
  \begin{lemma*}[Restatement of \cref{#1}]
    \MYload{#1}
  \end{lemma*}
  \let\label\MYoldlabel
}
\newcommand{\restateprop}[1]{%
  \let\MYoldlabel\label
  \let\label\MYdummylabel
  \begin{proposition*}[Restatement of \cref{#1}]
    \MYload{#1}
  \end{proposition*}
  \let\label\MYoldlabel
}
\newcommand{\restatefact}[1]{%
  \let\MYoldlabel\label
  \let\label\MYdummylabel
  \begin{fact*}[Restatement of \cref{#1}]
    \MYload{#1}
  \end{fact*}
  \let\label\MYoldlabel
}
\newcommand{\restate}[1]{%
  \let\MYoldlabel\label
  \let\label\MYdummylabel
  \MYload{#1}
  \let\label\MYoldlabel
}

\newcommand{\e}{\epsilon}
\newcommand{\eps}{\epsilon}

\allowdisplaybreaks
\newcommand*{\normop}[1]{\norm{#1}_{\mathrm{op}}}
\newcommand*{\Normop}[1]{\Norm{#1}_{\mathrm{op}}}

\newcommand*{\normf}[1]{\norm{#1}_{\mathrm{F}}}

\newcommand{\Ber}{\mathrm{Ber}}

\newenvironment{algorithmbox}{\begin{mdframed}[nobreak=true]
\begin{algorithm}}{\end{algorithm}\end{mdframed}}

\providecommand{\post}{\mathrm{post}}
\providecommand{\wpost}{w_{\mathrm{out}}}
\providecommand{\out}{\mathrm{out}}
\providecommand{\Xinput}{X_{\mathrm{input}}}
\providecommand{\yinput}{y_{\mathrm{input}}}
\providecommand{\cAres}{\cA_{\mathrm{res}}}

\providecommand{\wnull}{{w^*}}
\providecommand{\Xnull}{{X^\circ}}
\providecommand{\ynull}{{y^\circ}}

\providecommand{\ynullp}{\ynull^{\prime}}

\providecommand{\LS}{\text{LS}}
\providecommand{\init}{\text{init}}

\newcommand{\pE}{\tilde{\mathbb{E}}}

\newcommand*{\transpose}[1]{{#1}{}^{\mkern-1.5mu\mathsf{T}}}

\newcommand{\ball}[2]{\mathbb{B}_2^{#1}(#2)}

\newcommand{\Adv}{\mathrm{Adv}}
\newcommand{\tE}{\tilde{\E}}

\renewcommand{\triangleq}{\coloneqq}
\renewcommand{\defeq}{\coloneqq}

\title{Computation-Utility-Privacy Tradeoffs in Bayesian Estimation}

\author{
    Sitan Chen\thanks{Email: \texttt{sitan@seas.harvard.edu}, supported in part by NSF CAREER award CCF-2441635} \\
    Harvard
        \and 
    Jingqiu Ding \thanks{Email: \texttt{jingqiu.ding@inf.ethz.ch}, supported by European Research Council (ERC) under the European Union’s Horizon 2020 research and innovation programme (grant agreement No 815464)} \\
    ETH Zurich
        \and
    Mahbod Majid \thanks{Email: \texttt{mahbod@mit.edu}} \\
    MIT
        \and
    Walter McKelvie \thanks{Email: \texttt{wmckelvie@fas.harvard.edu}, supported in part by NSF GRFP award  2140743 and CAREER award CCF-2441635} \\
    Harvard}

\date{March 18, 2026}

\DeclareMathOperator{\tr}{tr}

\begin{document}

\pagestyle{empty}

\maketitle

 \input{content/abstract}

\pagestyle{plain}
\setcounter{page}{1}

\newpage
\tableofcontents
\newpage

\input{content/introduction}

\section{Outlook}

In this work we took the first steps towards elucidating the computational and statistical complexity of private Bayesian estimation in high dimensions. We developed efficient algorithms and matching computational hardness under the low-degree framework, showing for the first time that private Bayesian estimation exhibits an intriguing information-computation gap even in the simplest settings of mean estimation and linear regression with Gaussian prior. Our work opens up a wealth of future directions to explore:

\paragraph{Beyond Gaussian prior.} There is a rich literature on algorithms and computational-statistical gaps for Bayesian inference in high dimensions for general priors (see, e.g., \cite{zdeborova2016statistical, krzakala2024statistical, montanari2024friendly} for overviews). Mapping out the phase diagram of when computationally efficient estimation is possible vs. any nontrivial estimation is possible draws upon a sophisticated toolbox of ideas originating from statistical physics, and characterizing how the phase diagram changes when one imposes privacy or robustness constraints is a fascinating direction for future study.

Relatedly, whereas in the case of simple exponential family priors like Gaussian, posterior sampling follows immediately from posterior mean estimation, this is not the case in general. For other priors, can one develop private or robust \emph{posterior sampling} algorithms in these more challenging settings?

\paragraph{Other inference tasks.} There is a wealth of other natural Bayesian problems that could be studied in the spirit of our work, e.g. linear quadratic estimation, community detection, matrix and tensor estimation, group testing, etc. As discussed in \cref{sec:relatedwork}, there have already been some works on robustness for some of these tasks, although the guarantees therein do not appear to be applicable yet to obtaining privacy guarantees. The intriguing phenomenology and technical subtleties surfaced in this work even for mean estimation and linear regression suggest that there are many more interesting insights to be uncovered for these other tasks.

\paragraph{Approximate DP} While both of our robust estimators would yield approx-DP estimators by \cite{hopkins2023robustness}, the same accuracy gap exists between these approx-DP estimators and the optimal statistical rate. However, it becomes much less clear how to prove a matching computational lower bound for the efficient approx-DP estimator, since the privacy-to-robustness reduction of \cite{Georgiev2022PrivacyIR} no longer applies. A natural follow-up step would be to tightly characterize the computational cost of Bayes estimation under approx-DP, as well as alternative notions of privacy.

\paragraph{Tight computational lower bound.} Our low-degree lower bound for private empirical mean estimation is tight only in the regime $\beta = \exp(-d)$, so in principle it does not preclude a better estimator with constant failure probability (or failure probability that shrinks only polynomially fast in $d$). It would be interesting to see whether a better efficient private estimator exists in one of these regimes, or whether our lower bound can be strengthened to preclude such an improvement.

\paragraph{Beyond bucketing.} We briefly remark upon the $\log d$ dependence hidden in the $\tilde{O}(\cdot)$ in \cref{thm:anisotropic-private-est}, which is not obviously necessary. It remains an open question whether empirical mean estimation---or even traditional parametric mean estimation---with anisotropic (known) covariance can be done optimally without $\log d$ dependence coming from bucketing, even information-theoretically. \cite{dagan2024dimension} achieved dimension-independence for parameter mean estimation in the \emph{approx}-DP setting using, in our notation,
\[
    n = \tilde O \biggl(\frac{\tr(\Lambda^2) + \|\Lambda^2\|_\mathrm{op} \log \frac{1}{\beta}}{\alpha^2} + \frac{(\tr(\Lambda) + \|\Lambda\|_\mathrm{op}\log \frac{1}{\beta}) \sqrt{\log \frac{1}{\delta}}}{\alpha \varepsilon} + \frac{\log \frac{1}{\delta}}{\varepsilon}\biggr)
\]
samples, and in fact it can be shown that their algorithm estimates the empirical mean with sample complexity given by the last two terms above.
However, there is a suboptimal coupling between $\tr(\Lambda)$ and $\sqrt{\log \frac{1}{\delta}}$ in the second term, and it appears difficult to make their trimming-based approach work in the pure-DP setting.
In the isotropic case, these issues are solvable by going through the robustness-to-privacy reduction of \cite{hopkins2023robustness}. Unfortunately, that approach requires reasoning about volume ratios, which appears to be incompatible with the anisotropic setting (fundamentally, the issue is that the volume of a scaled $\Sigma$-ball will scale as a power of $\text{rank}(\Sigma)$ rather than $\tr(\Sigma)$, which introduces a dependence on $d$).

\input{content/preliminaries}
\input{content/mean-estimation}

\input{content/online-streaming}
\input{content/LinearRegression}

\subsection*{Acknowledgments}

The authors thank Cynthia Dwork, Pranay Tankala, and Linjun Zhang for a helpful conversation at an early stage of this project, Sam Hopkins, Gautam Kamath, Adam Smith, and Salil Vadhan for insightful discussions as well as generous feedback on prior versions of this manuscript.

\phantomsection
\addcontentsline{toc}{section}{References}
\bibliographystyle{amsalpha}
\bibliography{bib/mathreview,bib/dblp,bib/custom,bib/scholar,bib/dpbayes}

\appendix

\input{content/appendix-mean-estimation}

\input{content/appendix-frequentist-implications.tex}
\input{content/appendix_LR}

\end{document}

%% file: content/abstract.tex
\begin{abstract}
    Bayesian methods lie at the heart of modern data science and provide a powerful scaffolding for estimation in data-constrained settings and principled quantification and propagation of uncertainty. Yet in many real-world use cases where these methods are deployed, there is a natural need to preserve the privacy of the individuals whose data is being scrutinized. While a number of works have attempted to approach the problem of \emph{differentially private Bayesian estimation} through either reasoning about the inherent privacy of the posterior distribution or privatizing off-the-shelf Bayesian methods, these works generally do not come with rigorous utility guarantees beyond low-dimensional settings. In fact, even for the prototypical tasks of Gaussian mean estimation and linear regression, it was unknown how close one could get to the Bayes-optimal error with a private algorithm, even in the simplest case where the unknown parameter comes from a Gaussian prior.
    
    In this work, we give the first polynomial-time algorithms for both of these problems that achieve mean-squared error $(1 + o(1))\mathrm{OPT}$ and additionally show that both tasks exhibit an intriguing computational-statistical gap. For Bayesian mean estimation, we prove that the excess risk achieved by our method is optimal among all efficient algorithms within the low-degree framework, yet is provably worse than what is achievable by an exponential-time algorithm. For linear regression, we prove a qualitatively similar such lower bound. Our algorithms draw upon the privacy-to-robustness framework introduced by~\cite{hopkins2023robustness}, but with the curious twist that to achieve private Bayes-optimal estimation, we need to design sum-of-squares-based robust estimators for \emph{inherently non-robust} objects like the empirical mean and OLS estimator. Along the way we also add to the sum-of-squares toolkit a new kind of constraint based on short-flat decompositions.
\end{abstract}

%% file: content/introduction.tex
\section{Introduction}

Given data $x$ generated by a process characterized by an unknown parameter $\theta$, how do we update our prior beliefs about $\theta$? This is the goal of Bayesian inference, and there is a rich algorithmic and statistical toolkit for doing so optimally by extracting as much information about $\theta$ as possible from the data $x$. But in many real-world use cases where this toolkit could be -- and is -- applied, there is an inherent tension between optimal inference and a second, orthogonal desideratum: \emph{privacy}. Bayesian methods are routinely deployed by epidemiologists to estimate the rate of spread of a pandemic~\cite{epinow2,flaxman2020estimating}, by the US Census Bureau to infer statistics within small sub-populations~\cite{fay1979estimates,chung2020bayesian}, and by geneticists to predict complex disease risk in individuals~\cite{vilhjalmsson2015modeling}. Is there a way to retain the statistical power of these techniques yet protect the privacy of the individuals whose data is being scrutinized? 

While \emph{differential privacy (DP)}~\cite{dwork2006calibrating} offers the ideal formalism for studying this question, the bulk of the work on DP statistics has focused on the \emph{frequentist} setting where one does not have access to prior information about $\theta$. An exception to this was the early work of Dimitrakakis et al.~\cite{dimitrakakis2017differential} and Wang et al.~\cite{wang2015privacy}, who put forth a tantalizing proposal. They posited that the object that Bayesian inference targets, namely the \emph{posterior distribution} $p(\theta \mid x)$, is fundamentally a \emph{probabilistic} object, so perhaps a sample from this posterior has enough noise to be differentially private ``for free.'' This was the basis for a productive line of work on \emph{One Posterior Sample (OPS) mechanisms}~\cite{zheng2015bayes,dimitrakakis2017differential,wang2015privacy,foulds2016theory,mir2013differential,bassily2014revisited,zhang2016differential,geumlek2017renyi} which sought to output a sample from some annealed version of the posterior, with the rationale that under (very strong) assumptions on the prior and log-likelihoods $\ln p(x\mid \theta)$, this was already differentially private.

Unfortunately, there are many natural settings for which the true posterior is not private, even for canonical testbeds like mean estimation and linear regression. While numerous efforts have been made to remedy this with alternative algorithms like computing a suitable noising of the sufficient statistics~\cite{foulds2016theory,bernstein2018differentially,bernstein2019differentially,ferrando2024private} or running DP analogues of standard Bayesian methods~\cite{wang2015privacy,foulds2016theory, heikkila2019differentially, zhang2023dp,jalko2017differentially,honkela2018efficient,park2020variational}, the utility guarantees for all of these methods are either unknown or only quantitatively meaningful in low dimensions.

Another serious limitation of the OPS paradigm is that in typical applications of Bayesian inference, having a \emph{single} sample from the posterior is not particularly useful. If further downstream, one wants to carry forward the posterior as the full summary of one's uncertainty up to that point, then what one really wants is to privately estimate a \emph{description} of the posterior, or alternatively, a \emph{posterior sampling oracle} that can subsequently be used to generate an \emph{arbitrary} number of draws from the posterior. This would be possible if, for instance, we could privately estimate the \emph{sufficient statistics} of the data.

We therefore ask:
\begin{center}
    \emph{Can we rigorously establish privacy-utility tradeoffs for Bayesian posterior sampling in high dimensions?}
\end{center}
In this work, we answer this in the affirmative for the fundamental tasks of \emph{Gaussian mean estimation} and \emph{linear regression with Gaussian design}. In a departure from prior methods in DP Bayesian inference, we leverage the powerful \emph{sum-of-squares (SoS) hierarchy},
which was recently used to obtain state-of-the-art theoretical guarantees for mean estimation, regression, and related problems in the \emph{frequentist} DP setting for the aforementioned tasks~\cite{private-sbm-and-mixtures,anderson2025sample,hopkins2023robustness}. While we will build upon these techniques, as we will see, the Bayesian setting introduces a host of subtle phenomena and technical hurdles not present in the frequentist setting.

\vspace{0.3em}

\noindent \textbf{Bayesian mean estimation.} In the most basic incarnation of this problem, the unknown parameter $\theta$ corresponds to an unknown mean in $\R^d$ that is sampled from some prior $\pi$, and one observes data $x$ in the form of i.i.d. samples $x_1,\ldots,x_n$ from $\mathcal{N}(\theta,\Id)$. Throughout our discussion of mean estimation, we will use the notation $\mu$ instead of $\theta$ to denote the mean vector. The goal is to design a (randomized) estimator $\hat{\mu}(x)$ such that the \emph{mean squared error (MSE)}
\begin{equation}
    \mathrm{MSE}(\hat{\mu}) \coloneqq \E_{\mu\sim\pi}\,\E_{x\sim \mathcal{N}(\mu,\Id)^{\otimes n}}\|\mu - \hat{\mu}(x)\|^2_2
\end{equation}
is as small as possible, while ensuring that $\hat{\mu}$ is differentially private in the sense that for worst-case neighboring datasets $x,x'$, the distributions over outputs $\hat{\mu}(x)$ and $\hat{\mu}(x')$ are close. Without the constraint of privacy, the \emph{Bayes-optimal estimator} that minimizes MSE is the \emph{posterior mean} $\E[\mu\mid x]$. Note that the prior $\pi$ has to be known a priori in order for the posterior mean (or even MSE) to be well-defined---in this paper we will consider Gaussian priors of known covariance $\Sigma$.

For this task, we give the first estimator that simultaneously preserves privacy and provably (in fact, near-optimally) competes with the Bayes-optimal estimator when the prior $\pi$ is Gaussian. First, by an elementary variance decomposition, the mean-squared error achieved by any estimator $\hat{\mu}$ can be written as the sum of the \emph{minimum} mean-squared error (MMSE), which is achieved by the $\E[\mu\mid x]$, plus the expected squared difference between $\hat{\mu}(x)$ and $\E[\mu\mid x]$ (\cref{fact:error-decomposition}). Our goal is thus to estimate $\E[\mu\mid x]$ as well as possible. Note also that because the prior is Gaussian, the posterior is also Gaussian with \emph{known covariance}\footnote{We use swap adjacency, so $n$ is not private}, so once we have an estimate of $\E[\mu\mid x]$, we have a statistically close description of the full posterior \emph{distribution} and can subsequently generate as many samples from it as we please without additional cost to privacy.

\begin{theorem}[Upper bounds for private Bayesian mean estimation]\torestate{\label{thm:mean_estimation_upper_informal}
    Given $\mu$ drawn from prior $\pi = \mathcal{N}(0,\Sigma)$, if the data $x$ consists of $n \ge \tilde{\Omega}(\frac{d}{\epsilon}\log(\sqrt{\Tr(\Sigma)}/\alpha))$  i.i.d. samples from $\mathcal{N}(\mu,\Id)$, there is
    \begin{enumerate}
        \item A computationally inefficient $\epsilon$-DP estimator $\hat{\mu}(x)$ such that with probability at least $1 - \beta$,
        \begin{equation}
            \norm{\hat{\mu}(x) - \E[\mu\mid x]}_2 \le \tilde{O}\biggl(\frac{\tr(\Lambda) + \norm{\Lambda}_{\sf op} \log(1/\beta)}{\epsilon n}\biggr) \coloneqq \tilde{O}(\alpha_{\rm stat})
        \end{equation}
        \item A \emph{polynomial-time} $\epsilon$-DP estimator $\hat{\mu}(x)$ such that with probability at least $1 - \beta$,
        \begin{equation}
            \norm{\hat{\mu}(x) - \E[\mu\mid x]}_2 \le \tilde{O}\biggl(\frac{\norm{\Lambda}_{4/3} + \norm{\Lambda}_{\sf op} \log^{3/4} (1/\beta)}{\epsilon^{1/2} n^{3/4}} + \alpha_{\rm stat}\biggr) \coloneqq \tilde{O}(\alpha_{\rm comp} + \alpha_{\rm stat})\,.
        \end{equation}
    \end{enumerate}
    where $\Lambda \coloneqq (\Id + \frac{1}{n}\Sigma^{-1})^{-1}$ and $\norm{\Lambda}_{4/3}$ denotes that Schatten $4/3$-norm of $\Lambda$. Here the $\tilde{O}$ conceals factors of $\log1/\alpha$, $\log\log1/\beta$, and $\log d$. When $\Sigma = \Id$, we do not lose $\log d$ factors.
    }
\end{theorem}

\noindent Note first that in contrast to prior frequentist work, there is no "non-private" term: taking $\varepsilon \to \infty$ will allow one to estimate the posterior mean arbitrarily well. This makes sense, because given non-private access to the data the posterior mean can be computed exactly.

Most of the intuition for this result is captured in the \emph{isotropic} setting where $\Sigma = \sigma^2 \Id$ for $\sigma \asymp 1/\sqrt{n}$, in which case
\begin{equation}
    \alpha_{\rm stat} = \frac{d + \log(1/\beta)}{\epsilon n} \qquad \text{and} \qquad \alpha_{\rm comp} = \frac{d^{3/4} + \log^{3/4}(1/\beta)}{\epsilon^{1/2} n^{3/4}}\,.
\end{equation}
In this case, one can compute that 
\begin{equation}
    \E[\mu\mid x] = \frac{1}{n + 1/\sigma^2}\sum^n_{i=1} x_i\,. \label{eq:posteriormean}
\end{equation}

\noindent In other words, the Bayes-optimal estimator is given by computing the \emph{empirical mean} of the data and \emph{shrinking} it to the origin by a factor that depends on the sample size and prior variance. In this setting, one can compare the error achieved by our estimators to the error achieved by the true mean $\mu$, scaled by $\frac{n}{n + 1/\sigma^2} = \Theta(1)$, which satisfies
\begin{equation}
    \Bigl\|\frac{n}{n+1/\sigma^2}\cdot \mu - \E[\mu\mid x]\Bigr\|_2 \le \sqrt{\frac{d + 2\log(1/\beta)}{n}} \label{eq:confint}
\end{equation}
with probability at least $1 - \delta$.
Note that in the regime where $\epsilon = \Theta(1)$, this bound dominates both $\alpha_{\rm stat}$ and $\alpha_{\rm comp}$. In fact, we prove \cref{thm:mean_estimation_upper_informal} by constructing private estimators for the empirical mean that succeed with high probability for \emph{any} $\mu$ that lies in some ball of radius $R$. Since our empirical mean estimator succeeds regardless of (and is agnostic to) the value of $\mu$, it can be applied even in a \emph{frequentist} setting to yield an estimator for $\mu$ with the following property:

\begin{theorem}[Frequentist mean estimation with optimal constants]\torestate{\label{thm:mean_estimation_upper_informal_frequentist}
    Let $\Lambda$ be psd. For any $\mu \in \mathbb R^d$ with $\|\mu\| \le R$, if the data $x$ consists of $n \ge \tilde{\Omega}\left(\frac{d}{\epsilon}\log\left(\frac{R+\|\Lambda\|_F}{\alpha}\right)\right)$  i.i.d. samples from $\mathcal{N}(\mu,\Lambda^2)$, there is
    \begin{enumerate}
        \item A computationally inefficient $\epsilon$-DP estimator $\mu(x)$ such that with probability at least $1 - \beta$,
        \begin{equation*}
            \norm{\mu(x) - \mu}_2 \le (1+o(1))\sqrt{\frac{\|\Lambda\|_F^2 + \|\Lambda\|_\mathsf{op}^2\log \frac{1}{\beta}}{n}} + \tilde{O}(\alpha_{\rm stat})
        \end{equation*}
        \item A \emph{polynomial-time} $\epsilon$-DP estimator $\mu(x)$ such that with probability at least $1 - \beta$,
        \begin{equation*}
            \norm{\mu(x) - \mu}_2 \le (1+o(1)) \sqrt{\frac{\|\Lambda\|_F^2 + \|\Lambda\|_\mathsf{op}^2 \log \frac{1}{\beta}}{n}} + \tilde{O}(\alpha_{\rm comp} + \alpha_{\rm stat})\,.
        \end{equation*}
    \end{enumerate}
    Here the $\tilde{O}$ conceals factors of $\log1/\alpha$, $\log\log1/\beta$, and $\log d$. When $\Sigma = \Id$, we do not lose $\log d$ factors. The $o(1)$ in the leading term holds in regimes where $\|\Lambda\|_F/\|\Lambda\|_\op \to \infty$ or $\beta \to 0$.
    }
\end{theorem}

\noindent  This matches the asymptotically optimal rate for mean estimation attained by \cite[Theorem 1.5]{hopkins2023robustness}, but with the optimal constant in front of the leading non-private term. It can thus be thought of as a generalization of~\cite[Theorem 1.2]{karwa-vadhan} to higher dimensions (proof in \Cref{appendix:frequentist-implications}). 

All of this highlights the first conceptual twist in our work: in the Bayesian setting, private mean estimation is about privately estimating the \emph{empirical mean} of the data, rather than the \emph{true mean}. Indeed, even though the former converges to the latter in the infinite-sample limit, the finite-sample error in Eq.~\eqref{eq:confint} is significant -- of the same order as the MMSE -- whereas our \cref{thm:mean_estimation_upper_informal} shows that it is possible to privately achieve the MMSE up to only lower-order excess error. Focusing on the empirical mean is especially counterintuitive as we will ultimately prove \cref{thm:mean_estimation_upper_informal} by using the machinery of~\cite{hopkins2023robustness} to reduce to an unconventional \emph{robust statistics} problem where we seek to estimate the empirical mean given a corrupted dataset. Whereas the empirical mean is a famously non-robust statistic, robustly estimating it is thus absolutely essential to achieving DP Bayesian estimation. Our algorithm for robust empirical mean estimation is based on a careful adaptation of a sum-of-squares analysis for robust parameter mean estimation due to~\cite{kothari2022polynomial}. We defer the formal definition to \cref{sec:definition} and guarantees for the robust problem to \cref{sec:gaussian-posterior-mean}.

One might wonder whether the extra term in the \emph{computational} rate in \cref{thm:mean_estimation_upper_informal} is an artifact of our techniques. The next surprise in our setting is that it is actually inherent: DP Bayesian mean estimation naturally exhibits a computational-statistical gap! This is perhaps unexpected given that the problem formulation is absent of the usual signatures for such gaps like sparsity and low-rank structure.

\begin{theorem}[Informal, see \cref{thm:comp-lower-bound_mean}]\label{thm:compstat_bayes}
   Under the low-degree heuristic~\cite{kunisky2019notes,wein2025computational}, no polynomial time estimator $\hat{\mu}$ can achieve $\norm{\hat{\mu}(x) - \E[\mu\mid x]}_2 \ll \alpha_{\rm comp} + \alpha_{\rm stat}$ when $\beta = \exp(-\Theta(d))$ and $\pi = \mathcal{N}(0,\sigma^2 \Id)$ for $\sigma \asymp 1/\sqrt{n}$.
\end{theorem}

\noindent We prove this by carefully reducing from a distinguishing problem that was recently introduced by~\cite{diakonikolas2025sos} and proven to be low-degree hard, to the problem of DP Bayesian mean estimation.

Stepping back from the isotropic setting, another appealing feature of \cref{thm:mean_estimation_upper_informal} is that the privacy cost reflected in $\alpha_{\rm stat}$ and $\alpha_{\rm comp}$ does not have an explicit dependence on the dimension (up to log factors) and instead depends on the prior covariance $\Sigma$: in other words, the prior informs which dimension our algorithm spends its privacy budget focusing on estimating. This is one of the key selling points of DP Bayesian estimation over DP frequentist estimation: \emph{prior information allows us to optimally allocate our privacy budget by judiciously applying different levels of noise to different directions of the data}. This is especially crucial in data-constrained settings where simply adding isotropic noise to preserve privacy can cripple utility.

Note that our setting differs from the literature on "covariance-aware" mean estimation \cite{covariance-aware-private-me} \cite{brown-hopkins-smith-affine-invariant} \cite{duchi-a-pretty-fast-algorithm} in two respects: first, in our setting the covariance is known \emph{a priori} (in a literal sense); and second, we care about $\ell_2$ rather than Mahalanobis error. Rather, our results have the closest parametric analog in \cite{PLAN-variance-aware} and \cite{dagan2024dimension}, which study private mean estimation to $\ell_2$ norm under anisotropic noise with known covariance.

Another selling point of the Bayesian perspective is that it naturally enables \emph{modular propagation of uncertainty}. Bayesian statistics is built upon post-processing estimates of the posterior distribution, and it is a standard fact that DP is immune to post-processing. In \cref{sec:online}, we illustrate this point with a simple example where we use our algorithm to privately estimate the mean of a distribution whose samples are progressively given to us in a streaming setting.

\vspace{0.3em}

\noindent \textbf{Bayesian linear regression.} Next, we turn to our guarantees for linear regression. Here, the unknown parameter $\theta$ corresponds to an unknown regressor in $\R^d$ that is again sampled from some prior. When discussing linear regression, we will always use the letter $w$ in place of $\theta$. Given $w$, one observes data $x$ now in the form of i.i.d. pairs $(x_1,y_1),\ldots,(x_n,y_n)$, where $x_i \sim \mathcal{N}(0,\Id)$ and $y_i = \langle w, x_i\rangle + \xi_i$, where $\xi_i \sim \mathcal{N}(0,\varsigma^2)$. The goal is again to minimize the MSE, defined as
\begin{equation}
    \mathrm{MSE}(\hat{w}) \coloneqq \E_{w\sim\pi} \E_{\{(x_i,y_i)\} \mid w} \|w - \hat{w}(\{x_i,y_i\}) \|^2\,,
\end{equation}
while ensuring that $\hat{w}$ is differentially private. Without the constraint of privacy, the optimal estimator is again the posterior mean $\hat{w}(x) = \E[w\mid \{(x_i,y_i)\}]$. Letting $X \in \R^{n\times d}$ denote the design matrix whose rows consist of $x_1,\ldots,x_n$, and letting $y\in\R^n$ denote the vector of labels $y_1,\ldots,y_n$, the posterior mean is given by
\begin{equation}
    w_{\rm post} \coloneqq \E[w\mid\{(x_i,y_i)\}] = (\sigma^{-2} \Id + XX^\top)^{-1} Xy\,. \label{eq:LR_MSE}
\end{equation}
For this task, we give the first estimator that simultaneously preserves privacy and provably competes with the Bayes-optimal estimator when the prior $\pi$ is isotropic Gaussian.

\begin{theorem}[Informal, see \cref{lem:exponential-bayesian-error-regression} and \cref{thm:bayesian-error-regression}]\torestate{\label{thm:regression_upper_informal}
Let $\epsilon \leq O(1)$.
    Given $w$ drawn from prior $\pi = \mathcal{N}(0,\sigma^2\Id)$, if the data consists of $n \ge \Omega(\frac{d}{\e}\log(\sigma\sqrt{d}))$ i.i.d. samples of the form $(x_i, y_i)$ where $x_i \sim \mathcal{N}(0,\Id)$ and $y_i = \langle w, x_i\rangle + \xi_i$ for $\xi_i \sim\mathcal{N}(0,1)$, there exists:
    \begin{itemize}
        \item A computationally inefficient $\epsilon$-DP estimator $\hat{w}(\{(x_i,y_i)\})$ such that with probability at least $1 - \beta$,
        \begin{equation}
            \|\hat{w}(\{(x_i,y_i)\}) - w_{\rm post}\|_2 \le \tilde{O}\biggl(\frac{d+\log(1/\beta)}{\e n}\biggr)
        \end{equation}
        \item A polynomial-time $\epsilon$-DP estimator $\hat{w}(\{(x_i,y_i)\})$ such that with probability at least $1 - \beta$,
        \begin{equation}
            \|\hat{w}(\{(x_i,y_i)\}) - w_{\rm post}\|_2 \le  \tilde{O}\biggl(\frac{\paren{d+\log(1/\beta)}^{3/4}}{\e^{1/2} n^{3/4}} + \frac{d+\log(1/\beta)}{\e n}\biggr)\,.
        \end{equation}
    \end{itemize}
    }
\end{theorem}

\noindent The same idiosyncrasies from the mean estimation setting apply here. After noting a connection between the posterior mean and the \emph{ordinary least-squares (OLS) estimator}, we will give an algorithm for private Bayesian linear regression by giving a robust estimator for an inherently non-robust object: the OLS estimator. As with the proof of \cref{thm:mean_estimation_upper_informal}, we go through the robustness-to-privacy framework of~\cite{hopkins2023robustness}, which necessitates designing an efficient algorithm for approximating the OLS estimator when an $\eta$-fraction of the data has been corrupted (see \cref{sec:approxOLS} for the formal definition and guarantees).

As with mean estimation, for linear regression our results also have an immediate application to frequentist statistics. 
Our private estimator provides an approximation for the least-squares estimator $w_{\mathrm{OLS}} \coloneqq (XX^\top)^{-1}Xy$, regardless of the distribution of $w$, and can thus be applied in a frequentist setting to yield the following guarantee:

\begin{theorem}[Frequentist linear regression with optimal constants]\label{thm:lr-freq}
For any $w \in \mathbb{R}^d$ with $\|w\| \le R$, if the data consists of
$n \ge \widetilde{\Omega}\left(\frac{d}{\epsilon}\log(R/\alpha)\right)$ i.i.d. samples of the form
$(x_1,y_1),\ldots,(x_n,y_n)$ where $x_i \sim \mathcal{N}(0,\Id)$ and
$y_i = \langle w,x_i\rangle + \xi_i$ for $\xi_i \sim \mathcal{N}(0,1)$, there is
\begin{enumerate}
\item
A computationally inefficient $\epsilon$-DP estimator
$\hat{w}(\{(x_i,y_i)\})$ such that with probability at least $1-\beta$,
\[
\|\hat{w}(\{(x_i,y_i)\}) - w\|_2
\le
(1+o(1))\sqrt{\frac{d + \log \frac{1}{\beta}}{n}}
+
\widetilde{O}\left(
\frac{d + \log(1/\beta)}{\epsilon n}
\right).
\]

\item
A polynomial-time $\epsilon$-DP estimator
$\hat{w}(\{(x_i,y_i)\})$ such that with probability at least $1-\beta$,
\[
\|\hat{w}(\{(x_i,y_i)\}) - w\|_2
\le
(1+o(1))\sqrt{\frac{d + \log \frac{1}{\beta}}{n}}
+
\widetilde{O}\left(
\frac{(d + \log(1/\beta))^{3/4}}{\epsilon^{1/2}n^{3/4}}
+
\frac{d + \log(1/\beta)}{\epsilon n}
\right).
\]
\end{enumerate}
Here the $\widetilde{O}$ conceals factors of $\log(1/\alpha)$, $\log\log(1/\beta)$, and $\log d$.
\end{theorem}

\noindent The result follows by combining the corresponding private approximation guarantee for
$w_{\mathrm{OLS}}$ with the standard high-probability bound
$\|w_{\mathrm{OLS}} - w\|_2
\le
(1+o(1))\sqrt{\frac{d + \log(1/\beta)}{n}}$.

Note that unlike existing SoS works on robust or private Gaussian linear regression, \cref{thm:regression_upper_informal} has linear rather than quadratic sample complexity in the dimension. This is made possibly by our relaxed choice of model: the existing SoS literature has focused on the unknown-covariance setting, where an SQ lower bound shows that $\Omega(d^2)$ samples are computationally necessary (see \cite[Remark 1.5]{anderson2025sample}. However, efficient non-SoS estimators were known to exist for the known isotropic covariance case \cite{lr-efficient-algorithms-lower-bounds}. We construct the first SoS robust estimator for Gaussian linear regression under isotropic covariance with sample complexity linear in the dimension, circumventing the previous barrier of certifying $4$-th moment hypercontractivity. 
As we outline in \cref{sec:overview}, our Bayesian linear regression result requires new ideas, such as constraining the SoS program to search over subsets of data for which the corresponding $y$ values admit a decomposition into a short vector and a flat vector.

Note that unlike existing SoS works on robust or private Gaussian linear regression, \cref{thm:regression_upper_informal} has linear rather than quadratic sample complexity in the dimension.
The existing SoS literature has focused on the unknown-covariance setting, for which there is evidence that $\Omega(d^2)$ samples might be necessary for any computationally efficient algorithm (see \cite[Remark 1.5]{anderson2025sample}). On the other hand, efficient non-SoS estimators were known to exist for the known isotropic covariance case \cite{lr-efficient-algorithms-lower-bounds}.
We construct the first SoS based estimator for Gaussian linear regression under isotropic covariance with sample complexity linear in the dimension.  
As we outline in \cref{sec:overview}, our Bayesian linear regression result requires new ideas, such as constraining the SoS program to search over subsets of data for which the corresponding $y$ values admit a decomposition into a short vector and a flat vector.

Finally, similar to the mean estimation case, we show that the gap between the exponential-time algorithm and the computationally efficient algorithm is real: there is again a computational-statistical gap that emerges once one insists on both privacy and Bayesian optimality:

\begin{theorem}[Informal, see \cref{thm:privacy-lower-bound-regression}]
    Under the low-degree heuristic~\cite{kunisky2019notes,wein2025computational}, no polynomial time estimator $\hat{w}$ can achieve $\|\hat{w}(\{(x_i,y_i)\}) - \E[w\mid \{(x_i,y_i)\}]\}\|_2 \ll \alpha$ in $O(\frac{d^{1/3}\log^{2/3}(1/\beta)}{\epsilon^{2/3}\alpha^{4/3}})$ samples.
\end{theorem}

\noindent When $\beta = \exp(-\Theta(n))$, this shows that the rate achieved by the computationally efficient algorithm in \cref{thm:regression_upper_informal} is optimal up to constant factors.

\subsection{Related work}
\label{sec:relatedwork}

\paragraph{Differential private Bayesian estimation.} The literature here is extensive and we defer to, e.g., \cite[Section 7.2.2]{sarwate2024machine} for a comprehensive overview. Results in this direction fall under one of the following general themes: (A) ``One Posterior Sample'' methods that reason about the inherent privacy of a single sample from the posterior or temperings thereof, (B) private implementations of standard Bayesian methods, and (C) algorithms based on reporting noisy sufficient statistics. Works under (A) include~\cite{mir2013differential,dimitrakakis2017differential,wang2015privacy,geumlek2017renyi,zhang2016differential,zheng2015bayes}, and a common proof strategy therein is to argue that the log-likelihood $\ln p(x \mid \theta)$ is Lipschitz in $x$ with respect to dataset distance\--- in our setting, such a condition simply does not hold. Works under (B) include privatizations of MCMC~\cite{wang2015privacy,foulds2016theory, heikkila2019differentially, zhang2023dp}, variational inference \cite{jalko2017differentially,honkela2018efficient,park2020variational}, and noisy Bayesian updates in graphical model inference. In addition to lacking concrete utility guarantees, another drawback of these methods is that they incur per-iterate privacy costs that would cripple naive approaches targeting pure DP. We also note the numerous works on \emph{convergent privacy} of Langevin Monte Carlo, see e.g.,~\cite{bassily2014revisited,chourasia2021differential,altschuler2024privacy,altschuler2022privacy,bok2024shifted}, which are motivated by optimization rather than inference. Works under (C) include \cite{foulds2016theory} and the works of\cite{bernstein2018differentially,bernstein2019differentially,ferrando2024private} on inference in exponential families and Bayesian linear regression which are based on noisy sufficient statistics but do not come with provable utility guarantees. These works are also closely related to those on \emph{Bayesian inference on privatized data}, where the privacy mechanism is folded into the likelihood $p(x\mid \theta)$, see e.g., \cite{xiao2012bayesian,williams2010probabilistic,ju2022data}.
Finally, we note that there are various works that connect Bayesian statistics to \emph{definitions} of privacy like \cite{kasiviswanathan2014semantics,triastcyn2020bayesian}, though these are unrelated to our work.

\paragraph{Sum-of-squares for robustness and privacy.} The literature on using sum-of-squares for high-dimensional statistical estimation is too vast to do justice here. 
Instead, we focus on the relevant works that use sum-of-squares to connect robust estimation to differential privacy. The first sum-of-squares work to exploit this connection was that of~\cite{hopkins2022efficient}, which was subsequently extended by~\cite{hopkins2023robustness} to give the formal framework that our work builds upon. 
Those works leveraged this connection to obtain the first computationally efficient algorithms for pure DP mean estimation with optimal sample complexity. 
Other works which have exploited this tie include~\cite{anderson2025sample,private-sbm-and-mixtures}. We also mention the numerous works on robust mean estimation~\cite{hopkins2018mixture,kothari2017outlier,kothari2022polynomial} and robust regression~\cite{klivans2018efficient,bakshi2021robust,chen2022online,anderson2025sample} using sum-of-squares. 

On the robust mean estimation side, earlier sum-of-squares works such as~\cite{hopkins2018mixture,kothari2017outlier}
developed general relaxations for robust moment estimation over certifiably subgaussian / hypercontractive families, but,
when specialized to Gaussians, the resulting analysis did not recover the optimal dependence on the corruption fraction.
The contribution of~\cite{kothari2022polynomial} is to close this gap: they show that for Gaussians, essentially the same
canonical relaxation already attains the optimal $\widetilde{O}(\eta)$ error in polynomial time, via a new analysis that can exploit Gaussian moment lower bounds without requiring matching sum-of-squares certificates. 
Our estimator for robust
\emph{empirical} mean estimation in \cref{sec:apply_kmz} is a careful adaptation of this analysis, and \cref{sec:isotropic-estimation-lemma} provides the key
lemma that we use.
On the robust regression side, the closest sum-of-squares works to ours are~\cite{klivans2018efficient,bakshi2021robust}.
The former gives guarantees under certifiable $4$-hypercontractivity, while the latter extends this to certifiable
$k$-hypercontractivity and, under independent or suitably negatively correlated noise, obtains excess prediction error of
order $\widetilde{O}(\eta^{2-2/k})$, which in the isotropic Gaussian setting corresponds to parameter error
$\widetilde{O}(\eta^{1-1/k})$. These guarantees are nevertheless not directly sufficient for our setting. First, we need to robustly approximate the data-dependent OLS / posterior mean
of the uncorrupted samples. 
Second, constant degree moment assumptions give suboptimal $\eta$-dependence for us: for
example, degree $4$ leads only to an $\eta^{3/4}$-type parameter guarantee, while taking larger $k$ to improve the
$\eta$-dependence worsens the dimension dependence, since one must certify correspondingly higher moments from samples.
In particular, already degree $4$ moment certification incurs at least quadratic sample complexity in the dimension for
Gaussian design. 

For a far broader overview of modern developments in robust statistics, we refer the reader to the survey of~\cite{Diakonikolas_Kane_2023}.

\paragraph{Robust inference.} While the literature on robust \emph{Bayesian} statistics is far sparser in comparison, we note the recent work of~\cite{hopkins2024adversarially} on robust Bayesian inference for the broadcasting on trees problem and the work of~\cite{chen2022kalman} on robust Kalman filtering. Some other works have considered Bayesian-adjacent problems in a robustness context, for instance the work of~\cite{Liu2022Minimax} on achieving minimax rates for robust community detection, and the works of~\cite{ivkov2024semidefinite,ivkov2025fast} on robust spin glass optimization.

\paragraph{Comparison of frequentist application to prior work} Previous work on private estimation was roughly divided into two camps. one camp \cite{what-can-we-learn-privately, karwa-vadhan, privately-learning-high-dimiensional-distributions, hopkins2023robustness} takes a learning-theory perspective: samples drawn from some distribution, and the estimator aims to learn some property of the distribution. The other camp \cite{dwork2006calibrating, blr-smalldb, geometry-of-dp} aims to privately estimate some data-dependent statistic of arbitrary samples. \Cref{thm:mean_estimation_upper_informal_frequentist} demonstrates that the stronger distributional assumptions of the former camp can be fruitfully applied even toward the objective of the latter: without any distributional assumption on the data the best we could do would be a Laplace mechanism to get $n = \Theta(\frac{d R}{\varepsilon n})$
which is much worse for small $\Lambda$ or large $R$ than what the distributional assumption of Gaussianity allows.

\paragraph{Previous computational lower bounds in private estimation} Previous works \cite{complexity-of-dp-data-release, pcps-synthetic-data-hardness} have used strong cryptographic assumptions to obtain lower bounds for differential privacy. However, encoding a cryptographic problem requires much more structure than we have here: these hardness results generally apply only for synthetic data generation. \cite{non-interactive-dp-from-owf} applies beyond synthetic data, but requires an exponentially large query class. In contrast, mean estimation under Gaussian noise is a very simple and concrete statistical task with very little structure.

\section{Overview of Techniques}
\label{sec:overview}

\subsection{Privacy and robustness for data-dependent statistics}

Our first conceptual contribution is simple, yet novel and perhaps widely applicable to many different problems in robust and private statistics. Traditionally, papers in our line of work have focused on parameter estimation: a parameter is chosen unknown to the algorithm, data is generated i.i.d. from a distribution indexed by the parameter, and the algorithm must estimate the parameter given the data. This is the most common setting in statistics, since it encapsulates the fundamental statistical task of guessing based on noisy measurements.

Instead of estimating a property of the model, we look to estimate merely a property of the samples. From a statistical estimation point of view, this makes the problem completely trivial. However, with the imposition of robustness or privacy, it becomes even more difficult than the parameter estimation setting: no longer able to hide analytic slack behind a dominant statistical term, we must characterize and confront the cost of robustness or privacy alone.

Our main observation is that many of the techniques developed for a parameter-estimation setting are readily adaptable to a data-estimation setting, where more careful analysis gives them tighter guarantees. In particular, the ``robustness implies privacy'' result of \cite{hopkins2023robustness} still holds even when the statistics are properties of the data rather than of the model, as long as our robust estimators are ``robustly'' estimating the data-dependent property of the uncorrupted samples rather than of the underlying parameter of the model. We develop this perspective in detail in \cref{sec:robust_to_privacy_mean} for mean estimation, and \cref{sec:robust_to_private_LR} and \cref{sec:privatize} for linear regression.

\subsection{Bayesian mean estimation}

\paragraph{Reduction to robust empirical mean estimation.}

In this overview we focus on the isotropic prior case as it contains the bulk of the technical intuition. In this setting, as discussed in the introduction, the posterior mean is just a rescaling of the \emph{empirical mean}. To privately estimate this object, we will draw upon the powerful connection between privacy and robust statistics developed by~\cite{hopkins2022efficient,hopkins2023robustness}. This machinery, which we carefully adapt in \cref{sec:robust_to_privacy_mean} to our setting, distills the essence of our problem to a rather idiosyncratic question in robust statistics: given an $\eta$-corruption of a Gaussian dataset with unknown parameter mean, estimate the \emph{empirical mean} of the original uncorrupted dataset. At first glance this might appear to be paradoxical, as the empirical mean is a horribly non-robust estimator of the parameter mean. The reason there is no contradiction is that we are interested in the empirical mean of the \emph{uncorrupted} data, which by concentration of measure is in any case close to the parameter mean.

The crucial point however is that we do not want to pay in this concentration error, as it can be of the same order as the Bayes-optimal error in the original Bayesian estimation problem! Instead, in our setting it is essential that we specifically target the empirical mean in order to obtain excess risk bounds that are of lower order relative to the Bayes-optimal error.

\paragraph{Algorithm for robust empirical mean estimation.}

It turns out that in order to achieve the guarantee in \cref{thm:regression_upper_informal} by plugging into the robustness-to-privacy machinery of~\cite{hopkins2023robustness}, we must robustly estimate the empirical mean to error $O(\eta\sqrt{\log 1/\eta} + \sqrt{\eta\sqrt{d/n}})$. This is consistent with rates from prior work on robust mean estimation which showed that the \emph{parameter mean} can be estimated to error $\tilde{O}(\eta)$ in $n = O(d/\eta^2)$ samples, but in our setting we want a bound which gracefully degrades as one goes below $O(d/\eta^2)$ samples, and in particular which still tends to zero as $\eta \to 0$.

Fortunately, for mean estimation, existing work on frequentist Gaussian mean estimation can be adapted to give such an error bound. In \cref{sec:apply_kmz}, we describe how to lift ideas from a recent work of~\cite{kothari2022polynomial} giving a polynomial-time algorithm within sum-of-squares for robustly estimating the mean of a Gaussian. At a high level, the sum-of-squares program searches for a large subset of the data with bounded covariance, rather than bounded higher moments. The strategy therein is to reproduce the resilience-based approach of~\cite{diakonikolas2019robust} within sum-of-squares while sidestepping the fact that resilience is likely not efficiently certifiable. We show how to carefully account for the error terms that appear in the analysis of~\cite{kothari2022polynomial} and deduce that the same analysis actually gives an algorithm for robust \emph{empirical} mean estimation. As the technical tools already appeared in prior work, we do not dwell upon them here, deferring a formal treatment to \cref{sec:apply_kmz} and \cref{sec:isotropic-estimation-lemma}.

\paragraph{Information-computation gap for mean estimation.}

Surprisingly, the rate achieved by our sum-of-squares algorithm for robust empirical mean estimation is worse than the statistically optimal rate. Indeed, in \cref{sec:ineff-priv-iso-emp} we give a computationally inefficient algorithm that robustly estimates the empirical mean to error $O(\eta\sqrt{\log 1/\eta} + \sqrt{\eta}\sqrt{d/n})$. Roughly, the latter algorithm brute force searches for a set of points $(y_1,\ldots,y_n)$ which disagree with the corrupted dataset on at most an $\eta$ fraction of points and which, crucially, satisfies the property that the average of any $\eta$ fraction of the $y_i$'s is sufficiently close to the empirical mean $\bar{y}$ of the $y_i$'s. That is, for all $S\subseteq[n]$ of size $\eta n$, we insist that $\frac{1}{\eta n}\sum_{i\in S}y_i$ is close to $\bar{y}$. The maximum distance across all such $S$ is the so-called \emph{maximum $\eta$-sparse singular value} of the matrix $M$ whose columns consist of $\frac{1}{\eta n}(y_i - \bar{y})$. Whereas the maximum singular value of $M$ is $\sup_{v\in\mathbb{S}^{n-1}} \norm{Mv}$, the maximum $\eta$-sparse singular value is defined to be $\sup_{v\in \mathbb{S}^{n-1}: \norm{v}_0 = \eta n} \norm{Mv}$.

A simple union bound shows that if we took the $y_i$'s to be the true, uncorrupted Gaussian samples, then the maximum $\eta$-sparse singular value is at most $O(\sqrt{\eta}\sqrt{d/n})$, which ultimately gives rise to the statistically optimal rate. Unfortunately, it is far from clear how to computationally efficiently search for such $y_i$'s with this property. This turns out to be for good reason: it was shown in~\cite{diakonikolas2025sos} that the maximum sparse singular value is computationally hard to certify optimally.

While robust empirical mean estimation is not formally equivalent to sparse singular value certification, we are nevertheless able to use a similar lower bound construction as the one in~\cite{diakonikolas2025sos} to argue that the achieving the optimal rate for former problem is also computationally hard. In \cref{sec:robustmeanhard}, we provide evidence under the \emph{low-degree framework} that any computationally efficient algorithm for estimating the empirical mean to error $\alpha$ under corruption fraction $\eta$ requires $n \ge \Omega\left(\frac{d \eta^2}{\alpha^4}\right)$ samples.

At a high level, the hardness result proceeds as follows. We construct an imbalanced mixture of two Gaussians, where one of the components has only $\eta$ mass (\cref{def:eta-gaussian-mixtures}), and the other component has mean shifted $\Theta(\alpha)$ away from the origin in an unknown, random direction. By standard calculations, we show this ensemble is hard for any low-degree algorithm to distinguish from standard Gaussian using $o(d\eta^2/\alpha^4)$ samples. Importantly, this implies that no such algorithm would be able to distinguish whether their data was uncorrupted, i.i.d draws from standard Gaussian, versus whether it was $\eta$-corrupted and sampled from a Gaussian with nonzero mean. We then show that this hard problem can be reduced to robust empirical mean estimation. Indeed, to tell apart the two ensembles, we can simply compare the distance between the empirical mean and the output of the robust estimator. This reduction turns out to be strong enough to show the desired tight information-computation tradeoff for robust empirical mean estimation.

Some additional work is required to deduce a lower bound in our original Bayesian setting as the prior over the mean parameter for which the above reduction shows hardness is a mixture between $0$ and the uniform distribution over vectors of length $\Theta(\alpha)$ over vectors; we defer the details to the proof of \cref{thm:comp-lower-bound_mean} at the end of \cref{sec:privatemeanhard}. The hardness for \emph{private} Bayesian mean estimation then follows from the privacy-to-robustness reduction of~\cite{Georgiev2022PrivacyIR}.

\subsection{Bayesian linear regression}

In Bayesian linear regression, the posterior mean that we wish to estimate is given by $w_{\rm post} = (\frac{1}{\sigma^2}\Id + XX^\top)^{-1} Xy$, which is closely related to its frequentist cousin, the ordinary least squares (OLS) estimator $w_{\rm OLS} = (XX^\top)^{-1}Xy$. Just as the bulk of the difficulty with obtaining private algorithms for Bayesian mean estimation boiled down to the unconventional question of robust estimation of a simple empirical statistic of the data (the empirical mean), the crux of the analysis here is similarly centered around robustly estimating $w_{\rm OLS}$. That is, given an $\eta$-corrupted dataset $(\tilde{X},\tilde{y})$, we wish to estimate $w_{\rm OLS}$ for the original uncorrupted dataset $(X,y)$ as well as possible. Furthermore, because $XX^\top\approx n\Id$ in operator norm, the essence of our argument is really centered around estimating the vector $\frac{1}{n}Xy$.

While countless incarnations of such questions have been studied in the robust statistics literature, to our knowledge none achieve the fine-grained bound that we are aiming for. Indeed, in order to get the claimed private rate, for the corresponding robust problem it is essential to estimate $w_{\rm OLS}$ in the presence of $\eta$-corrupted data to error $\eta\sqrt{\log 1/\eta} + \sqrt{\eta\sqrt{d/n}}$, in direct analogy with the rate we targeted for robust empirical mean estimation previously. Note that this significantly ties our hands to look for algorithms that work even when the number of samples is only \emph{linear} in dimension. While it is known how to achieve error $\tilde{O}(\eta)$ using $O(d)$ samples using non-SoS methods, even as early as the work of~\cite{diakonikolas2019efficient}, it is essential in our setting to use sum-of-squares in order to draw upon the robustness-to-privacy machinery of~\cite{hopkins2023robustness}. Unfortunately, all known sum-of-squares results on robust regression which achieve error $\tilde{O}(\eta)$ at least require certifying fourth moments of the covariates, which requires $O(d^2)$ sample complexity. This is the core technical challenge, and below we describe a novel workaround.

\paragraph{Computationally efficient algorithm using short-flat decompositions.}

Note that if $\sigma \ll 1/\sqrt{n}$, then the prior is sufficiently strong that the Bayes-optimal error can be achieved simply by ignoring the data and sampling from the prior. In contrast, when $\sigma \gg 1/\sqrt{n}$ (``large variance''), the prior is sufficiently weak that it can be safely ignored, and the private Bayesian problem reduces to private OLS. The large variance regime is actually the most technically involved part of our linear regression results, requiring solving a sequence of three sum-of-squares programs that build upon each other, and we defer a discussion of this to later in the paper (\cref{sec:weakprioroverview}).

In this overview, we focus on the ``proportional variance'' regime where $\sigma \asymp 1/\sqrt{n}$, in which we can ignore neither the prior nor the data and must use both in tandem to achieve the Bayes-optimal error rate.  As discussed above, the key challenge is to devise a \emph{robust} estimator for $\frac{1}{n} Xy$. 

For this, we will employ a sum-of-squares algorithm that builds on prior approaches with one essential twist. First, as is typical in this literature, we search for a $(1 - \eta)$-fraction of the data which satisfies structural features that the true uncorrupted dataset satisfies. In our setting, one of these structural features is standard, namely enforcing that the empirical covariance of the selected points is close to the population covariance. More precisely, we search for a $n\times d$ matrix $\hat{X}$ which agrees with the corrupted dataset's $\tilde{X}$ on a $(1-\eta)$-fraction of the rows and satisfies $\frac{1}{n}\hat{X}\hat{X}^\top \approx \Id$.

The key innovation in our work is that we additionally enforce a non-standard notion of \emph{resilience in the $y$ values} that we select. More precisely, we search for a vector $\hat{y}$ which agrees with the corrupted dataset's $\tilde{y}$ values on a $(1-\eta)$-fraction and which moreover admits an additive decomposition
\begin{equation*}
    \hat{y} = y' + y''
\end{equation*}
into:
\begin{itemize}
    \item A vector $y'$ with $L_2$ norm $O(\sqrt{\eta n\log(1/\eta)})$
    \item A vector $y''$ with $L_\infty$ norm $O(\sqrt{\log 1/\eta})$.
\end{itemize}
We first explain why we expect such a $\hat{y}$ and decomposition thereof to exist, and then why this is useful for our algorithm. 

If the $\hat{y}$ actually consisted of all labels from the \emph{uncorrupted} dataset, then each is an independent scalar Gaussian with $O(1)$ variance. By standard order statistics bounds, there is at most an $\eta$ fraction of ``outliers'' with magnitude exceeding $O(\sqrt{\log 1/\eta})$ (\cref{thm:gaussian-order-stat-sharp}). The $(1 - \eta)$-fraction of remaining entries can thus be taken to be $y''$. Furthermore, it is possible to show that the sum of squares of the $\eta$ fraction of outliers is bounded by $O(\eta n \log(1/\eta))$ (\cref{thm:maximum-subset-sum-gaussian-correct}). We note that the decomposition $\hat{y} = y' + y''$ is loosely related to the notion of a \emph{short-flat decomposition}, introduced by~\cite{kelner2023semi} in a very different context (for fast semi-random matrix completion), but to our knowledge our work is the first to leverage such a notion within the sum-of-squares framework.

Finally, let us see what this short-flat decomposition buys us. The ultimate goal is to argue that $\frac{1}{n}\hat{X}\hat{y}$, that is, the OLS estimator computed with the objects $\hat{X}$ and $\hat{y}$ that we have searched for via sum-of-squares, is close to the OLS estimator $\frac{1}{n}Xy$ that uses the true uncorrupted data. $\hat{X}$ and $\hat{y}$ were already chosen to agree with the true $X$ and $y$ on all but $O(\eta n)$ points, so the discrepancy between $\frac{1}{n}\hat{X}\hat{y}$ and $\frac{1}{n}Xy$ really only comes from the small set of $O(\eta n)$ data points which were either corrupted or which our sum-of-squares program identified as corrupted.

Denote this small set of ``bad points'' by $S$, and consider its contribution to $\frac{1}{n}\hat{X}\hat{y}$, namely 
\begin{equation}
    \frac{1}{n}\sum_{i\in S} \hat{y}_i \hat{X}_i\,.
\end{equation}
We want to show that this contribution is small, i.e., that its inner product with any unit direction $u$ is small. By Cauchy-Schwarz, this can be bounded via
\begin{equation}
    \biggl|\frac{1}{n}\sum_{i\in S} \hat{y}_i \langle \hat{X}_i, u\rangle\biggr| \le \biggl(\frac{1}{n}\sum_{i\in S} \hat{y}_i^2\biggr)^{1/2} \cdot \biggl(\frac{1}{n}\sum_{i\in S} \langle \hat{X}_i, u\rangle^2 \biggr)^{1/2}\,.
\end{equation}
The second term on the right-hand side can be bounded using standard tools (\cref{lem:sos-subset-sum}) recently introduced in~\cite{anderson2025sample}  which allow us to control the covariance of any $O(\eta n)$-sized subset of the $\hat{X}_i$'s. The main innovation in our proof is instead in how to handle the first term on the right-hand side. This is where we exploit the fact that $\hat{y}$ admits a short-flat decomposition. Indeed, we have
\begin{equation}
    \frac{1}{n}\sum_{i\in S} \hat{y}_i^2 = \frac{1}{n}\sum_{i\in S} (y'_i + y''_i)^2 \le \frac{2}{n}\sum_{i\in S} y'^2_i + \frac{2}{n}\sum_{i\in S} y''^2_i \le \frac{2}{n}\norm{y'}^2_2 + 2\eta \norm{y''}^2_\infty \le \tilde{O}(\eta)\,,
\end{equation}
which turns out to sufficiently small for our purposes.

\paragraph{Information-computation gap for linear regression.}

Just as with mean estimation, for linear regression there is a computationally inefficient algorithm which achieves strictly superior rate compared to our sum-of-squares estimator (\cref{sec:inefficientreg}). The algorithm proceeds similarly to the one for mean estimation, involving a brute force search for a certain matrix with small maximum sparse singular value (\cref{algo:completion-mean}).

In \cref{sec:private-lr-lb} we then give low-degree evidence that the gap between our polynomial-time
estimator and the statistically optimal one is inherent. The main challenge compared to the mean estimation setting is that it is far less clear how to design an appropriate pair of ensembles which are hard to distinguish from each other. In particular, while there is a natural choice for one of the ensembles --- simply take the dataset of $(x,y)$'s to be independent standard Gaussians --- it is unclear how to design the other ensemble so as to be (1) computationally indistinguishable from the former ensemble, and (2) behave differently than the former given an oracle for robust linear regression. For the former condition, similar to the lower bound instance in the mean estimation setting, we want the joint distribution over $(x,y)$'s to look indistinguishable from fully Gaussian except in a single hidden direction. It turns out that the correct choice is to take an imbalanced \emph{mixture of linear regressions}, where the regression vectors are scalings of an unknown, random direction (\cref{def:variance-matched-xcorruption-test}). 

By a more involved calculation than in the mean estimation setting, these two ensembles can be shown to be hard to distinguish by any low-degree estimator with  $n = o(d\eta^2/\alpha^{4})$ samples (\cref{lem:ldlr-xcorruption}). By a similar argument as in the mean estimation case, this hard distinguishing task can then be reduced to robust Bayesian linear regression (\cref{lem:reduction-xcorruption}), by comparing the output of the robust estimator to the output of the non-robust posterior mean. By the privacy-to-robustness reduction of~\cite{Georgiev2022PrivacyIR}, we then obtain the claimed hardness result for private Bayesian regression. (\cref{thm:privacy-lower-bound-regression}).

%% file: content/preliminaries.tex
\section{Preliminaries}

\Jnote{Put the definition of the $k$-sparse norms}
\subsection{Bayesian estimation under robustness and privacy}
\label{sec:definition}

We begin by reviewing the standard setting of (non-robust, non-private) Bayesian estimation. Given a parameter $\theta$ drawn from some prior distribution $\mathcal P$, and given $n$ samples drawn from i.i.d. from a distribution $\mathcal D_\theta$ parameterized by $\theta$, \emph{Bayes estimation} aims to output an estimate $\hat \theta$ with minimal \emph{Bayes risk}, defined as (for some loss function $\ell$)
\[
    \mathop{\mathbb E}_{\theta \sim \mathcal P} \mathop{\mathbb E}_{x_1,...,x_n \sim \mathcal D_\theta} [\ell(\hat \theta(x_1,...,x_n), \theta)]
\]
In the case of mean estimation (where $\mathcal D_\theta := \mathcal N(\theta, \Id_d)$, for example), it is typical to take the loss function to be the the squared-error
\[
    \ell_\text{mse} := \|\hat \theta - \theta\|_2^2\,,
\]
in which case the Bayes risk is called the \emph{mean-squared error}, denoted as $\mathrm{MSE}(\hat{\theta})$. Denoting by $\mathbb E[\theta | x]$ the posterior mean of $\theta$ given the samples, we see that 
\begin{align}%
    \mathrm{MSE}(\hat{\theta}) &= \mathop{\mathbb E}_{\theta \sim \mathcal P} \mathop{\mathbb E}_{x_1,...,x_n \sim \mathcal D_\theta} [\|\hat \theta - \mathbb E[\theta | x] + \mathbb E[\theta | x] - \theta\|_2^2] \nonumber\\
    &= \mathbb E [\norm{\hat{\theta}(x)-\E[\theta|x]}_2^2]+\mathbb E[\norm{\E[\theta|x]-\theta}^2] +2\mathbb E [(\hat{\theta}(x)-\E[\theta|x])(\E[\theta|x]-\theta)] \nonumber\\
    &= \mathbb E [\norm{\hat{\theta}(x)-\E[\theta|x]}_2^2]+\mathbb E[\norm{\E[\theta|x]-\theta}^2] \label{fact:error-decomposition}
\end{align}
where the last equality is because $\E[(\E[\theta|x]-\theta) \mid x]=0$; so, the MSE is minimized by letting $\hat \theta$ be the mean of the posterior distribution given the samples one has seen. One justification for the choice of MSE as a loss function is that
\[
    \mathop{\mathbb E}_{z \sim \mathcal D_\theta} \|\hat \theta - z \|_2^2 = \mathop{\mathbb E}_{z \sim \mathcal D_\theta} \|\hat \theta - \theta \|_2^2 + \text{Var}_\theta(z)
\]
so the estimator that minimizes MSE also minimizes expected squared-error on a fresh sample. The ``fresh-sample perspective'' on Bayes risk motivates our choice of loss in the linear regression setting: if $\mathcal D_\theta$ is $(x,y)$ where $x \sim \mathcal N(0, \Id)$ and $y \sim  \mathcal N(\theta^\top x,\varsigma^2)$, then we have
\begin{align*}
    \mathbb E_{(x,y) \sim \mathcal D_\theta} \|\hat \theta^\top X - y\|_2^2 &= \varsigma^2 + \mathbb E_{(x,y) \sim \mathcal D_\theta} \|\hat \theta^\top x - \theta^\top x\|_2^2 \\
    &= \varsigma^2 + \mathbb E_{(x,y) \sim \mathcal D_\theta} \|\hat \theta - \theta\|_2^2
\end{align*}
and so if we would like to have a linear model that is near-optimally predictive given fresh data, we care about estimating the regression vector to $\ell_2$-norm.

Note importantly that these accuracy guarantees are average-case in nature. In this work however, we will consider other constraints like privacy and robustness; the former is entirely a worst-case guarantee, while the latter is a hybrid between average- and worst-case. We formally define these desiderata next.

\paragraph{Differential privacy.}

We will work in the standard setting of ``pure'' $\varepsilon$-DP originally proposed by \cite{dmns06}. 
\begin{definition}[Pure differential privacy]
    We call two possible datasets $X, X' \in \R^{n\times m}$ \emph{adjacent} if $X'$ can be obtained by changing a single row in $X$. 
    A randomized mechanism $\mathcal M$ is said to be \emph{$\varepsilon$-differentially private} if, for all adjacent $X, X'$ and all measurable events $E$:
    \[
        \mathbb P[\mathcal M(X) \in E] \le e^\varepsilon  \cdot \mathbb P[\mathcal M(X) \in E]\,.
    \]
\end{definition}

\noindent In the context of private statistical estimation, the mechanism $\mathcal{M}$ is the estimator $\hat{\theta}(x_1,\ldots,x_n)$, and the dataset $X$ has rows consisting of the samples $x_1,\ldots,x_n$. Importantly, privacy is a constraint that the learning algorithm must respect for \emph{all} datasets $X$, not just ``in-distribution'' datasets sampled from $\mathcal{D}_\theta$.

\paragraph{Robustness to adaptive corruptions.}

Whenever we say robustness in this paper, we will mean robustness in the \emph{strong contamination model}~\cite{diakonikolas2019robust,lai2016agnostic}:

\begin{definition}[Strong contamination model]
    Let $0 \le \eta < 1/2$. Let $\mathcal{D}_{\theta}$ be some parametrized distribution, and let $x_1,\ldots,\tilde{x}_n$ denote i.i.d. draws from $\mathcal{D}_\theta$. In the \emph{strong contamination model}, an adversary sees all of these samples and can select any $\eta n$ of them and replace with arbitrary points. Let $(\tilde{x}_1,\ldots,\tilde{x}_n)$ denote the corrupted dataset.
\end{definition}

\noindent Previous works that studied this model focused on estimating the parameter $\theta$ defining the distribution generating the uncorrupted data, given access to the corrupted data. In this work, we instead consider estimating a \emph{data-dependent} statistic of the samples (namely the posterior mean).

\begin{definition}[Robust estimation of data dependent statistics]
    Let $0 \le \eta < 1/2$. Let $\mathcal{D}_\theta$ be some parametrized distribution over $\R^d$; let $(x_1,\ldots,x_n)$ denote a dataset of i.i.d. draws from $\mathcal{D}_\theta$, and let $(\tilde{x}_1,\ldots,\tilde{x}_n)$ denote an $\eta$-corruption thereof under the strong contamination model. Let $f: (\R^d)^n\to \R^m$ denote some data-dependent statistic. We say that a learning algorithm $\mathcal{M}$ robustly estimates $f$ to error $\alpha$ with probability $1 - \beta$ if $\Pr[\norm{\mathcal{M}(\tilde{x}_1,\ldots,\tilde{x}_n) - f(x_1,\ldots,x_n)}_2 \le \alpha] \ge 1 - \beta$.
\end{definition}

\noindent Note that robustness is an average-case notion in the sense that the uncorrupted dataset $(x_1,\ldots,x_n)$ is assumed to come from the true distribution $\mathcal{D}_\theta$, but also a worst-case notion in the sense that the corruptions to the $\eta$ fraction of the dataset are arbitrary.

\subsection{Low-degree hardness}

In this work we will prove computational hardness results under the \emph{low-degree framework} introduced in~\cite{barak2019nearly,hopkins2017efficient,hopkins2017power}. Here we provide a brief overview of this framework, deferring a thorough treatment to the surveys of~\cite{kunisky2019notes,wein2025computational}. Roughly, this framework predicts that for ``well-behaved'' statistical estimation tasks, computationally hardness coincides with the failure of estimators that can be implemented using low-degree polynomials. For hypothesis testing problems, this is formalized as follows:

\begin{definition}[Degree-$D$ Advantage]\label{def:ldadv}
    Let $(\mathcal{P}_N)_{N\in\mathbb{N}}$ and $(\mathcal{Q}_N)_{N\in\mathbb{N}}$ be two sequences of distributions. Consider the sequence of hypothesis testing problems of distinguishing between whether a dataset $X$ was sampled from $\mathcal P_N$ or $\mathcal Q_N$. For a given $N$, the \emph{degree-$D$ advantage} for this problem is the quantity
    \[
        \mathrm{Adv}_{\le D, N} := \sup_{f \in \mathbb R[X]_{\le D}} \frac{\mathbb E_{X \sim \mathcal P_N}[f(X)]}{\sqrt{\mathbb E_{X \sim \mathcal Q_N}[f(X)^2]}}\,,
    \]
    where $\mathbb{R}[X]_{\le D}$ denotes the ring of polynomials in $X$ of total degree at most $D$. A standard calculation (see \cite[Proposition 1.15]{kunisky2019notes}) shows that if $L_{\le D, N}$ denotes the \emph{low-degree likelihood ratio} given by projecting the likelihood ratio $L_N \triangleq \frac{\mathrm{d}\mathcal{P}_N}{\mathrm{d}\mathcal{Q}_N}$ to the space of degree-$D$ polynomials with respect to $L_2(\mathcal{Q}_N)$, then
    \begin{equation}
        \mathrm{Adv}_{\le D, N} = \norm{L_{\le D, N}}_{L_2(\mathcal{Q}_N)}\,.
    \end{equation}
    We say that this sequence of hypothesis testing problems is \emph{low-degree hard} if for $D = \omega(\log N)$,
    \[
        \text{Adv}_{\le D, N} = O_N(1)
    \]
    When $N$ is clear from context, we omit it from the subscripts above. 
\end{definition}

\noindent While there is currently no formal conjecture that low-degree hardness implies computational hardness under a concrete set of conditions on the distributions $\mathcal{P}_N, \mathcal{Q}_N$,\footnote{See  \cite{Holmgren2020CounterexamplesTT,buhai2025quasi} for counterexamples to prior attempts at formulating such a conjecture.} it has been remarkably predictive of computational thresholds for average-case problems lacking algebraic structure and is generally consistent with a host of complementary forms of evidence for hardness (e.g., statistical query lower bounds~\cite{brennan2020statistical}, landscape analysis~\cite{bandeira2022franz}, sum-of-squares lower bounds~\cite{barak2014sum,hopkins2017power}, failure of message passing~\cite{montanari2025equivalence}). In cases where the low-degree framework incorrectly predicts computational hardness, there is generally algebraic~\cite{Holmgren2020CounterexamplesTT,buhai2025quasi} or lattice~\cite{diakonikolas2022non,zadik2022lattice} structure in the problem; given that this structure does not appear in the problems we consider, we believe low-degree hardness is a suitable proxy for computational hardness in our setting.

More precisely, we will establish computational hardness of private Bayesian estimation by exhibiting computationally efficient (but not necessarily low-degree) \emph{reductions} from low-degree hard problems to private Bayesian estimation.

Below, we collect some standard tools for proving low-degree lower bounds.

\input{content/prelims-ldlr}

\input{content/prelims-sos}

\subsection{Connections between robustness and privacy}

We will use the robustness to privacy framework \cite{hopkins2023robustness} to transform robust algorithms into private ones. We will use both the inefficient and efficient implementations of this reduction. 

\begin{lemma}[Inefficient Robustness to Privacy (Lemmas~2.1 and 2.2 in \cite{hopkins2023robustness})]
\label{lemma:inefficient-reduction-meta}
Let $\eta_0 < \eta^\star \in [0,1]$ be such that $\eta^\star n$ is a sufficiently large constant.
For every $\varepsilon,\delta>0$, there is an $(O(\varepsilon),\,O(e^{2\varepsilon}\delta))$-DP mechanism which,
for any $\theta^\ast \in \mathbb R^d$, takes $x_1,\ldots,x_n \sim p_{\theta^\ast}$ and, with probability $1-\beta$, outputs
$\theta$ such that $\|\theta-\theta^\ast\|\le 2\alpha(\eta_0)$, if
\[
n \;\ge\; O\Biggl(
\max_{\eta_0 \le \eta \le \eta^*}
\frac{ d \cdot \log\!\Big(\frac{2\alpha(\eta)}{\alpha(\eta_0)}\Big)
      + \log(1/\beta) + \log(\eta n)}{\eta\varepsilon}
\;+\;
\frac{\log(1/\delta)}{\eta^*\varepsilon}
\Biggr) \mper
\]
Additionally, there is a pure $\eps$-DP mechanism obtaining the same guarantees if 

\[
n \;\ge\; O\Biggl(\max_{\eta_0 \le \eta \le 1}
\frac{\,d\cdot \log\!\Big(\frac{2\alpha(\eta)}{\alpha(\eta_0)}\Big) + \log(1/\beta)}
{\eta\varepsilon}\Biggr) \mper
\]
\end{lemma}

\begin{theorem}[Efficient Robustness to Pure Differential Privacy (Theorem~4.1 in \cite{hopkins2023robustness})]
\label{thm:efficient-reduction-meta-pure}
Let \(0<\eta, r<1<R\) be fixed parameters. Suppose we have a score
function \(S(\theta,\mathcal{Y})\in[0,n]\) that takes as input a dataset
\(\mathcal{Y}=\{y_1,\ldots,y_n\}\) and a parameter
\(\theta\in\Theta\subset \mathbb{B}(\mathbb{R})^{d}\) (where \(\Theta\) is
convex and contained in a ball of radius \(R\)), with the following
properties:
\begin{itemize}
  \item \textbf{(Bounded Sensitivity)} For any two adjacent datasets
  \(\mathcal{Y},\mathcal{Y}'\) and any \(\theta\in\Theta\),
  \(\lvert S(\theta,\mathcal{Y})-S(\theta,\mathcal{Y}')\rvert \le 1\).

  \item \textbf{(Quasi-Convexity)} For any fixed dataset \(\mathcal{Y}\),
  any \(\theta,\theta'\in\Theta\), and any \(0\le \lambda\le 1\), we have
  \[
    S\big(\lambda\theta+(1-\lambda)\theta',\mathcal{Y}\big)
    \le \max\!\big(S(\theta,\mathcal{Y}),\,S(\theta',\mathcal{Y})\big).
  \]

  \item \textbf{(Efficiently Computable)} For any given \(\theta\in\Theta\)
  and dataset \(\mathcal{Y}\), we can compute \(S(\theta,\mathcal{Y})\) up to
  error \(\gamma\) in
  \(\operatorname{poly}\!\big(n,d,\log \tfrac{R}{r},\log \gamma^{-1}\big)\)
  time for any \(\gamma>0\).

  \item \textbf{(Robust algorithm finds low-scoring point)} For a given
  dataset \(\mathcal{Y}\), let
  \(T=\min_{\theta_0\in\Theta} S(\theta_0,\mathcal{Y})\).
  Then we can find some point \(\theta\) such that for all \(\theta'\)
  within distance \(r\) of \(\theta\), \(S(\theta',\mathcal{Y})\le T+1\), in
  time \(\operatorname{poly}\!\big(n,d,\log \tfrac{R}{r}\big)\).

  \item \textbf{(Volume)} For any given dataset \(\mathcal{Y}\) and
  \(\eta'\ge \eta\), let \(V_{\eta'}(\mathcal{Y})\) represent the
  \(d\)-dimensional Lebesgue volume of points \(\theta\in\Theta\subset\mathbb{R}^d\)
  with score at most \(\eta' n\). (Note that \(V_{1}(\mathcal{Y})\) is the
  full volume of \(\Theta\).)
\end{itemize}
Then, we have a pure \(\varepsilon\)-DP algorithm \(\mathcal{M}\) on datasets
of size \(n\), that runs in
\(\operatorname{poly}\!\big(n,d,\log \tfrac{R}{r}\big)\) time, with the
following property. For any dataset \(\mathcal{Y}\), if there exists
\(\theta\) with \(S(\theta,\mathcal{Y})\le \eta n\) and if
\[
  n \ge
  \Omega\Biggl(
  \max_{\eta':\,\eta\le \eta'\le 1}
  \frac{
    \log\bigl(\tfrac{V_{\eta'}(\mathcal{Y})}{V_{\eta}(\mathcal{Y})}\bigr)
    + \log\bigl(\tfrac{1}{\beta\cdot\eta}\bigr)
  }{\varepsilon\cdot \eta'}
  \Biggr)\,,
\]
then \(\mathcal{M}(\mathcal{Y})\) outputs some \(\theta\in\Theta\) of score
at most \(2\eta n\) with probability \(1-\beta\).
\end{theorem}

\begin{remark}[Applying Robustness to Privacy to Bayesian Estimation]
We note that the robustness to privacy lemmas and theorems in \cite{hopkins2023robustness} are stated for fixed $\theta^*$ that is independent from the data. However, upon close inspection, the same lemmas and theorems hold for estimation of any fixed functions of the input variables as well. More precisely, let $f$ denote any fixed function and $\set{x_1^*, \ldots, x_n^*}$ denote the uncorrupted samples. Then the same lemmas and theorems hold if $\theta^*$ is replaced with any fixed function of the input variables $f(x_1^*, \ldots, x_n^*)$.
\end{remark}

\noindent For proving computational lower bounds for differential privacy, we also need the following theorem reduction in the other direction, from privacy to robustness.

\begin{theorem}[Privacy to Robustness Reduction, Theorem 3.1 in \cite{Georgiev2022PrivacyIR}]\label{thm:privacy-robustness-reduction}
Let $\cM : \cX^* \to \cO$ be an $\epsilon$-private map from datasets $\cX^*$ to outputs $\cO$.
For every dataset $X_1,\ldots,X_n$, let $G_{X_1,\ldots,X_n} \subseteq \cO$ be a set of good outputs.
Suppose that $\cM(X_1,\ldots,X_n) \in G_{X_1,\ldots,X_n}$ with probability at least $1-\beta$ for some
$\beta = \beta(n)$. Then, for every $n \in \mathbb{N}$, on $n$-element datasets $\cM$ is robust to adversarial
corruption of any $\eta(n)$-fraction of inputs, where
\[
  \eta(n)
  = O\biggl(  \frac{\log (1/\beta)}{\e n}
  \biggr),
\]
meaning that for every $X_1,\ldots,X_n$ and $X_1',\ldots,X_n'$ differing on only $\eta n$ elements,
$\cM(X_1',\ldots,X_n') \in G_{X_1,\ldots,X_n}$ with probability at least $1 - \beta^{\Omega(1)}$.
\end{theorem}

\subsection{Sparse spectral norm of matrices}
For the analysis of algorithm for Bayesian linear regression, we use the sparse spectral norm of matrices.
\begin{definition}[Sparse spectral norm of matrices]\label{def:sparse-spectral-norm}
Let $X\in\R^{d\times n}$.
For an index set $S\subseteq[n]$, let $X_S\in\R^{d\times |S|}$ denote the submatrix formed by the
columns indexed by $S$.  For $k\in[n]$ define the $k$-sparse spectral norm
\[
\|X\|_{\mathrm{sp}}(k)\;\defeq\;\max_{\substack{S\subseteq[n]\\ |S|\le k}}\ \|X_S\|_{\op}\,.
\]
\end{definition}

%% file: content/prelims-ldlr.tex
\begin{definition}[Non-Gaussian component analysis]
    Let $A$ be a distribution over $\mathbb{R}$, and let $d(N)$ and $n(N)$ be increasing functions of the parameter $N$ corresponding to dimension and sample complexity. \emph{Non-Gaussian component analysis (NGCA)} is the problem of distinguishing between samples from $\mathcal{P}_N$ and $\mathcal{Q}_N$ given as follows. Define $\mathcal{Q}_N = \mathcal{N}(0,\Id_{d(N)})$. To sample from $\mathcal{P}_N$, a unit vector $v$ is sampled uniformly from $\mathbb{S}^{d(N)-1}$, and then $y_1,\ldots,y_{n(N)}$ are sampled i.i.d. from the product distribution whose orthogonal projection onto $v$ is $A$, and whose orthogonal projection onto the orthogonal complement of $v$ is $\mathcal{N}(0,\Id_{d(N)-1})$.  
\end{definition}

\begin{lemma}[Theorem 7.2 in \cite{diakonikolas2025sos}, Degree-$D$ Advantage and Growth of Hermite Coefficients; implicit in~\cite{mao2021optimal}]\label{lem:general-ldlr}
Suppose the univariate distribution $A$ satisfies the following conditions for an even
$j_\star\in\mathbb{N}$ and parameters $\kappa,\tau\in\mathbb{R}_+$:
\begin{itemize}
    \item $\E_{X\sim A}\, h_j(X)=0 \quad \text{for } j\in [j_\star-1]$.
    \item For $j\ge j_\star$, \; $\bigl|\E_{X\sim A} h_j(X)\bigr|^2 \le j^{\,j}\,\kappa\,\tau^{-\,{j}/{j_\star}}$.
\end{itemize}
Then if $n\kappa\ge 1$ and
\[
  n \;\le\; \frac{1}{\poly(D)}\,\frac{d^{\,j_\star/2}\tau}{\kappa}\,,
\]
the degree-$D$ advantage of the NGCA problem with the hidden distribution $A$
satisfies $\Adv_{\le D,n}\le O(1)$.
\end{lemma}

\begin{fact}[Fact 3.10 in \cite{diakonikolas2025sos}, Properties of Hermite Polynomials]\label{fact:hermite-props}
For the normalized probabilist's Hermite polynomials $h_j(X)$, we have the following properties:
\begin{itemize}
\item For any $x$, 
\(\E_{X\sim \mathcal N(x,1)}[h_j(X)] \;=\; x^j/\sqrt{j!}\).
\item For $\sigma^2\le 1$, let $\rho^2 := 1-\sigma^2$. Then
\(
\E_{X\sim \mathcal N(x,\sigma^2)}[h_j(X)] \;=\; \rho^{j}\, h_j(x/\rho).
\)

\item There exists a universal constant $C$ such that for any $|x|^2\le 2i$,
\(
h_i(x)^2 \;\le\; C e^{x^2/2}/\sqrt{i} .
\)
In particular, for all $|x|\le 1$,
\(
h_i(x)^2 \;\le\; 3C/\sqrt{i}
\). 
\end{itemize}
\end{fact}

%% file: content/prelims-sos.tex
\subsection{Sum-of-squares background}
In this paper, we use the sum-of-squares (SoS) semidefinite programming hierarchy 
\cite{barak2014sum,raghavendra2018high} 
for both algorithm design and analysis.
The sum-of-squares proofs-to-algorithms framework has been proven useful in many optimal or state-of-the-art results in algorithmic statistics.
We provide here a brief introduction to pseudo-distributions, sum-of-squares proofs, and sum-of-squares algorithms.

\subsubsection{Sum-of-squares proofs and algorithms}
\paragraph{Pseudo-distribution.} 

We can represent a finitely supported probability distribution over $\R^n$ by its probability mass function $\mu\from \R^n \to \R$ such that $\mu \geq 0$ and $\sum_{z\in\supp(\mu)} \mu(z) = 1$.
We define pseudo-distributions as generalizations of such probability mass distributions by relaxing the constraint $\mu\ge 0$ to only require that $\mu$ passes certain low-degree non-negativity tests.

\begin{definition}[Pseudo-distribution]
  \label{def:pseudo-distribution}
  A \emph{level-$\ell$ pseudo-distribution} $\mu$ over $\R^n$ is a finitely supported function $\mu:\R^n \rightarrow \R$ such that $\sum_{z\in\supp(\mu)} \mu(z) = 1$ and $\sum_{z\in\supp(\mu)} \mu(z)f(z)^2 \geq 0$ for every polynomial $f$ of degree at most $\ell/2$.
\end{definition}

\noindent We can define the expectation of a pseudo-distribution in the same way as the expectation of a finitely supported probability distribution.

\begin{definition}[Pseudo-expectation]
  Given a pseudo-distribution $\mu$ over $\R^n$, we define the \emph{pseudo-expectation} of a function $f:\R^n\to\R$ by
  \begin{equation}
    \tE_\mu f \seteq \sum_{z\in\supp(\mu)} \mu(z) f(z) \,.
  \end{equation}
\end{definition}

\noindent The following definition formalizes what it means for a pseudo-distribution to satisfy a system of polynomial constraints.

\begin{definition}[Constrained pseudo-distributions]
  Let $\mu:\R^n\to\R$ be a level-$\ell$ pseudo-distribution over $\R^n$.
  Let $\cA = \{f_1\ge 0, \ldots, f_m\ge 0\}$ be a system of polynomial constraints.
  We say that \emph{$\mu$ satisfies $\cA$} at level $r$, denoted by $\mu \sdtstile{r}{} \cA$, if for every multiset $S\subseteq[m]$ and every sum-of-squares polynomial $h$ such that $\deg(h)+\sum_{i\in S}\max\{\deg(f_i),r\} \leq \ell$,
  \begin{equation}
    \label{eq:constrained-pseudo-distribution}
    \tE_{\mu} h \cdot \prod_{i\in S}f_i \ge 0 \,.
  \end{equation}
  We say $\mu$ satisfies $\cA$ and write $\mu \sdtstile{}{} \cA$ (without further specifying the degree) if $\mu \sdtstile{0}{} \cA$.
\end{definition}

\noindent We remark that if $\mu$ is an actual finitely supported probability distribution, then we have  $\mu\sdtstile{}{}\cA$ if and only if $\mu$ is supported on solutions to $\cA$.

\paragraph{Sum-of-squares proof.} 

We introduce sum-of-squares proofs as the dual objects of pseudo-distributions, which can be used to reason about properties of pseudo-distributions.
We say a polynomial $p$ is a sum-of-squares polynomial if there exist polynomials $(q_i)$ such that $p = \sum_i q_i^2$.

\begin{definition}[Sum-of-squares proof]
  \label{def:sos-proof}
  A \emph{sum-of-squares proof} that a system of polynomial constraints $\cA = \{f_1\ge 0, \ldots, f_m\ge 0\}$ implies $q\ge0$ consists of sum-of-squares polynomials $(p_S)_{S\subseteq[m]}$ such that\footnote{Here we follow the convention that $\prod_{i\in S}f_i=1$ for $S=\emptyset$.}
  \[
    q = \sum_{\text{multiset } S\subseteq[m]} p_S \cdot \prod_{i\in S} f_i \,.
  \]
  If such a proof exists, we say that \(\cA\) \emph{(sos-)proves} \(q\ge 0\) within degree \(\ell\), denoted by $\mathcal{A}\sststile{\ell}{} q\geq 0$.
  In order to clarify the variables quantified by the proof, we often write \(\cA(z)\sststile{\ell}{z} q(z)\geq 0\).
\end{definition}

\noindent The following lemma shows that sum-of-squares proofs allow us to deduce properties of pseudo-distributions that satisfy some constraints.
\begin{lemma}
  \label[lemma]{lem:sos-soundness}
  Let $\mu$ be a pseudo-distribution, and let $\cA,\cB$ be systems of polynomial constraints.
  Suppose there exists a sum-of-squares proof $\cA \sststile{r'}{} \cB$.
  If $\mu \sdtstile{r}{} \cA$, then $\mu \sdtstile{r\cdot r' + r'}{} \cB$.
\end{lemma}

\paragraph{Sum-of-squares algorithm.}

Given a system of polynomial constraints, the \emph{sum-of-squares algorithm} uses semidefinite programming to search through the space of pseudo-distributions that satisfy this polynomial system.

Since semidefinite programming can only be solved approximately, we can only find pseudo-distributions that approximately satisfy a given polynomial system.
We say that a level-$\ell$ pseudo-distribution \emph{approximately satisfies} a polynomial system, if the inequalities in \cref{eq:constrained-pseudo-distribution} are satisfied up to an additive error of $2^{-n^\ell}\cdot \norm{h}\cdot\prod_{i\in S}\norm{f_i}$, where $\norm{\cdot}$ denotes the Euclidean norm\footnote{The choice of norm is not important here because the factor $2^{-n^\ell}$ swamps the effects of choosing another norm.} of the coefficients of a polynomial in the monomial basis.

\begin{theorem}[Sum-of-squares algorithm]
  \label{theorem:SOS_algorithm}
  There is an $(n+ m)^{O(\ell)} $-time algorithm that, given any explicitly bounded\footnote{A system of polynomial constraints is \emph{explicitly bounded} if it contains a constraint of the form $\|z\|^2 \leq M$.} and satisfiable system\footnote{Here we assume that the bit complexity of the constraints in $\cA$ is $(n+m)^{O(1)}$.} $\cA$ of $m$ polynomial constraints in $n$ variables, outputs a level-$\ell$ pseudo-distribution that satisfies $\cA$ approximately.
\end{theorem}

\begin{remark}[Approximation error and bit complexity]
  \label{remark:sos-numerical-issue}  
  For a pseudo-distribution that only approximately satisfies a polynomial system, we can still use sum-of-squares proofs to reason about it in the same way as \cref{lem:sos-soundness}.
  In order for approximation errors not to amplify throughout reasoning, we need to ensure that the bit complexity of the coefficients in the sum-of-squares proof are polynomially bounded.  
\end{remark}

\subsubsection{Useful sum-of-squares inequalities}

In this part, we provide some basic SoS proofs that are useful in our paper.

\begin{lemma}
\label[lemma]{lem:preliminary-SOS-abs-to-square}
Given constant $C \geq 0$, we have
    \begin{equation*}
         \Bigset{-C \leq z \leq C} \sststile{2}{z} z^2 \leq C^2  \,.
    \end{equation*}
\end{lemma}

\begin{lemma}[SoS selector lemma] \label{lem:sos-subset-sum}
    Let $a_1, \dots, a_n \in \R$ and $B \in \R$.
    Suppose for any subset $S \subseteq [n]$ with $|S| \le k$, we have $\abs{\sum_{i \in S} a_i} \le B$.
    Then 
    \begin{equation*}
        \Bigset{
            0 \le x_1, \dots, x_n \le 1, \;
            \sum_i x_i \le k
        }
        \; \sststile{1}{x_1, \dots, x_n} \;
        {
            \Bigabs{\sum_{i=1}^{n} a_i x_i} \le B
        } \,.
    \end{equation*}
\end{lemma}
\begin{proof}
    We first show $\sum_{i=1}^{n} a_i x_i \le B$.
    Without loss of generality, assume $a_1 \ge \dots \ge a_n$.
    
    \noindent \textbf{Case 1:} $a_k \ge 0$.
    It is straightforward to see $\{0 \le x_1, \dots, x_n \le 1, \; \sum_i x_i \le k\} \vdash_{1}$
    \begin{align*}
        B - \sum_{i=1}^{n} a_i x_i
        &\ge \sum_{i=1}^{k} a_i - \sum_{i=1}^{n} a_i x_i \tag{$\sum_{i=1}^{k} a_i \le B$} \\
        &= \sum_{i=1}^{k} a_i (1 - x_i) - \sum_{i=k+1}^{n} a_i x_i \\
        &\ge \sum_{i=1}^{k} a_k (1 - x_i) - \sum_{i=k+1}^{n} a_k x_i \tag{$0 \le x_i \le 1$} \\
        &= a_k \cdot \Paren{k - \sum_{i=1}^{n} x_i} \\
        &\ge 0 \,. \tag{$\sum_i x_i \le k$}
    \end{align*}
    
    \noindent \textbf{Case 2:} $a_k < 0$.
    Let $\ell$ be the largest index such that $a_\ell \ge 0$. (Note $\ell \in \{0, 1, \dots, k-1\}$.)
    Then
    \begin{align*}
        B - \sum_{i=1}^{n} a_i x_i
        &\ge \sum_{i=1}^{\ell} a_i - \sum_{i=1}^{n} a_i x_i \tag{$\sum_{i=1}^{\ell} a_i \le B$} \\
        &= \sum_{i=1}^{\ell} a_i (1 - x_i) + \sum_{i=\ell+1}^{n} (-a_i) x_i \\
        &\ge 0 \,,
    \end{align*}
    where in the last inequality we used $a_1, \dots, a_\ell \ge 0$ and $a_{\ell+1}, \dots, a_n < 0$, as well as $0 \le x_i \le 1$ for all $i$.

    Observe that we can apply the same argument above to $-a_1, \dots, -a_n$ and conclude $\sum_{i=1}^{n} (-a_i) x_i \le B$, or equivalently, $\sum_{i=1}^{n} a_i x_i \ge -B$.

    Therefore,
    \begin{align*}
        \Bigset{
            0 \le x_1, \dots, x_n \le 1, \;
            \sum_i x_i \le k
        }
        \; &\sststile{1}{x_1,\ldots,x_n} \;
        \Bigset{
            -B \le \sum_{i=1}^{n} a_i x_i \le B
        } \,. \qedhere
    \end{align*}
\end{proof}
\subsubsection{Approximately satisfying linear operators}
Similar to \cite{hopkins2023robustness}, we give the definition of approximately satisfying linear operators, for dealing with the numerical issues in sum-of-squares exponential mechanism.
\begin{definition}[$\tau$-approximately satisfying linear operators]
\label{def:tau-approximately-satisfying}
Let \(n\in\mathbb{N}\) and \(T\in[n]\). Let \(\mathcal{A}\) be a system of
polynomial inequalities in variables \(w=(w_1,\ldots,w_n)\) and \(z\)
(potentially many) such that
\[
  \mathcal{A}_T
  = \{\, q_1(w,z)\ge 0,\ldots,q_m(w,z)\ge 0 \,\}
    \;\cup\;
    \left\{\, \sum_{i=1}^{n} w_i \ge n-T \,\right\}.
\]
Let \(\tau>0\). We say that a linear operator \(\mathcal{L}\)
\(\tau\)-approximately satisfies \(\mathcal{A}_T\) at degree \(D\) if the
following hold:
\begin{enumerate}
  \item \(\mathcal{L}1=1\).
  \item For all polynomials \(p\) such that \(\deg(p^2)\le D\) and
  \(\|\mathcal{R}(p)\|_2\le 1\) (where \(\mathcal{R}(p)\) is the vector
  representation of the coefficients of \(p\)), it holds that
  \(\mathcal{L}\,p^2 \ge -\tau T\).
  \item For all \(i=1,\ldots,m\) and polynomials \(p\) such that
  \(\deg(p^2\cdot q_i)\le D\) and \(\|\mathcal{R}(p)\|_2\le 1\) (where
  \(\mathcal{R}(p)\) is the vector representation of the coefficients of \(p\)),
  it holds that \(\mathcal{L}\,p^2 q_i \ge -\tau T\).
  \item For every polynomial \(p\) such that
  \(\deg\bigl(p^2\cdot\sum_{i=1}^n w_i-n+T\bigr)\le D\) and
  \(\|\mathcal{R}(p)\|_2\le 1\) (where \(\mathcal{R}(p)\) is the vector
  representation of the coefficients of \(p\)), it holds that
  \[
    \mathcal{L}\biggl(p^2\sum_{i=1}^n w_i-n+T\biggr)
    \ge -5\tau T n .
  \]
\end{enumerate}
\end{definition}

%% file: content/mean-estimation.tex
\section{Robust and Private Bayesian Mean Estimation}\label{sec:gaussian-posterior-mean}

In this section we prove our main results for Bayesian Gaussian mean estimation. 
The model is as follows:
\begin{definition}[Bayesian mean estimation under Gaussian prior]\label{def:bayesian-meanestimation-under-gaussian-prior}
    A parameter $\mu$ is drawn from a prior $\mathcal N(0, \Sigma)$, and then $n$ samples $x = (x_1,\ldots,x_n)$ are drawn i.i.d. from $\mathcal N(\mu, \Id_d)$. Given these samples, our goal is to privately estimate the posterior mean $\mu_{\mathrm{post}} :=\E[\mu | x]$ to $\ell_2$ norm.
\end{definition}

\noindent Our main positive result for this problem is \cref{thm:mean_estimation_upper_informal}, restated here for convenience:

\restatetheorem{thm:mean_estimation_upper_informal}

\noindent We will also show that the gap between the statistical and computational rates is fundamental, under the low-degree hardness framework:

\begin{theorem}[Information-computation gap for Bayesian posterior mean estimation]\label{thm:comp-lower-bound_mean}
    Consider the case of an isotropic prior $\mu \sim \mathcal N(0, \sigma^2 \Id_d)$, and suppose that $\beta < \frac{1}{12}$, $d > 2$, and $\alpha \le \frac{\sqrt{2\pi}}{63} \sigma$. There is a distinguishing problem $\Pi$ which is degree-$n^{o(1)}$ hard such that, given a polynomial-time algorithm for private posterior mean estimation under prior $\Id_d$ with
    \[
        n = o\left(\frac{\min(d, \log \frac{1}{\beta})}{\alpha^{4/3} \varepsilon^{2/3}}\right)
    \]
    samples, there would be a polynomial-time distinguishing algorithm for $\Pi$.
\end{theorem}

\subsection{Reformulation as private empirical mean estimation}
\label{subsec:private-me-main-results}
 
To prove these results, we first give a straightforward reduction to the task of private \emph{empirical mean} estimation.

Suppose without loss of generality that $\Sigma$ is invertible (otherwise, we can simply change basis to exclude its kernel). By a standard calculation, there is a simple linear correspondence between the posterior mean and the empirical mean $\bar{x}\coloneqq \frac{1}{n}\sum_i x_i$:
\[
    \mu_\mathrm{post} = \Bigl(\Id + \frac{1}{n} \Sigma^{-1}\Bigr)^{-1} \bar x =: \Lambda \bar x\,.
\]
Obtaining an estimate $\hat \mu$ of the posterior mean in $L_2$ norm thus corresponds to obtaining an estimate $\hat x$ of the empirical mean in a certain Mahalanobis norm:
\[
    \|\hat \mu - \mu_\mathrm{post}\|_2 = \|\Lambda(\hat x - \bar x)\| = \|\hat x - \bar x\|_{\Lambda^{-2}}
\]
Furthermore, note that after applying (known) linear transformation $\Lambda$ to the dataset, this is equivalent to estimating the empirical mean of a dataset of samples from $\mathcal{N}(0,\Lambda^2)$ in $L_2$ norm.

The remainder of this section will thus focus on this latter problem of empirical mean estimation for Gaussians of covariance $\Lambda^2$. \Cref{thm:mean_estimation_upper_informal} follows directly (see \cref{app:bayesian-me-from-emp-me} for details) from the following upper bound:

\begin{theorem}[Upper bounds for private empirical mean estimation]
\label{thm:anisotropic-private-est}
    Fix a public parameter $R > 1$. For all $\mu \in \mathbb R^d$, given $n$ i.i.d. samples from $\mathcal N(\mu, \Lambda^2)$  with empirical mean $\|\bar x\|_2$, if $\|\bar x\|_2 \le R$,
    \begin{enumerate}
        \item There is a (computationally inefficient) $\varepsilon$-private estimator which uses at most
            \[
            n = \tilde O\biggl(\frac{\mathrm{tr}(\Lambda) + \|\Lambda\|_\mathrm{op} \log(1/\beta)}{\alpha \varepsilon} + \frac{d \log(R/\alpha)}{\varepsilon}\biggr)\,, 
        \]
        then the estimator will output $\hat x$ such that $\|\hat x - \bar x\|_2 \le \alpha$ with probability at least $1 - \beta$.
        \item There is a polynomial-time $\varepsilon$-private estimator which uses at most
        \[
            n = \tilde O\biggl(\frac{\tr(\Lambda^{4/3}) + \|\Lambda^{4/3}\|_\mathrm{op} \log(1/\beta)}{\alpha^{4/3} \varepsilon^{2/3}} + \frac{\tr(\Lambda) + \|\Lambda\|_\mathrm{op} \log(1/\beta)}{\alpha \varepsilon} + \frac{d \log(R/\alpha)}{\varepsilon}\biggr)\,,
        \]
        then the same conclusion holds.
    \end{enumerate}
    Here, the $\tilde O$ hides log factors of $\alpha, n, d$. If $\Lambda$ is proportional to $\Id_d$, then it only hides log factors of $\alpha, n$.
\end{theorem}

\noindent A simple bucketing argument on the eigenvalues of $\Lambda$ shows that to prove these bounds for general $\Lambda$, it suffices to prove them in the special case where $\Lambda$ is proportional to $\Id_d$, at the loss of a logarithmic factor in $d$ (see \cref{app:aniso-bucketing-argument} for details).

The statistical rate in \cref{thm:anisotropic-private-est} is optimal up to $\log 1/\alpha$ factors, as improving this would violate existing lower bounds for private parameter mean estimation. More precisely, the last term in the rate is necessary by a standard packing argument, the $\tr(\Lambda)$ term is necessary by~\cite[Theorem 1.3]{dagan2024dimension}, and the $\norm{\Lambda}_{\sf op} \log(1/\beta)$ term is necessary by the univariate lower bound of~\cite[Theorem 6.2]{karwa-vadhan} applied to the top eigenspace of $\Lambda$.

Additionally, we provide low-degree evidence that the computational rate in \cref{thm:anisotropic-private-est} is optimal.

\begin{theorem}[Information-computation gap for private empirical mean estimation]\label{thm:isotropic-lbd_mean}
    Let $\beta < \frac{1}{3}$ and $d>1$. There is a distinguishing problem $\Pi$ (see \cref{def:eta-gaussian-mixtures}) which is degree-$d^{o(1)}$ hard such that if there were a polynomial-time $\epsilon$-DP algorithm for empirical mean estimation for isotropic Gaussian data which achieves error $\alpha$ with probability $1 - \beta$ using
    \[
        n = o\left(\frac{\min(d, \log \frac{1}{\beta})}{\alpha^{4/3} \varepsilon^{2/3}}\right)
    \]
    samples, then there would be an efficient distinguishing algorithm for $\Pi$.
\end{theorem}

\noindent Although interesting in its own right, the above is a worst-case lower bound that does not suffice to prove \Cref{thm:comp-lower-bound_mean}, which holds even in the average case over $\mu \sim \mathcal N(0, \Id_d)$. We show an average-case variant of this lower bound that holds for sufficiently small $\alpha$ and $\beta$. A full proof is contained in 
\Cref{app:bayesian-me-from-emp-me}.

\subsection{Statistical rate for private isotropic empirical mean estimation}
\label{sec:ineff-priv-iso-emp}

In this section, we prove the statistical bound in \cref{thm:anisotropic-private-est} in the case where $\Lambda$ is isotropic. Below, we first prove a statistical upper bound for \emph{robust} empirical mean estimation for isotropic Gaussians, and then apply the privacy-to-robustness reduction of~\cite{hopkins2023robustness}.

\begin{theorem}[Statistical rate for robust empirical mean estimation]
\label{thm:robust-empirical-mean-stat}
    Suppose $0 \le \eta < 1/2$ is bounded away from $1/2$. There exists a (computationally inefficient) $\eta$-robust estimator that, for all $\mu \in \mathbb R^d$, given $n$ i.i.d. samples from $\mathcal N(\mu, \Id_d)$ with empirical mean $\bar x$, outputs $\hat x$ such that $\|\hat x - \bar x\|_2 \le \alpha$, with probability at least $1 - \beta$, for
    \begin{equation}
        \alpha = O\biggl(\eta\sqrt{\log(1/\eta)} + \sqrt{\eta}\cdot \sqrt{\frac{d + \log(1/\beta)}{n}}\biggr)\,.\label{eq:alphabound}
    \end{equation}
\end{theorem}

\noindent Note that this implies that to achieve error $O(\eta\sqrt{\log(1/\eta)})$ in robustly estimating the empirical mean, it suffices to use $\tilde{O}(d/\eta)$ samples, rather than the usual ``parametric'' rate of $\tilde{O}(d/\eta^2)$. Additionally, observe that the error, as expected, tends to zero as $\eta \to 0$.

The proof of \cref{thm:robust-empirical-mean-stat} is a basic resilience-type argument~\cite{steinhardt2017resilience}: 

\begin{proof}
    We will construct our estimator in the following way: let $x_1,...,x_n$ be the uncorrupted input samples with mean $\bar x$, and let $\tilde x_1,...,\tilde x_n$ be the corrupted input samples. Our estimator will search over vector-valued variables $y_1,\ldots,y_n$ satisfying the following constraints:
    \begin{itemize}
        \item $|\{i: \tilde x_i \neq y_i\}| \le \eta n$
        \item For all subsets $S \subseteq [n]$ of size at most $2\eta n$:
        \[
            \left\|\frac{1}{|S|} \sum_{i \in S} (y_i - \bar y)\right\|_2 \le C\sqrt{\frac{d + \log \frac{1}{\beta}}{\eta n} + \log \frac{1}{\eta}}
        \]
    \end{itemize}
    for sufficiently large absolute constant $C > 0$, where $\overline{y} = \frac{1}{n}\sum_i y_i$.
    If the above conditions are infeasible, the estimator will output $\bot$. However, by standard concentration and a union bound over all $O(1/\eta)^{\eta n}$ subsets $S$ of size at most $2\eta n$, the constraints are satisfied by the ground truth $(y_1,\ldots,y_n) = (x_1,\ldots,x_n)$.
    Next, let $(y_1,...,y_n)$ and $(y_1',...,y_n')$ be two assignments which both satisfy the above constraints, and let $S = \{i: y_i \neq y_i'\}$ (where we note $|S| \le 2 \eta n$). Define $\overline{y}' = \frac{1}{n}\sum_i y'_i$. We will prove that $\|\bar{y} - \bar{y}'\|_2 \le O(\alpha)$ for $\alpha$ given in Eq.~\eqref{eq:alphabound}, which will imply that any assignment satisfying the above set of constraints will achieve the bound claimed in the theorem. Then
    \begin{align*}
        \|\bar y - \bar y'\|_2 &=\frac{1}{n}\Biggr\|\sum_{i \in S} (y_i - y'_i)\Biggl\|_2 \\
        &= \frac{1}{n} \Biggl\|\sum_{i \in S} (y_i - \bar y) - \sum_{i \in S} (y_i' - \bar y') + 2\eta n(\bar y - \bar y')\Biggr\|_2 \\
        &\le 2C\eta \sqrt{\frac{d + \log \frac{1}{\beta}}{\eta n} + \log \frac{1}{\eta}} + 2 \eta \cdot \|\bar y - \bar y'\|_2\,.
    \end{align*}
    After rearranging and using the assumption $\eta$ is bounded away from $1/2$, we conclude that $\|\bar{y} - \bar{y}'\|_2 \le O(\alpha)$ as desired.
\end{proof}

\noindent We now prove the statistical rate for private isotropic empirical mean estimation (first part of \cref{thm:anisotropic-private-est} when $\Lambda = \Id$) by combining \cref{thm:robust-empirical-mean-stat} and \cref{lemma:inefficient-reduction-meta}.

\begin{proof}[Proof of first part of \cref{thm:anisotropic-private-est}]
Let $\eta^* \coloneqq 1/3$, and let $\alpha(\eta)$ be the error term in \cref{thm:robust-empirical-mean-stat} for $\eta < \eta^*$ for robust empirical mean estimation, and for larger $\eta$ let $\alpha(\eta) = R$.
Let $\alpha$ be the target accuracy of the private algorithm, which we assume to be smaller than a sufficiently small constant, and let $\eta_0$ be such that $\alpha(\eta_0) = \alpha$.
Now from \cref{lemma:inefficient-reduction-meta}, there is a computationally inefficient $\epsilon$-DP algorithm that takes $n$ samples from $\cN(\mu, \Id)$ and outputs $\hat{x}$ such that $\normt{\hat{x} - \bar{x}} \le 2 \alpha(\eta_0) = 2 \alpha$ with probability $1-2\beta$, as long as 
\begin{equation}
n \ge \max_{\eta_0 \le \eta' \le 1}
\frac{d \cdot \log \frac{2\alpha(\eta')}{\alpha(\eta_0)} + \log(1/\beta)}{\eta' \epsilon} \mper \label{eq:nbound}
\end{equation}
This maximum is at most the sum of the maximums over $\eta_0 \le \eta' \le \eta^*$ and $\eta^*\le \eta' < 1$, that is,
\begin{equation*}
\frac{d \log \frac{2\alpha(\eta^*)}{\alpha} + \log(1/\beta)}{\eta_0 \eps} + \frac{d \log\frac{2R}{\alpha} + \log(1/\beta)}{\eps} \mper
\end{equation*}
Note that because $\eta_0$ is such that $\alpha = O(\eta_0\sqrt{\log 1/\eta_0} + \sqrt{\frac{\eta_0 d}{n}})$, we must have $\eta_0 = \tilde{\Omega}(\min\{\alpha, \alpha^2n/d\})$. Substituting this into Eq.~\eqref{eq:nbound} yields the claimed bound. %
\end{proof}

\subsection{Computational rate for private isotropic empirical mean estimation}
\label{sec:eff-priv-iso-emp}

Here we will prove the computational bound in Theorem~\ref{thm:anisotropic-private-est} in the case where $\Lambda$ is isotropic. We will first develop an efficient sum-of-squares algorithm for robust empirical mean estimation, and then apply the privacy-to-robustness reduction.

\subsubsection{Efficient robust empirical mean estimation}
\label{sec:apply_kmz}
\begin{theorem}
\label{thm:robust-empirical-mean-efficient}
    Let $\eta > 0$ be smaller than a sufficiently small constant.  There is an efficient $\eta$-robust estimator such that, for all $\mu \in \mathbb R^d$, given $n$ i.i.d. samples from $\mathcal N(\mu, \Id_d)$ such that the uncorrupted samples have empirical mean $\bar{x}$, the estimator will output an $\hat x$ such that $\|\hat x - \bar x\|_2 \le \alpha$ except with probability at most $\beta$ for
    \[
        \alpha \le O\Biggr(\eta\sqrt{\log 1/\eta} + \sqrt{\eta \sqrt{\frac{d + \log \frac{1}{\beta}}{n}}}\Bigg)\,.
    \]
\end{theorem}

\begin{proof}
The proof follows from a careful analysis of the algorithm in \cite{kothari2022polynomial}.
Let $\alpha_0, \alpha_1, \alpha_2$ be values to be determined later.
Let $x_1,\ldots,x_n$ denote the observed samples, an $\eta$ fraction of which have been corrupted. Let $x^{\circ}_1,\ldots,x^{\circ}_n$ denote the original, uncorrupted samples, which we do not have access to and which have empirical mean $\bar{x}$.
We consider the following set of polynomial constraints in the indeterminates $w_1, \dots, w_n$ (indicating our guesses for which points are uncorrupted), and $x_1', \dots, x_n'$:
\begin{equation*}
\mathcal{A} =
\left\{
\begin{aligned}
&\forall i\in [n]: &w_i &= w_i^2 \\
&\forall i\in[n]: &w_i x_i &= w_i x_i' \\
& &\sum_{i = 1}^n w_i &\ge (1-\eta) n \\
&\forall v \in \mathbb{S}^{d-1}: &\E_i \iprod{x_i' - \mu', v}^2&\le (1 + \alpha_0)
\end{aligned}
\right\}
\end{equation*}
where $\mu' = \E_i x_i'$.
Then the algorithm is as follows:
\begin{algorithmbox}[Robust algorithm for estimating the empirical mean]\label{algo:efficient-emp-mean-est}
    \mbox{}\\
    \textbf{Input:} $\eta$-corrupted samples $x_1, \ldots x_n \in \R^d$ \\
    \textbf{Output:} a polynomial time estimate $\hat{x} \in \R^d$ for the empirical mean
\begin{enumerate}
    \item Solve the degree-$O(1)$ SoS relaxation of the constraint system $\mathcal{A}(x', w ; x)$, and obtain a pseudo-expectation operator $\tE$. 
    \item Set $\hat{x} \gets \tE[\mu']$.
    \item \textbf{Return} $\hat{x}$.
\end{enumerate}
\end{algorithmbox}
Suppose we know that for all vectors  $v \in \mathbb{S}^{d-1}$ and $b \in [0,1]^n$ such that $\E_i b_i \ge 1- \eta$, the bounds 
\begin{align*}
    \Abs{\E_i b_i \iprod{x_i^\circ - \bar{x}, v}} &\le \alpha_1 \\
    \Abs{\E_i b_i \iprod{x_i^\circ - \bar{x}, v}^2- 1} &\le \alpha_2    
\end{align*}
hold.
Then tracking the error terms in the proof of \cite[Lemma 4.1]{kothari2022polynomial} gives us that if $\pE$ is a degree-$O(1)$ pseudoexpectation that satisfies $\cA$, then 
\begin{equation*}
    \normt{\pE \mu' - \bar{x}} \le O(\alpha_1 + \sqrt{\eta \cdot (\alpha_0 + \alpha_2) + \eta^2})\,.
\end{equation*}
We include a proof of this statement in \Cref{sec:isotropic-estimation-lemma} (see \cref{lem:technicallem-unit-variance}).
Now it suffices to set $\set{\alpha_0, \alpha_1, \alpha_2}$. Applying \Cref{lem:stability-Gaussian}, since $n \ge d + \log(1/\beta)$, we can take 
\begin{align*}
\alpha_0 = \alpha_2 &= O\Paren{\sqrt{\frac{d + \log(1/\beta)}{n}} + \eta \log(1/\eta)}  \mcom
\\
\text{ and } \alpha_1 &= O\Paren{\sqrt{\frac{\eta(d + \log(1/\beta))}{n}} + \eta \sqrt{\log(1/\eta)}} \mper
\end{align*}
Substituting these choices of $\alpha_0, \alpha_1, \alpha_2$, finishes the proof. 
For feasibility, note that pseudoexpectations generalize expectations which generalize masses on single points. In order to verify feasibility it suffices to take $\pE$ to be the pseudoexpectation that assigns each $x_i'$ to $x_i^\circ$ and sets each $w_i$ equal to $1$ if that sample is uncorrupted and to $0$ otherwise.
\end{proof}

\begin{lemma}[Stability]\torestate{
\label{lem:stability-Gaussian}
Let $x_1,\ldots,x_n$ be i.i.d. samples from $\cN(\mu, \Id_d)$ and $\bar{x}$ be their empirical mean. Then with probability $1-\beta$, for all vectors $v \in \mathbb{S}^{d-1}$ and $b \in [0,1]^n$ such that $\E_i b_i \ge 1-\eta$, we have 
\begin{align*}
\Abs{\E_i b_i \iprod{x_i - \bar{x}, v}} &\le O\Paren{\sqrt{\frac{\eta(d + \log(1/\beta))}{n}}
+
\eta\sqrt{\log(1/\eta)}
} 
\\
\Abs{\E_i b_i \iprod{x_i - \bar{x}, v}^2 - 1}
&\le 
O\Paren{
\sqrt{\frac{d + \log(1/\beta)}{n}} +\frac{d + \log(1/\beta)}{n}  + \eta \log(1/\eta)}
\end{align*}
}
\end{lemma}

\noindent We prove the bound in \cref{sec:mean_concentration}.

\subsubsection{Robustness-to-privacy reduction}
\label{sec:robust_to_privacy_mean}

In this section we apply the robustness to privacy reduction and prove the efficient isotropic case of \Cref{thm:anisotropic-private-est}. Our goal is to apply \Cref{thm:efficient-reduction-meta-pure}. We follow the reduction in Section~5 of \cite{hopkins2023robustness} closely. The program we used in the proof of \Cref{thm:robust-empirical-mean-efficient} is the same as the program used by \cite{hopkins2023robustness}, the only difference being the accuracy guarantee we aim to obtain from the program. Their focus is on estimating the parameter mean of the distribution, while our focus is estimating the empirical mean of the uncorrupted samples. Despite this distinction, we can still instantiate the same reduction by using the same score function. Most of the properties required by \Cref{thm:efficient-reduction-meta-pure} --- bounded sensitivity, quasi-convexity, efficient computability, efficient computability, and efficient finding of low-scoring points --- follow automatically. It suffices for us to verify the accuracy guarantees and the volume ratios.

We now define the score function. First we need to define what it means for a candidate point to be a certifiable mean:

\begin{definition}[$(\alpha, \tau, T)$-Certifiable Mean]
Let $x$ be a set of inputs, we call a candidate point $\tilde{x}$ an \emph{$(\alpha, \tau, T)$-certifiable mean} for the this set of inputs if there exists a linear functional $\cL$ that takes as inputs polynomials of degree $O(1)$ in indeterminates $(x', w, M)$ with the following properties: $\cL$ satisfies
$\normt{\cL \mu' - \tilde{x}}\le \alpha$ and $\normt{\cL \mu'} \le 2R + \tau T$ where $\mu' = \E_i x_i'$, $\normf{\cR(\cL)} \le R' + n \tau$, for $R' = \poly(n, d, R)$ sufficiently large, where $\cR$ maps $\cL$ to its matrix representation.
Moreover, $\cL$ must 
$\tau$-approximately (as in \Cref{def:tau-approximately-satisfying}) satisfy the following system of polynomial inequalities up to degree $6$.
\begin{equation*}
\left\{
\begin{aligned}
&\forall i\in [n]: &w_i &= w_i^2 \\
&\forall i\in[n]: &w_i x_i &= w_i x_i' \\
& &\sum_{i = 1}^n w_i &\ge n - T \\
& &\E_i (x_i' - \mu')\transpose{(x_i' - \mu')} + M\transpose{M} &= (1 + \alpha_0) I
\end{aligned}
\right\}
\end{equation*}
We view that last equality as $k^2$ many equalities.
Moreover, we choose $\eta$ to be such that $\alpha = \alpha(\eta) = \Theta(\eta\sqrt{\log 1/\eta} + \sqrt{\eta\sqrt{\frac{d + \log(1/\beta)}{n}}})$ (as in \Cref{thm:robust-empirical-mean-efficient}), and we choose $\alpha_0 = \alpha_0(\eta)$ to be as in the proof of \Cref{thm:robust-empirical-mean-efficient}, as a function of $\eta$ and consequently $\alpha$, $\alpha_0(\eta) = \Theta(\eta\log(1/\eta) + \sqrt{\frac{d + \log(1/\beta)}{n}})$.
\end{definition}

\noindent Given the above definition the score function is defined as follows:
\begin{definition}[Score Function]
Let $\alpha, \tau > 0$.
Let $\mathbb{B}_2^d(2R + n\tau + \alpha)$ denote the $\ell_2$ ball of radius $2R + n\tau + \alpha$ in $d$ dimensions.
Given $x$ be a set of inputs, and $\tilde{x} \in \R^d$, we define the \emph{score function} $\cS$ over $\mathbb{B}_2^d(2R + n\tau + \alpha)$ as
\begin{equation*}
\cS(\tilde{x}; x, \alpha, \tau) = \min\{T\ge 0: \tilde{x} \text{ is an } (\alpha, \tau, T) \text{-certifiable mean with respect to } x\} \mper
\end{equation*}
\end{definition}

\noindent Next we bound the volume ratios.

\begin{lemma}[Volume of High and Low Scoring Points]
\label{lem:volume-ratio-bound}
Assume $R = \Omega\paren{\alpha}$ is sufficiently large.
Let $\eta \ge 0$ be at most a sufficiently small constant $\eta^*$.
Assume the set of input points $x$ satisfy the properties of \Cref{lem:stability-Gaussian}. Then:
\begin{enumerate}
    \item Every point in a ball of radius $\alpha(\eta)$ around $\bar{x}$ has score at most $\eta n$
    \item Every point of score at most $\eta^*n$ is contained in a ball of radius $O(\alpha(\eta^*))$.
\end{enumerate}
\end{lemma}
\begin{proof}
Note that since the set of the samples satisfy the properties of \Cref{lem:stability-Gaussian}. The original program is feasible: consider the solution given by the point mass that assigns each $x_i'$ to the uncorrupted values. Therefore it suffices to take $\cL$ to be the pseudoexpectation corresponding to that assignment. Then the closeness constraint on $\cL$, namely $\normt{\cL\tilde{\mu} - \mu} \le \alpha$, will be  satisfied, and since $R \ge \alpha$, the constraint $\normt{\cL \mu} \le 2R + \tau T$, is satisfied. Therefore, every point in radius $\alpha(\eta)$ must have score at most $\eta n$.

For the second part, note that since $\cL$ satisfies the constraint system, then from the guarantees of the robust algorithm we must have that $\normt{\cL \mu' - \bar{x}} \le O(\alpha(\eta^*))$. Additionally, by the closeness constraint on $\cL$ we know that $\normt{\cL \mu' - \tilde{\mu}} \le \alpha(\eta) \le \alpha(\eta^*)$. Triangle inequality gives us the desired result.
\end{proof}

\noindent As mentioned earlier, the rest of the requirements of \Cref{thm:efficient-reduction-meta-pure} are satisfied as the score function is the same as the one in \cite{hopkins2023robustness}. Now we may prove the computationally efficient version of \Cref{thm:anisotropic-private-est} in the case where the prior is isotropic. 

\begin{proof}[Proof of second part of \Cref{thm:anisotropic-private-est}]
Let $\eta$ be such that $\alpha = \alpha(\eta)$, where $\alpha(\eta)$ is as in \Cref{thm:robust-empirical-mean-efficient}.
From \Cref{thm:efficient-reduction-meta-pure} we know that it suffices for $n$ to be such that
\begin{equation*}
n \ge \Omega\Paren{\max_{\eta': \eta \le \eta' \le 1} \frac{\log V_\eta' / V_\eta + \log(1/\beta \eta)}{\eps \eta'}} \mper
\end{equation*}
We divide above into two parts, $[\eta, \eta^*] \cup [\eta^*, 1]$. 
Now we apply \Cref{lem:volume-ratio-bound}.
From the first part it suffices to have
\begin{equation*}
n \ge \Omega\Paren{\frac{d \log (\alpha(\eta^*)/\alpha(\eta)) + \log(1/\beta \eta)}{\eps \eta}} \mper
\end{equation*}
From the second part it suffices to have 
\begin{equation*}
n \ge \Omega\Paren{\frac{d \log(R/\alpha(\eta^*)) + \log(1/\beta \eta)}{\eps}} \mper
\end{equation*}
Plugging in the definition of $\eta$ and $\alpha(\eta^*)$ as a function of $\alpha, d, n, \beta$, it suffices to take $n$ such that 
\begin{equation*}
n \ge \tilde{\Omega}\Paren{\frac{d + \log(1/\beta)}{\alpha^{4/3} \eps^{2/3}} + \frac{d + \log(1/\beta)}{\alpha \eps} + \frac{d \log(R/\alpha) + \log(1/\beta)}{\eps}} \mcom
\end{equation*}
hiding logarithmic factors in $\alpha$.
\end{proof}

\subsection{Information-computation gap for Bayesian mean estimation}
\label{sec:priv-iso-lb}

We might hope for a computationally efficient sum-of-squares algorithm that even matches the statistically optimal rate of \cref{thm:robust-empirical-mean-stat} for robust empirical mean estimation, so that we can efficiently apply the reduction of \cite{hopkins2023robustness} to attain a statistically optimal rate for private empirical mean estimation. In this section, we will prove that this is not possible. We will first show a computational lower bound for the robust problem using the low-degree framework, and then a computational lower bound for the private problem via the privacy-to-robustness reduction proposed in \cite{Georgiev2022PrivacyIR}.

\subsubsection{Hardness of robust empirical mean estimation}
\label{sec:robustmeanhard}

For simplicity, assume $\eta n$ is an integer. To prove an information-computation gap for the robust problem, we will exhibit an efficient reduction from the following problem:

\begin{definition}[$(\eta, \delta)$-Gaussian mixtures distinguishing problem]\label{def:eta-gaussian-mixtures}
    Given $\eta, \delta>0$ and $n, d \in \mathbb N$, consider the following two distributions:
    \begin{itemize}
        \item Null distribution: $\mathcal N(0,\Id_d)$
        \item Planted distribution: Sample $v$ uniformly from the unit sphere. Then the planted distribution is $(1-\eta) \mathcal N(-\eta \delta v, \Id_d) + \eta \mathcal N((1-\eta)\delta v, \Id_d)$.
    \end{itemize}
    The problem is to distinguish between the two distributions to constant advantage, given $n$ samples.
\end{definition}

\noindent This is closely related to a distinguishing problem considered in~\cite{diakonikolas2025sos}, and in \cref{lem:distinguish_hard} at the end of this section, we verify that this problem is computationally hard for $\delta = \Theta(\alpha/\eta)$ under the low-degree framework.

We now proceed to reduce from this problem to obtain a hardness result for empirical mean estimation.

\begin{theorem}[Worst-case hardness of empirical mean estimation]
\label{thm:low-degree-robustness-lower-bound}
    Let $\eta \in [0, 1/3)$. Consider any polynomial-time estimator which uses $n \le o(d\eta^2/\alpha^4)$ samples from $\mathcal N(\mu, \Id_d)$ to perform $\eta$-robust empirical mean estimation to error $\alpha$. If any such estimator succeeds with worst-case probability $\frac{2}{3}$ for all $\|\mu\| \le 21 \alpha$, then there is a polynomial-time algorithm for the $(\eta,\Theta(\alpha/\eta))$-Gaussian mixtures distinguishing problem.
\end{theorem}

\noindent The same reduction will give us the following average-case lower-bound, which will be crucial in proving the hardness of private Bayesian mean estimation as the latter is inherently an average-case problem over the draw of $\mu$ from the prior:

\begin{lemma}[Average-case hardness of robust empirical mean estimation]
\label{lem:low-degree-robustness-lb-average-case}
    Let $\eta \in [0, \frac{1}{3})$ and \linebreak $\mu \sim 21 \alpha v \cdot \text{Bern}(1/2)$ where $v \sim \mathbb S^{d-1}$ uniformly. Consider any polynomial-time estimator which uses $n = o\left(\frac{d \eta^2}{\alpha^4}\right)$ samples from $\mathcal N(\mu, \Id_d)$ to perform $\eta$-robust empirical mean estimation to error $\alpha$. If any such estimator succeeds with average-case probability $\frac{5}{6}$ over the draw of $\mu$ as described, then there is a polynomial-time algorithm for the $(\eta,\Theta(\alpha/\eta))$-Gaussian mixtures distinguishing problem.
\end{lemma}

\noindent Both of these results are immediate from the following reduction:

\begin{lemma}[Estimation implies distinguishability]
\label{lem:robust-me-hard-distributions}
    Assume there is a polynomial-time $\eta$-robust empirical mean estimator that takes in $n$ samples from $\mathcal N(\mu, \Id)$ where $\mu$ is generated as either
    \begin{itemize}
        \item $\mu = 0$
        \item $\mu = 21 \alpha v$ where $v \sim \mathbb S^{d-1}$ uniformly.
    \end{itemize}
    and in either case outputs an estimate of the empirical mean to $\ell_2$ accuracy $\alpha$ with probability $1-\beta > 2/3$. Then there is a polynomial-time algorithm that solves the $(\eta, \delta)$-Gaussian mixtures distinguishability problem from \cref{def:eta-gaussian-mixtures} with $n$ samples to constant advantage, for $\delta = 21\alpha/\eta$.
\end{lemma}
\begin{proof}
    Our algorithm will be as follows: run the robust empirical mean estimator on the $n$ samples to get an estimate $\hat \mu$, and let $\tilde x$ be the mean of the input samples. We will guess ``null'' if $\|\hat \mu - \tilde x\| \le \alpha$ and guess ``planted'' otherwise. The idea is that, if we are drawing from the null distribution, clearly
    \begin{equation*}
        \|\hat \mu - \tilde x\| \le \alpha
    \end{equation*}
    with probability $1-\beta > 2/3$ (by the accuracy guarantee in the $\eta=0$ case). However, if we are drawing from the planted distribution, we can view the ``planting'' process as an $\eta$-corruption of clean samples drawn from $\mathcal N(-\eta \delta v,\Id_d)$; denote the empirical mean of these ``clean samples'' by $\bar x$. Furthermore, letting $X_i \sim \mathcal N(-\eta \delta v, \Id_d)$, $Y_i \sim \mathcal N((1-\eta) \delta v, \Id_d)$, and $Z_i \sim \mathcal N(\delta v, 2 \Id_d)$, we see that $\tilde x - \bar x$ is distributed as
    \[
        \frac{1}{n} \sum_{i=1}^{\eta n} [Y_i-X_i] = \frac{1}{n} \sum_{i=1}^{\eta n} Z_i \sim \mathcal N(\eta \delta v, \frac{2 \eta}{n} \Id_d)
    \]
    It is not hard to find that
    \[
        \sup_{\sigma > 0} \mathop{\mathbb P}_{X \sim \mathcal N(v, \sigma^2\Id_d)}\left[\|X\|_2 \le \frac{1}{10}\right] \le \frac{1}{9^d \sqrt{\pi d}} < \frac{1}{10}
    \]
    (from upper bounding the PDF of the Gaussian on the ball and the ball's volume, plugging in the $\sigma$ from first-order optimality, and then using Stirling's approximation), so in particular we can rescale to get
    \[
        \mathbb P\left[\|\tilde x - \bar{x}\|_2 \le \frac{1}{10} \eta \delta\right] < \frac{1}{10}
    \]
    Thus in the second case, a reverse triangle inequality between the above and the algorithm's accuracy guarantee suffices to give
    \begin{align*}
        \|\hat \mu - \tilde x\| &\ge \|\bar x - \tilde x\| - \|\hat \mu - \bar x\| \\
        &\ge \frac{1}{10} \eta \delta - \alpha \\
        &> \alpha
    \end{align*}
    except with probability $\frac{1}{3} + \frac{1}{10} < \frac{1}{2}$ by a union bound, giving our distinguisher constant advantage.
\end{proof}

\noindent Note that an empirical mean estimator that beats the lower bound stated in \cref{thm:low-degree-robustness-lower-bound} would clearly meet the criteria for the estimator used in \cref{lem:robust-me-hard-distributions}. Additionally, an estimator that beats the average-case lower bound in \cref{lem:low-degree-robustness-lb-average-case} must succeed with probability $2/3$ in each of the two cases, and thus can also be reduced to the testing problem by \cref{lem:robust-me-hard-distributions}.

It remains to verify low-degree hardness of the $(\eta,\delta)$-Gaussian mixtures distinguishing problem:

\begin{lemma}\label{lem:distinguish_hard}
    The degree-$D$ advantage for the $(\eta, \delta)$-Gaussian mixtures distinguishability problem for $\delta = O(\frac{\alpha}{\eta})$ given $n$ samples is at most $O(1)$ if $n\eta^2 \ge 1$ and $n \le \Omega\left(\frac{d \eta^2}{\alpha^4 \mathrm{poly}(D)}\right)$.
\end{lemma}
\begin{proof}
    We will apply \cref{lem:general-ldlr}, setting $\kappa := \eta^2$, $\tau := \frac{1}{\delta^4}$, and $j_* := 2$ therein.
    The first condition in \cref{lem:general-ldlr} holds by construction, and for the second condition we need to show that for all $j \ge j_*$,
    \begin{align*}
        \bigl|\mathop{\mathbb E}_{X \sim A} h_j(X) \bigr|^2 &\le j^j \kappa \tau^{-j/j^*}
    \end{align*}
    Using \cref{fact:hermite-props},
    \begin{align*}
        \bigl|\mathop{\mathbb E}_{X \sim A} h_j(X) \bigr| &= \bigl|(1-\eta) \mathop{\mathbb E}_{X \sim \mathcal N(- \eta \delta,1) } [h_j(X)] + \eta \mathop{\mathbb E}_{X \sim \mathcal N(\delta - \eta \delta,1) } [h_j(X)]\bigr| \\
        &\le \bigl|(1-\eta) \cdot (-\eta \delta)^j/\sqrt{j!}\bigr| + \bigl|\eta \cdot (\delta-\eta \delta)^j/\sqrt{j!}\bigr|
    \end{align*}
    and squaring both sides gives
    \begin{align*}
        \bigl|\mathop{\mathbb E}_{X \sim A} h_j(X)\bigr|^2 &\le 2\bigl|(1-\eta) \cdot (-\eta \delta)^j/\sqrt{j!}\bigr|^2 + 2\bigl|\eta \cdot (\delta-\eta \delta)^j/\sqrt{j!}\bigr|^2 \\
        &\le 2\bigl|(-\kappa^{1/2} \tau^{-1/4})^j/\sqrt{j!}\bigr|^2 + 2\kappa \bigl|\tau^{-j/4}/\sqrt{j!}\bigr|^2 \\
        &\le 2 \kappa^j \tau^{-j/2}/j! + 2 \kappa \tau^{-j/2}/j! \\
        &\le 4 \kappa \tau^{-j/j^*}/j! \\
        &\le j^j \kappa \tau^{-j/j^*}
    \end{align*}
    as desired. This yields a degree-$D$ lower bound for sample complexity
    \[
        n \ge d^{j^*/2} \frac{\tau}{\kappa} \cdot \frac{1}{\poly(D)} = \Omega\left(\frac{d \eta^2}{\alpha^4}\right)\cdot \frac{1}{\poly(D)}
    \]
    as claimed.
\end{proof}

\subsubsection{Hardness of private empirical mean estimation}
\label{sec:privatemeanhard}

We can use the privacy-to-robustness reduction of \cite{Georgiev2022PrivacyIR} along with \cref{thm:low-degree-robustness-lower-bound} to prove the computational lower bound of \cref{thm:isotropic-lbd_mean}:

\begin{proof}[Proof of \cref{thm:isotropic-lbd_mean}]
Let $p,q>0$ to be set later, and suppose that we could perform private empirical mean estimation with $n_{\sf priv}$ isotropic samples, where
\[
    n_{\sf priv} = o\left(\frac{\min(d, \log \frac{1}{\beta})}{\alpha^p \varepsilon^q}\right)
\]
Without loss of generality, say $\log 1/\beta= d$ (otherwise we could decrease the dimension or increase $\beta$, which should make the problem easier). By the privacy-to-robustness reduction from \cite{Georgiev2022PrivacyIR}, restated as \cref{thm:privacy-robustness-reduction} above, there is an $\eta$-robust empirical mean estimator with sample complexity $n_{\sf rob} = o\left(\frac{d}{\alpha^p \varepsilon^q}\right)$ for corruption fraction $\eta = \frac{d}{\varepsilon n_{\sf rob}}$.

Solving for $\epsilon$ and reorganizing variables, there is an $\eta$-robust empirical mean estimator with sample complexity
\begin{equation*}
    n_{\sf rob} = o\left(\frac{d}{\alpha^p} \cdot \left(\frac{\eta n_{\sf rob}}{d}\right)^q\right) = o\left(\frac{d^{1-q} \eta^q}{\alpha^p} n_{\sf rob}^q\right) = o\left(\frac{d \eta^{q/(1-q)}}{\alpha^{p/(1-q)}}\right)
\end{equation*}
For $p = 4/3$ and $q = 2/3$, this would contradict the computational hardness result of \cref{thm:low-degree-robustness-lower-bound} for robust empirical mean estimation.
\end{proof}

\noindent We can apply the same reduction to the average-case robust lower bound stated in \cref{lem:low-degree-robustness-lb-average-case} to get an average-case private lower bound:
\begin{corollary}
\label{cor:private-comp-lb-average-case}
     Let $\beta < 1/6$ and assume $\mu = 21 \alpha v\cdot \mathrm{Ber}(1/2)$ where $v\sim \mathbb S^{d-1}$. If there is a polynomial-time $\varepsilon$-DP estimator of the empirical mean that succeeds to error $\alpha$ and confidence $\beta$ given $n$ samples from $\mathcal N(\mu, \Id_d)$ and has sample complexity
     \[o\Bigg(\frac{\min\big(d, \log \frac{1}{\beta}\big)}{\alpha^{4/3} \varepsilon^{2/3}}\Bigg) \mcom
     \]
     then there is a polynomial-time algorithm for the $(\eta,\delta)$-Gaussian mixture distinguishing problem for $\eta = \Theta(\alpha^{4/3}/\epsilon^{1/3})$ and $\delta = \Theta(\epsilon^{1/3}/\alpha^{1/3})$ with $o(\frac{d}{\alpha^{4/3}\varepsilon^{2/3}})$ samples, which is low-degree-hard.
\end{corollary}

\noindent Finally, we complete the proof of \cref{thm:comp-lower-bound_mean}, which establishes an information-computation gap for the original question of private \emph{Bayesian} mean estimation. The proof requires adapting the average-case lower bound of \cref{cor:private-comp-lb-average-case} so that the distribution over $\mu$ matches the Gaussian prior in the setting of \cref{thm:comp-lower-bound_mean}.

\begin{proof}[Proof of \Cref{thm:comp-lower-bound_mean}]
    We construct a coupling of $(X, Y)$ with the following properties:
    \begin{itemize}
        \item $X$ has marginal distribution $\mathcal D_1 :=\mathcal N(0, \sigma^2 I)$.
        \item $Y$ has marginal distribution $\mathcal D_2$ that is $1/6$-close in TV distance to the prior.
        \item $(X,Y)$ are jointly distributed as follows: first $X$ is drawn according to $\mathcal{D}_1$. Then, $Y$ is set to be equal to $X + 21 \alpha v$ where $v \sim \mathbb S^{d-1}$ uniformly and independently and $\alpha$ is chosen appropriately.
    \end{itemize}

    \noindent Under this construction, by standard bounds relating TV to parameter distance for Gaussians, there is an absolute constant $c > 0$ such that the second property is satisfied provided $\alpha \le c\sigma$.
    
    The third property combined with \cref{cor:private-comp-lb-average-case} shows that no $\varepsilon$-private estimator of the empirical mean using 
    \begin{equation}
        o\Bigg(\frac{\min(d, \log \frac{1}{\beta})}{\alpha^{4/3} \varepsilon^{2/3}}\Bigg) \label{eq:toofew}
    \end{equation}
    samples can be $\alpha$-accurate with probability more than $5/6$ when $\mu$ is drawn from $\frac{1}{2}(\mathcal D_1 + \mathcal D_2)$. Indeed, suppose to the contrary that some efficient estimator achieves success probability $5/6$. Then by averaging, there must be some $x$ such that conditioned on $X=x$, it succeeds with probability $5/6$ over $\mu = x + 21\alpha \nu \cdot \mathrm{Ber}(1/2)$. The guarantee in \cref{cor:private-comp-lb-average-case} is clearly translation-invariant, so the existence of this estimator implies an efficient algorithm for the $(\eta,\delta)$-Gaussian mixture distinguishing problem, contradicting low-degree hardness. 
    
    The TV distance between $\frac{1}{2}(\cD_1 + \cD_2)$ and $\cD_1$ is at most $1/12$, so by triangle inequality, no $\varepsilon$-private estimator of the empirical mean with sample complexity given by Eq.~\eqref{eq:toofew} can be $\alpha$-accurate with probability more than $\frac{11}{12}$ when $\mu$ is drawn from the prior. Therefore we require $\beta < \frac{1}{12}$, and the connection between Bayesian mean estimation and empirical mean estimation elucidated in \cref{subsec:private-me-main-results} yields the desired lower bound for the former.
\end{proof}

%% file: content/online-streaming.tex
\subsection{Application to private Gaussian mean estimation under continual observation}
\label{sec:online}

As an application of our algorithm for private Bayesian mean estimation with Gaussian prior, we consider an online model of mean estimation under continual observation.:

\begin{definition}[Private online mean estimation]
    Let $\mu \in \R^d$ be an unknown parameter with $|\mu| \le R$. Suppose we see a stream of $k$ batches each of size $n$, where the samples are independently sampled from $\mathcal{N}(\mu,\Id)$. After each successive batch, we must report an estimate of $\mu$ given what we have seen so far, and the overall mechanism must be differentially private with respect to all $nk$ samples encountered in the stream.
\end{definition}

If the streaming algorithm can safely keep non-private state hidden from the public, the solution becomes simple: take
\[
    \varepsilon_i = \frac{\varepsilon}{i \log k}
\]
and after each iteration publish an $\varepsilon_i$-private empirical mean estimate of the samples you have seen so far. By a simple calculation, if $n \ge \frac{d \log \frac{R}{\alpha}}{\varepsilon}$ then with probability $1-\beta$, for all $i \in [k]$
\[
    \|\hat \mu_i - \mu\|_2 \le \tilde O\left(\sqrt{\frac{d + \log \frac{k}{\beta}}{n}} + \frac{d + \log \frac{k}{\beta}}{\varepsilon n} \cdot \log k\right)
\]
However, the stronger notion of \emph{pan-privacy} is often studied in the the streaming setting. In this model, in addition to seeing all of the outputs, the adversary can view the hidden state kept by the algorithm at $m$ times (between processing of consecutive batches) of her choice. The commonly-studied problem is to estimate the empirical mean when the batch size is $1$ and $x_i$ are chosen arbitrarily from some bounded domain; in this setting, pan-privacy against $m$ intrusions can be shown to require accuracy loss that degrades with $\sqrt{m}$ \cite{dwork2010differential}. In contrast, we will show that in a stochastic setting---when $x_i$ are drawn i.i.d. from a Gaussian rather than chosen arbitrarily, and for sufficiently large batch sizes---pan-privacy can be assured against \emph{any} number of intrusions, since no hidden state is kept whatsoever between batches.

Although this task is frequentist as there is no prior on $\mu$, we will use our private Bayesian mean estimation algorithm through the following natural strategy. 
We will start with the uniform improper prior over $\mathbb R^d$ so that after each batch, conditioned on all samples seen thus far the posterior mean is simply the maximum likelihood estimate (MLE). At the start of every new batch, we have an estimate for the posterior conditioned on everything up to and including the previous batch, and we use this posterior as a prior upon which we perform a private Bayesian update using the samples in the new batch. Composition of the privacy-utility guarantee from Section~\ref{sec:gaussian-posterior-mean} across batches yields the following guarantee for private online mean estimation:

\begin{theorem}
    Let $\mu \in \mathbb R^d$ be a mean parameter, unknown to the algorithm, satisfying $\|\mu\| \le R$. Given a steam of $k$ batches of i.i.d. samples $x_i \sim \mathcal N(\mu, \Id_d)$ each of size $n>\frac{d \log \frac{R}{\alpha}}{\varepsilon}$, there is a computationally inefficient (resp. computationally efficient) $\varepsilon$-DP mechanism that outputs an estimate of $\mu$ after every batch such that, with probability $1-\beta$, for all $t \in [k]$, the $t$-th estimate has error at most
    \[
        \tilde O\left(\sqrt{\frac{d + \log \frac{k}{\beta}}{n t}} + \frac{d + \log \frac{k}{\beta}}{\varepsilon n}\right)\,,\ \ \ \mathrm{resp.} \ \ \ \tilde O\left(\sqrt{\frac{d + \log \frac{k}{\beta}}{n t}} + \left(\frac{d + \log \frac{k}{\beta}}{\varepsilon^{2/3} n}\right)^{3/4} + \frac{d + \log \frac{k}{\beta}}{\varepsilon n}\right)\,.
    \]
\end{theorem}

\begin{proof}
    In order to get a frequentist guarantee, we will start with the uniform improper prior (so after each iteration, the true posterior mean we are trying to estimate is just the maximum likelihood estimator). After each batch, we will use our previous posterior as a prior, and do a posterior mean estimate using the batch data. Assume as always that $n \ge \tilde \Omega\left(\frac{d \log \frac{R}{\alpha}}{\varepsilon}\right)$. We note that the posterior variance at step $t$ is independent of the samples, and is defined by
    \[
        \sigma_{t+1}^2 = \bigg(n+\frac{1}{\sigma_t^2}\bigg)^{-1}
    \]
    which (subject to $\sigma_0^{-2} = 0$) has solution
    \[
        \sigma_t^{-2} = n t
    \]
    The posterior mean update can be written as
    \begin{align*}
        \mu_{t+1} &= \mu_t + \frac{1}{1 + \frac{1}{\sigma_t^2 n}}(\bar x_t - \mu_t) \\
        &= \frac{1}{1 + \sigma_t^2 n} \mu_t + \frac{1}{1 + \frac{1}{\sigma_t^2 n}} \bar x_t
    \end{align*}
    Let $E_t$ be an upper bound on the error of the posterior mean (or equivalently in this case, maximum likelihood) estimate at time $t$, and suppose that we can estimate the empirical mean of $n$ samples in $d$ dimensions with $\varepsilon$-privacy and $\beta$-confidence to error $C$.
    Then we note that
    \begin{align*}
        E_{t+1} &\le \frac{1}{1 + \sigma_t^2 n}  E_t + \frac{1}{1 + \frac{1}{\sigma_t^2 n}} \cdot C \\
        &= \frac{t}{1 + t} E_t + \frac{1}{1 + t}  \cdot C
    \end{align*}
    which implies that $E_t \le C$ for all $t \ge 1$.
    Hence, using Theorem \ref{thm:anisotropic-private-est}, we can inefficiently (resp. computationally efficiently) estimate the MLE to error
    \[
        \alpha \le \tilde O\left(\frac{d + \log \frac{k}{\beta}}{\varepsilon n}\right)\,, \ \ \ \mathrm{resp.} \ \ \ \alpha \le \tilde O\left(\frac{d + \log \frac{k}{\beta}}{\varepsilon n} + \left(\frac{d + \log \frac{k}{\beta}}{\varepsilon^{2/3} n}\right)^{3/4}\right)\,.
    \]
    The theorem then follows by the convergence of the MLE to the parameter mean in the Gaussian model.
\end{proof}

%% file: content/LinearRegression.tex
\section{Robust and Private Bayesian Linear Regression}\label{sec:linear-regression}

In this section, we prove our main results for Bayesian linear regression. The model is as follows:
\begin{definition}[Bayesian regression under Gaussian prior]\label{def:bayesian-regression-under-gaussian-prior}
    A parameter $\wnull \in \R^d$ is drawn from a prior $\mathcal{N}(0,\Sigma)$, and then data is generated as follows. The columns of design matrix $\Xnull\in\R^{d\times n}$ are sampled i.i.d. from the distribution $\mathcal{N}(0,\Id_d)$, and the response vector $\ynull\sim \mathcal{N}(\Xnull^\top \wnull, \Id_n)$. Each sample is given by $(\Xnull(\cdot,i), \ynull(i))$, where $\Xnull(\cdot,i)$ denotes the $i$-th column of the design matrix $\Xnull$. Given these samples, our goal is to privately estimate the posterior mean $w_{\rm post} \defeq \E[\wnull\mid \Xnull, \ynull]$ to $\ell_2$ norm.

    We will also consider the setting in which an $\eta$-fraction of the samples have been arbitrarily corrupted. In this case, privacy is defined with respect to the corrupted samples, but the posterior mean that we wish to estimate is with respect to the uncorrupted samples, even though the algorithm only gets access to the corrupted ones.
\end{definition}

\noindent Our main positive results for this problem are algorithms in the special case where $\Sigma = \sigma^2\Id_d$. Specifically, we will prove the following slight strengthenings of \cref{thm:regression_upper_informal}.

\begin{theorem}[Statistical rate for private, robust Bayesian linear regression]\torestate{\label{lem:exponential-bayesian-error-regression}
 Let $\Sigma=\sigma^2 \Id_d$ for known parameter $\sigma > 0$.
 Let $\beta \in [0,1], \epsilon\leq O(1)$.
 Suppose $n\geq \max\Set{d\log(\sigma\sqrt{d})/\e, (d+\log(1/\beta))/\eta}$.
 Then there is a computationally inefficient algorithm that, given $\eta$-corrupted samples $(\Xinput,\yinput)$ generated as in \cref{def:bayesian-regression-under-gaussian-prior}, outputs an $\e$-DP estimator $\hat{w}$ such that with probability at least $1-\beta$,
 \begin{equation*}
     \norm{\hat{w}-\E[\wnull \mid \Xnull,\ynull]}_2 \leq \tilde{O}(\eta+\frac{d+\log(1/\beta)}{n\e})\,,
 \end{equation*}
 where $(\Xnull,\ynull)$ are the uncorrupted samples.
 }
\end{theorem}

\begin{theorem}[Computational rate for private, robust Bayesian linear regression]\label{thm:bayesian-error-regression}
Let $\Sigma=\sigma^2 \Id_d$ for known parameter $\sigma > 0$. 
Let $\beta \in [0,1]$, $\epsilon \leq O(1)$.
Suppose $n\geq \max\Set{d\log(\sigma\sqrt{d})/\e, (d+\log(1/\beta))/\eta}$.
There is a polynomial-time algorithm that, given $\eta$-corrupted samples $(X,y)$ generated as in \cref{def:bayesian-regression-under-gaussian-prior}, outputs an $\e$-DP estimator $\hat{w}$ such that with probability at least $1-\beta$,
 \begin{equation*}
     \norm{\hat{w}-\E[\wnull \mid \Xnull,\ynull]}_2 \leq \tilde{O}\biggl(\sqrt{\eta\sqrt{\frac{d+\log(1/\beta)}{n}}}+\eta+\frac{\Paren{d+\log(1/\beta)}^{3/4}}{n^{3/4}\e^{1/2}}+\frac{d+\log(1/\beta)}{n\e}\biggr) \,,
 \end{equation*}
 where $(\Xnull, \ynull)$ are the uncorrupted samples.
\end{theorem}

\noindent We note that the Bayes-optimal estimator achieves $\mathrm{MSE}(\hat{w})^{1/2}$, as defined in Eq.~\eqref{eq:LR_MSE}, of the order $(1+\frac{1}{n\sigma^2})^{-1}\sqrt{d/n}$.
It follows as a corollary that, under the condition that $\sqrt{d/n}\ll \eta+d/n\epsilon$ and $n\geq \Omega(d\log(\sigma\sqrt{d})/\epsilon)$, one achieves the Bayesian optimal error rate asymptotically.

There are three different regimes for $\sigma$:
\begin{itemize}[itemsep=0pt,topsep=0pt]
    \item \emph{Overly strong prior}: When $\sigma\ll 1/\sqrt{n}$, private and robust estimation in this setting is trivial. The prior is sufficiently strong that one can ignore the samples and still asymptotically achieve the Bayes-optimal MSE simply by outputting the mean of the prior. 
    \item \emph{Overly weak prior}: When $\sigma\gg 1/\sqrt{n}$, private and robust Bayesian estimation reduces to private and robust \emph{frequentist} estimation. The prior is sufficiently weak that the Bayes-optimal MSE is achievable asymptotically by ordinary least squares, though proving that this can be done optimally under privacy and robustness still takes some work.
    \item \emph{Critical prior}: The most interesting regime is when $\sigma=\Theta(1/\sqrt{n})$, in which case one needs to combine information from the prior with information from the samples to compete with the Bayes estimator.
\end{itemize}

In this section we will focus on computationally efficient algorithms, deferring our proof of the statistical rate in \cref{lem:exponential-bayesian-error-regression} to \cref{app:exponential-regression}. In \cref{sec:comp_LR}, we consider the critical prior regime, and in \cref{sec:approxOLS}, we consider the weak prior regime.

Additionally, we provide low-degree evidence that the computational rate is optimal. 

\begin{theorem}[Information-computation gap for private Bayesian linear regression]\torestate{\label{thm:privacy-lower-bound-regression}
Suppose $\e\leq O(1)$ and $\alpha \leq \min\Set{\eta d^{1/4}/D,\eta^2}$.
    There is a distinguishing problem $\Pi$ (see \cref{def:variance-matched-xcorruption-test}) which is degree-$D$ hard such that if there were a polynomial-time $\epsilon$-DP algorithm for Bayesian linear regression with standard Gaussian design and $\mathcal{N}(0,\sigma^2\Id)$ prior that estimates the posterior mean to error $\alpha$ with probability $1 - \beta$ using
    \begin{equation*}
        n = o\biggl(\frac{d^{1/3}\log^{2/3}(1/\beta)}{\alpha^{4/3}\epsilon^{2/3}\cdot D^{1/3}}\biggr)
    \end{equation*}
    samples, then there would be an efficient distinguishing algorithm for $\Pi$.
    }
\end{theorem}

\noindent Note that when the failure probability $\beta$ is exponentially small in $d$, this matches the rate achieved by \cref{thm:bayesian-error-regression}, up to log factors. In fact, even for general $\beta$, we obtain a \emph{bona fide} information-computation gap (\cref{remark:gap}).

While the proof of the above lower bound also draws upon non-Gaussian component analysis like the one for Bayesian mean estimation, the former is more involved as the embedding of non-Gaussian component analysis into Bayesian linear regression is more delicate. 
We defer the details to \cref{sec:robust-lr-lb} and \cref{sec:private-lr-lb}.

\subsection{Computational rate for private Bayesian linear regression with critical prior}
\label{sec:comp_LR}

Our main guarantees in the critical prior regime are as follows:

\begin{theorem}[Computational rate under critical prior]\torestate{\label{thm:private-bayesian-linear-regression}
Let $\Sigma=\sigma^2 \Id_d$ for known parameter $\sigma = \Theta(1/\sqrt{n})$. There is a polynomial-time algorithm that, given $\eta$-corrupted samples $(X,y)$ generated as in \cref{def:bayesian-regression-under-gaussian-prior}, outputs an $\epsilon$-DP estimator $\hat{w}$ such that with probability at least $1-\beta$, \begin{equation*}
    \norm{\hat{w}-\E[\wnull\mid \Xnull,\ynull]}_2 \leq O(\alpha)\,,
 \end{equation*}
whenever $\alpha\geq \eta \log(1/\eta)$, $\e\leq O(1)$  and  \begin{equation*}
    n\geq \tilde{\Omega}\Paren{\frac{d+\log(1/\beta)}{ \alpha^{4/3}\e^{2/3}}+\frac{d+\log(1/\beta)}{\alpha\epsilon}+\frac{d\log(\sigma\sqrt{d})}{\epsilon}+\frac{(d+\log(1/\beta))\eta^2}{\alpha^4}}\,.
    \end{equation*}
    }
\end{theorem}

We first give robust, non-private algorithms for estimating the posterior mean, and then apply the privacy-to-robustness reduction.

\subsubsection{Efficient robust regression with critical prior}
\label{sec:robust_LR_crit}

Here we give a robust, non-private algorithm for estimating the posterior mean when $\sigma=\Theta(\sqrt{1/n})$.

We will design a sum-of-squares estimator for $(\Sigma^{-1}+n \Id_d)^{-1} \Xnull \ynull$ and then show that this is close to $\wpost$ by concentration.
As discussed in \cref{sec:overview}, the key difference compared to our algorithm for Bayesian mean estimation is enforcing the existence of a short-flat decomposition for our estimate $y$ for the true responses $\ynull$.

Let $\Xinput, \yinput$ denote the observed data, an $\eta$ fraction of which have been corrupted. Let $\Xnull, \ynull$ denote the original, uncorrupted data, which we do not have access to. We consider the following set of polynomial constraints in the indeterminates $\xi = (\xi(1),\ldots,\xi(n))$ (indicating our guesses for which points are corrupted), and $X,y$. First, we have the constraints indicating that the $\xi(i)$'s try to select for corrupted points:
\begin{equation}\label{eq:corruption_constraints}
\cA_{\mathrm{corr}}(X,y,\xi;\Xinput,\yinput)
= \left\{
\begin{aligned}\,& \xi\odot\xi=\xi,\quad \|\xi\|_1 \le \eta n,\\
& (1-\xi(i))\,X(i,\cdot)=(1-\xi(i))\,\Xinput(i,\cdot),\ \forall i\in[n],\\
& (1-\xi(i))\,y(i)=(1-\xi(i))\,\yinput(i),\ \forall i\in[n]\end{aligned}\right\}\,.
\end{equation}
Next, we have the resilience constraints that enforce that our estimates $X$ for the uncorrupted covariates have bounded covariance, and that our estimates $y$ for the uncorrupted responses admit a short-flat decomposition into a vector of indeterminates $z_1$ with small $\ell_2$ norm and a vector $z_2$ with small $\ell_\infty$ norm:
\begin{equation}
\cAres(X,y)=
\left\{
\begin{aligned}
&\Normop{\frac{1}{n}XX^\top-\Id}
   \le \chi \sqrt{\frac{d+\log(1/\beta)}{n}},\\
&y =z_1 + z_2\\
&\|z_1\|^2_2 \le \chi\eta n \log(1/\eta)+\chi\log(1/\beta),\\
&\eta n\|z_2\|_\infty^2 \le \chi\eta n \log(1/\eta)+\chi\log(1/\beta)
\end{aligned}
\right\}\,.
\end{equation}
Here $\chi > 0$ is some fixed universal constant.

Finally, we have a constraint defining a variable that corresponds to our estimate for the posterior mean:

\begin{equation}\label{eq:fine-grained-post-constraint}
\cA_\post(X,y,w_{\out})=\Set{w_{\out}=(1/\sigma^2 + n)^{-1}Xy}\,.
\end{equation}
With these constraints in place, we present our algorithm in \cref{algo:posterior-mean-regression} below.
\begin{algorithmbox}[Robust algorithm for approximating posterior mean estimator]\label{algo:posterior-mean-regression}
    \mbox{}\\
    \textbf{Input:} $\eta$-corrupted design matrix $X_{\mathrm{input}}\in \R^{n\times m}$ and responses $y_{\mathrm{input}}\in \R^n$\\
    \textbf{Output:} a polynomial time estimate $\hat{w}\in \R^d$ for the posterior mean
\begin{enumerate}
    \item Solve the degree-$O(1)$ SoS relaxation of the constraints
    \[
        \cA(X,y,\xi,w;\Xinput,\yinput)
        =\cA_{\mathrm{corr}}(X,y,\xi;\Xinput,\yinput)\cup \cAres(X,y,w)\cup \cA_\post(X,y,w_\out),
    \]
    obtaining a pseudo-expectation operator $\tE$.
    \item Set $\hat{w}\gets \tE[w_\out]$.
    \item \textbf{return} $\hat{w}$.
\end{enumerate}
\end{algorithmbox}

\noindent First, we verify feasibility of the program $\cA$. The proof ingredients for this are given in \cref{sec:concentration}.

\begin{lemma}\label{lem:feasibility-proof-isotropic-regression}
    The program constraints $\cA(X,y,\xi,w;\Xinput,\yinput)$ are satisfied by the clean samples with probability at least $1-\beta$. 
\end{lemma}
\begin{proof}
    The corruption constraints $\cA_{\rm corr}$ are satisfied by definition. 
    For the resilience constraints $\cA_{\rm res}$, by \cref{thm:covariance-concentration}, with probability at least $1-\beta$, the bound on $\Normop{\frac{1}{n}\Xnull \Xnull^\top-\Id}$ holds.
    On the other hand, the vector $y$ is i.i.d Gaussian with variance bounded by $2$. 
     We can let $z_1$ represent the $\eta n$ entries in $y$ with largest absolute values, and $z_2$ the remaining entries.
    By \cref{thm:gaussian-order-stat-sharp}, we have
    \begin{equation*}
        \norm{z_2}_\infty^2\leq \chi\log(1/\eta)+\frac{\chi\log(1/\beta)}{\eta n}\,.
    \end{equation*}
    By \cref{thm:maximum-subset-sum-gaussian-correct}, we have 
    \begin{equation*}
        \norm{z_1}^2_2 \leq \chi\eta n\log(1/\eta)+\chi\log(1/\beta)\,.
    \end{equation*}
    Therefore the resilience constraint is satisfied with probability at least $1-\beta$. 
\end{proof}

\noindent Next, we show that our algorithm outputs an accurate estimate of $(1/\sigma^2 + n)^{-1}\Xnull\ynull$:

\begin{lemma}\label{lem:identifiability-proof-isotropic-regression}
    Consider the setting of \cref{thm:private-bayesian-linear-regression}. 
    Then with probability at least $1-\beta$, \cref{algo:posterior-mean-regression} outputs $\hat{w}\in \R^d$ such that 
    \begin{equation*}
        \|\hat{w}-(1/\sigma^2+n)^{-1}\Xnull\ynull\|^2_2\leq O\Paren{\eta^2\log^2(1/\eta)+\eta \log(1/\eta)\sqrt{\frac{d+\log(1/\beta)}{n}}}\,.
    \end{equation*}
\end{lemma}
\begin{proof}
    Let $r\in \Set{0,1}^n$ be the true indicator for the corrupted samples. 
    Then in the support of $(1-\xi)\odot (1-r)$, the program variables $X,y$ are identical to the ground truth samples, i.e
    \begin{align*}
         X((1-\xi)\odot (1-r)) &=\Xnull((1-\xi)\odot (1-r))\\
        (1-\xi)\odot (1-r)\odot y &=(1-\xi)\odot (1-r) \odot \ynull\,.
    \end{align*}
    We have 
    \begin{equation*}
        Xy-\Xnull \ynull=X(y\odot v)-\Xnull (\ynull\odot v)\,.
    \end{equation*}
    where $v=1-(1-r)\odot (1-\xi)$ indicates the points which are truly corrupted and/or were marked by $\xi$ as corrupted.

    Now we bound the two terms on the right-hand side separately. 
    For every unit vector $u$, using Cauchy-Schwarz, we have 
    \begin{equation*}
        \iprod{X(y\odot v),u}^2\leq \norm{Xu\odot v}^2\norm{y\odot v}^2\,. 
    \end{equation*}
    Now by the constraint $\cAres(X,y,\xi)$, we have 
\begin{equation*}
    \norm{y\odot v}^2\lesssim \norm{z_1\odot v}^2+\norm{z_2\odot v}^2\leq \norm{z_1}^2+\eta n\norm{z_2}_\infty^2\lesssim \eta n \log(1/\eta)+\log(1/\beta)\,.
\end{equation*}
    Moreover, similar to the setting of empirical mean estimation where we invoked the isotropic estimation lemma (\cref{lem:technicallem-unit-variance}), 
    we have
    \begin{equation*}
        \Norm{Xu\odot v}^2=\sum_i v_i\iprod{X(\cdot,i),u}^2\leq \eta n \log(1/\eta)+\sqrt{\frac{d + \log(1/\beta)}{n}}n\,.
    \end{equation*}
    We can bound  $\iprod{\Xnull(\ynull\odot v),u}^2$ similarly, as the ground truth is a feasible solution. As a result, for every direction $u$, we have a degree-$O(1)$ sum-of-squares proof that
\begin{equation*}
        \iprod{Xy-\Xnull \ynull,u}^2\lesssim \paren{\eta n\log(1/\eta)+\log(1/\beta)}\cdot\Paren{\eta n \log(1/\eta)+\sqrt{\frac{d + \log(1/\beta)}{n}}n}\,.
    \end{equation*}
    Taking pseudo-expectation of $w_\out=(1/\sigma^2+n)^{-1}Xy$, we get an estimator with the claimed error rate: 
    \begin{align*}
        \alpha&\leq O\Paren{\eta \log(1/\eta)+\sqrt{\eta\log(1/\eta)\sqrt{\frac{d+\log(1/\beta)}{n}}}}\,. \qedhere
    \end{align*}
\end{proof} 

\noindent Finally, we show that our estimator is close to the posterior mean by concentration.
\begin{lemma}\label{lem:Lipshitz-least-square}
    Consider the setting of \cref{thm:private-bayesian-linear-regression}.
    With probability at least $1-\beta$, we have 
    \begin{equation*}
        \norm{\E[\wnull\mid \Xnull,\ynull]-(\Sigma^{-1}+n\Id_n)^{-1}\Xnull \ynull}\lesssim \frac{d+\log(1/\beta)}{n} \,.
\end{equation*}
\end{lemma}
 
\begin{proof}
    The posterior mean of our regression model is given by 
    \begin{equation*}
        \E[\wnull\mid \Xnull,\ynull]=\Paren{\Xnull \Xnull^\top+\Sigma^{-1}}^{-1} \Xnull \ynull\,.
    \end{equation*}
    By \cref{thm:covariance-concentration}, with probability $1-\beta$ we have
    \begin{equation*}
        \Normop{\frac{1}{n}\Xnull \Xnull^\top-\Id_n}\leq \sqrt{\frac{d+\log(1/\beta)}{n}}
    \end{equation*}
    Therefore, as a result we have
\begin{equation*}
\Normop{n(\Sigma^{-1}+n\Id_n)^{-1}- n(\Sigma^{-1}+\Xnull \Xnull^\top)^{-1}}\lesssim \sqrt{\frac{d+\log(1/\beta)}{n}}\,.
\end{equation*}
 Since with probability at least $1-\beta$ we have $\Norm{\frac{1}{n}\Xnull \ynull}\leq O(\sqrt{\frac{d+\log(1/\beta)}{n}})$, the claim thus follows.
\end{proof}

\begin{proof}[Proof of \cref{thm:private-bayesian-linear-regression}]
    Combining \cref{lem:Lipshitz-least-square} and \cref{lem:identifiability-proof-isotropic-regression} with triangle inequality, we know that with probability at least $1-\beta$, we have 
         \begin{equation*}
        \norm{\hat{w}-(1/\sigma^2+\Xnull\Xnull^\top)^{-1}\Xnull\ynull}^2\leq O\Paren{\eta^2\log^2(1/\eta)+\eta\log(1/\eta)\sqrt{\frac{d+\log(1/\beta)}{n}}}\,.
    \end{equation*}
Converting to sample complexity bound, we have the claim. 
\end{proof}

\subsubsection{Robustness-to-privacy reduction}
\label{sec:robust_to_private_LR}

In this section we apply the robustness to privacy reduction to prove \cref{thm:private-bayesian-linear-regression}. We will privatize our \cref{algo:posterior-mean-regression}.
We make the following definition for certifiable posterior mean for regression.
\begin{definition}[(\(\alpha,\tau,T\))-certifiable posterior mean for regression]\label{def:cert}
Fix $R>0$ and $\alpha,\tau>0$. For an input $(X_{\mathrm{in}},y_{\mathrm{in}})$, a vector $\widetilde w\in\R^d$
is \emph{$(\alpha,\tau,T)$-certifiable} if there exists a degree-$O(1)$ linear functional $L$ (a pseudoexpectation) over indeterminates $(X,y,\xi,w,\wpost,M)$ such that:
\begin{enumerate}[leftmargin=2.1em]
\item \textbf{Closeness/boundedness:}
$\norm{L[\wpost]-\widetilde w}_2\le \alpha$ and $\norm{L[\wpost]}_2\le 2R+\tau T$.
\item \textbf{Approximate feasibility (degree-6):} $L$ $\tau$-approximately satisfies the union of constraint families
\[
\mathcal A \;=\; \underbrace{\mathcal A_{\mathrm{corr}}}_{\text{corruption}}\ \cup\ 
\underbrace{\mathcal A_{\mathrm{res}}}_{\text{design/residual stability}}\ \cup\ 
\underbrace{\mathcal A_{\mathrm{post}}}_{\wpost=(1/\sigma^2+n)^{-1}Xy}\,,
\]
concretely:
\begin{itemize}[leftmargin=1.5em]
\item $\mathcal A_{\mathrm{corr}}$: $\xi\odot\xi=\xi$, $\|\xi\|_1\le \eta n$ and $(1-\xi(i))X(i,\cdot)=(1-\xi(i))X_{\mathrm{in}}(i,\cdot)$, $(1-\xi(i))y(i)=(1-\xi(i))y_{\mathrm{in}}(i)$ for all $i$.
\item $\mathcal A_{\mathrm{res}}$: 
$\big\|\frac{1}{n}XX^\top-I\big\|_{\op}\lesssim \sqrt{(d+\log(1/\beta))/n}$ and a decomposition $y=z_1+z_2$ with
$\|z_1\|_2^2\lesssim \eta n\log(1/\eta)+\log(1/\beta)$ and $\|z_2\|_\infty^2\lesssim \log(1/\eta)+\log(1/\beta)/(\eta n)$.
\item $\mathcal A_{\mathrm{post}}$: $\wpost=(1/\sigma^2+n)^{-1}Xy$.
\end{itemize}
\end{enumerate}
We choose $R=\Theta(\sigma\sqrt d)$ and take the slack inside $\mathcal A_{\mathrm{res}}$ to be
\[
\alpha_0(\eta)\;\asymp\; \sqrt{\frac{d+\log(1/\beta)}{n}}\ +\ \eta\log1/\eta\,.
\]
\end{definition}

\begin{definition}[Regression score]\label{def:regression-score-1}
For $\widetilde w\in \ball{d}{\,2R+n\tau+\alpha\,}$, define
\[
\mathcal{S}(\widetilde w;X_{\mathrm{in}},y_{\mathrm{in}},\alpha,\tau)
\;\defeq\; \min\bigl\{T\ge 0:\ \widetilde w\ \text{is $(\alpha,\tau,T)$-certifiable for }(X_{\mathrm{in}},y_{\mathrm{in}})\bigr\}.
\]
\end{definition}

\noindent Now we bound the volume ratio for different score values.

\begin{lemma}[Low-score sets localize near $\wpost$]\label{lem:vol}
There exists a universal $\eta_\star>0$ such that, on the stability event from the resilience constraints (which holds with prob.\ $1-\beta$), the following hold simultaneously. Let
\[
\alpha(\eta)^2\ \lesssim\ \eta^2\log^2 1/\eta +\ \eta\log(1/\eta)\cdot \sqrt{\frac{d+\log(1/\beta)}{n}}\,.
\]
\begin{enumerate}[leftmargin=2.1em]
\item \textbf{Completeness (inner ball).} If $\norm{\widetilde w-\wpost}_2\le \alpha(\eta)$ then $\mathcal{S}(\widetilde w;\cdot)\le \eta n$.
\item \textbf{Soundness (outer ball).} If $\mathcal{S}(\widetilde w;\cdot)\le \eta_\star n$ then $\norm{\widetilde w-\wpost}_2\lesssim \alpha(\eta_\star)$.
\end{enumerate}
\end{lemma}

\begin{proof}
For (1), use the witness that puts point mass on the uncorrupted assignment consistent with $(X_{\mathrm{in}},y_{\mathrm{in}})$; the constraints are feasible and $L[\wpost]=\wpost$.
For (2), combine the SoS identifiability for regression under $\mathcal A_{\mathrm{res}}$ with the closeness constraint of \cref{def:cert}.
\end{proof}

\noindent We need to prove that our score function has value bounded by $n$ for all candidate points $\tilde{w}$. 

\begin{lemma}[Well-definedness of the score function]\label{lem:regression-satisfiability-parameter}
    For any $\tilde{w}\in \R^d$, the value of score function is bounded by $n$.  
\end{lemma}
\begin{proof}
For feasibility, we note that we can take $y=\mathbf{1}$, and the rest of the proof is similar to \cite[Lemma 6.15]{anderson2025sample}.
In particular, when the norm of $\tilde{w}$ is bounded by $2R+\tau T$, we can set $X_i=\Paren{\frac{1}{n\sigma^2}+1}^{-1}\tilde{w}$. 
Otherwise we project $\tilde{w}$ to the ball of radius $2R+\tau T-\phi$ for some $\phi$, and sample from the distribution $N(\bar{w},\Id)$, where $\bar{w}$ is the projection. 
\end{proof}

\begin{lemma}\label{lem:quasi-convexity-regression-1}
    The score function defined in \cref{def:regression-score-1} is quasi-convex.
\end{lemma}
\begin{proof}
    Most of the constraints are similar to the program for empirical mean estimation, except for the resilience constraints on $y$.
    However, by the property of pseudo-distribution, it is easy to see that these convex constraints admit quasi-convexity.
\end{proof}

\noindent Now we conclude the proof of the \cref{thm:private-bayesian-linear-regression} with utility analysis.

\begin{proof}[Proof of \cref{thm:private-bayesian-linear-regression}]
Apply the robustness-to-privacy reduction to the score $\mathcal{S}$ on the domain $\ball{d}{\,2R+n\tau+\alpha\,}$.
Let $V_\eta$ denote the volume of $\{\widetilde w:\ \mathcal{S}(\widetilde w)\le \eta n\}$ inside this domain.
By \cref{lem:vol}:
\begin{itemize}
    \item \emph{Local regime} $\eta\le \eta'\le \eta_\star$: Low-score sets are sandwiched between $\ell_2$-balls around $\wpost$ with radii $\alpha(\eta)$ and $\alpha(\eta')$, hence
    \[
    \log\!\frac{V_{\eta'}}{V_\eta}\ \lesssim\ d\cdot \log\!\frac{\alpha(\eta')}{\alpha(\eta)}.
    \]
    
    \item \emph{Global regime} $\eta_\star\le \eta'\le 1$: The score-$\le \eta_\star n$ set lies in an $O(\alpha(\eta_\star))$-ball while the ambient radius is $O(R+n\tau)$, so
    \[
    \log\!\frac{V_{\eta'}}{V_\eta}\ \lesssim\ d\cdot \log\!\frac{R+n\tau}{\alpha(\eta_\star)}.
    \]
\end{itemize}
Finally it is easy to verify that the score function is quasi-convex and can be efficiently evaluated. The reduction in \cref{thm:efficient-reduction-meta-pure} then yields a sufficient condition
\[
n\ \gtrsim\ 
\max_{\eta\le\eta'\le\eta_\star}\frac{d\log(\alpha(\eta')/\alpha(\eta))+\log(1/\beta_\eta)}{\varepsilon \eta'}
\;+\;
\max_{\eta_\star\le\eta'\le 1}\frac{d\log\!\big((R+n\tau)/\alpha(\eta_\star)\big)+\log(1/\beta_\eta)}{\varepsilon}.
\]
Optimizing over $\eta$ using $\alpha(\eta)^2\asymp \eta\cdot\frac{d+\log(1/\beta)}{n}$ in the small-$\eta$ regime yields the two ``local'' terms
$\tfrac{d+\log(1/\beta)}{\alpha^{4/3}\varepsilon^{2/3}}$ and $\tfrac{d+\log(1/\beta)}{\alpha\varepsilon}$.
For the global term, take $R=\Theta(\sigma\sqrt d)$ so that the contribution is $\tfrac{d\log(\sigma\sqrt d)}{\varepsilon}$, up to log factors.
\end{proof}

\begin{remark}[On $\sigma$ and transferring from $\wpost$]\label{rem:sigma}
The $\tfrac{d\log(\sigma\sqrt d)}{\varepsilon}$ term is the regression analogue of the global-volume term in mean estimation and arises purely from the ratio between the ambient ball and the localized score set.
In the regime $\sigma=\Omega(n^{-1/2})$, $R=\Theta(\sigma\sqrt d)$ matches the typical scale of $\|\wpost\|_2$.
Standard concentration shows $\wpost=(\frac{1}{\sigma^2}+n)^{-1}Xy$ is within $\widetilde{O}(\sqrt{(d+\log(1/\beta))/n})$ of the true Bayesian posterior mean, which is absorbed into~$\alpha$ for the stated range. 
\end{remark}

\subsection{Computational rate for private Bayesian linear regression with weak prior}
\label{sec:approxOLS}

Our main guarantee in the weak prior regime is as follows:

\begin{theorem}\torestate{\label{thm:private-approximation-least-square}
    Let $\Sigma = \sigma^2\Id_d$ for known parameter $\sigma = \omega(1/\sqrt{n})$. Let $\beta\in[0,1]$, $\epsilon \le O(1)$. There is a polynomial-time algorithm that, given $\eta$-corrupted samples $(X,y)$ generated as in \cref{def:bayesian-regression-under-gaussian-prior}, outputs an $\epsilon$-DP estimator $\hat{w}$ such that with probability at least $1-\beta$,
    \begin{equation*}
        \Norm{\hat{w}-\E[\wnull\mid \Xnull,\ynull]}^2\leq \alpha^2\,,
    \end{equation*}  
    where $(\Xnull, \ynull)$ are the uncorrupted samples, under the conditions that $\alpha\geq \eta \log(1/\eta)$ and
    \begin{equation*}
        n\geq \tilde{\Omega}\,\Paren{\frac{d+\log(1/\beta)}{\alpha^{4/3}\e^{2/3}}+\frac{d+\log(1/\beta)}{\alpha\epsilon}+\frac{d\log(\sigma\sqrt{d})}{\epsilon}+\frac{d\eta^2}{\alpha^4}}\,.
    \end{equation*}}
\end{theorem}

\noindent 
The argument is more involved than in the critical prior setting, and we defer some of the details to \cref{app:regression}.

In the regime of $\sigma = \omega(1/\sqrt{n})$, the prior is sufficiently weak that the Bayes-optimal error is asymptotically achieved by the OLS estimator. 
We design a \emph{multi-stage} algorithm.
In the first stage, we obtain a rough estimator with error rate $O(1)$. 
In the second and the third stage, we refine the estimator to the target level of accuracy.

\subsubsection{Overview of multi-stage algorithm}
\label{sec:weakprioroverview}

Here we overview our approach in the weak prior regime. For starters, without any \emph{a priori} upper bound on the variance of the true regressor $w$, the strategy in the critical prior regime cannot be directly applied as we have no control over the magnitudes of the labels $y$. The first step is thus to ``re-center'' the labels by learning a \emph{coarse estimate} of $w$ (\cref{sec:stage1}). This can be done again with the short-flat decomposition constraint, but this time imposed not upon the label vector $\hat{y}$, but on the \emph{residual vector} $\hat{y} - \hat{X}^\top\hat{w}$. An argument similar to the above can be used to show that this will yield an estimate $w_{\rm init}$ of the true $w$ which is accurate to within $O(1)$ error in $L_2$.

After the coarse estimation stage, intuitively we can now ``subtract out'' our coarse estimate $w_{\rm init}$ from $w$ to reduce the problem to something resembling the proportional regime. This runs into several obstacles. 

First, we do not want to directly subtract out $\tilde{X}w_{\rm init}$ from $\tilde{y}$ because of the corruptions in $\tilde{X}$, so instead we indirectly fold our coarse estimate into a new sum-of-squares program and search for short-flat decompositions of both $\hat{y} - \hat{X}^\top \hat{w}$, where $\hat{w}$ is the SoS variable for the regressor, as well as $X^\top (\hat{w} - w_{\rm init})$. 

Second, it turns out that even this is only enough to refine the $O(1)$ error achieved by $w_{\rm init}$ to error $\tilde{O}(\sqrt{d/n})$ in $L_2$ for estimating the posterior mean.
This is because the spectral norm of $\frac{1}{n}XX^\top-\Id$ is only bounded by $O(\sqrt{d/n})$ with high probability. 
Therefore, naively one can only show that $\frac{1}{n}Xy$ is within $L_2$ distance $\sqrt{d/n}$ to the least square estimator $(XX^\top)^{-1} Xy$. 
Consequently, given the vector $w_1$ achieving this latter error bound, the above program which looks for short-flat decompositions of both labels and residuals needs to be re-run, with $w_{\rm init}$ replaced by $w_1$.
Particularly, because of the refined error bound from $w_1$, one can now show that the vectors $v_1\seteq (n+\frac{1}{\sigma^2})^{-1} X(y-X^\top w_1)$ and $v_2\seteq (XX^\top+\frac{1}{\sigma^2})^{-1} X(y-X^\top w_1)$ are close in the sense that with high probability 
    \begin{equation*}
         \Norm{v_1-v_2} \\
         \leq O\Paren{\sqrt{\eta\log 1/\eta}+\sqrt{d/n}}\cdot \Paren{\sqrt{d/n}}\,.
    \end{equation*}
This will enable us to prove the desired error bound of $\tilde{O}(\eta + \sqrt{\eta\sqrt{d/n}})$. We defer the details of these latter two SoS programs to \cref{sec:stage2}. 

Finally, we note that because we are working in a Bayesian setting, the three SoS programs that are used in sequence to obtain the final posterior mean estimate are necessarily run on the \emph{same dataset}. This might appear to be problematic as the latter two SoS programs are adaptively constructed: some of the constraints therein depend on estimates $w_{\rm init}, w_1$ computed in preceding SoS programs, and this data reuse could potentially break feasibility. As we will see (\cref{lem:program-satisfiability-fine-regression}), this turns out not to be an issue, essentially because we can use a bound on the ``sparse singular values'' of the design matrix $X$ to argue that regardless of $w_{\rm init}, w_1$, as long as they are $O(1)$-close to $w$, the vectors $X^\top(w - w_{\rm init})$ and $X^\top(w - w_1)$ admit a short-flat decomposition.

\subsubsection{Stage 1: Rough estimation}
\label{sec:stage1}

In the first stage, we will use a sum-of-squares program that bears resemblance to the one used in the critical prior regime, with the following key differences.

First, the corruption constraints $\cA_{\rm corr}$ are the same as before (\cref{eq:corruption_constraints}). The resilience constraints are similar, except instead of insisting on a short-flat decomposition for the labels $y$, we insist on a short-flat decomposition for the residual $y - X^\top w$:
\begin{equation}\label{eq:resilience_constraints_rough}
\cAres'(X,y,w)=
\left\{
\begin{aligned}
&\Normop{\frac{1}{n}XX^\top-\Id}
   \le \chi \sqrt{\frac{d+\log(1/\beta)}{n}},\\
&y-X^\top w = z_1 + z_2, \norm{z_1+z_2}^2\le 2n\\
&\|z_1\|^2 \le \eta n \log(1/\eta)+\log(1/\beta),\\
&\|z_2\|_\infty^2 \le \chi \log(1/\eta)+\log(1/\beta)/\eta n
\end{aligned}
\right\}\,.
\end{equation}

\noindent Formally, we have the following guarantee showing that the true parameter $\wnull$ can be recovered to within \emph{constant error}:

\begin{lemma}[Rough estimation under weak prior]\torestate{\label{lem:robust-regression-constant-error}
    Let $\Sigma = \sigma^2\Id_d$ for known parameter $\sigma = \omega(1/\sqrt{n})$. There is a polynomial-time algorithm that, given $\eta$-corrupted samples $(X,y)$ generated as in \cref{def:bayesian-regression-under-gaussian-prior}, outputs $\hat{w}$ such that with probability at least $1 - \beta$, we have $\norm{\hat{w} - \wnull}_2 \le O(1)$ provided that $n \ge \Omega((d + \log (1/\beta))/\eta)$.}
\end{lemma}

\noindent The algorithm will be to solve the degree-$O(1)$ SoS relaxation of the program $\cA_{\rm corr} \cup \cAres'$ and naively round the resulting pseudo-distribution (see \cref{algo:constant-error-regression}). We defer the details of the SoS identifiability proof to \cref{sec:robust-regression-constant-error}.

\subsubsection{Stage 2: Refinement}
\label{sec:stage2}

In this next stage, we will run two SoS programs in sequence. The first will take as a parameter, call it $w_\init$, the estimate from the previous stage. This will result in a new estimate $\hat{w}_1$ which achieves error $O(\sqrt{\eta\log 1/\eta} + \sqrt{\frac{d + \log(1/\beta)}{n}})$ in estimating the true parameter $\wnull$. We will pass $\hat{w}_1$ into a final SoS program which will then refine it to produce an estimate $\hat{w}_2$ with the desired level of error for approximating the posterior mean.

As in Stage 1, we use the same corruption constraints $\cA_{\rm corr}$ from \cref{eq:corruption_constraints}. For the resilience constraints, for the first program, in addition to insisting on a short-flat decomposition for the residual $y - X^\top w$, we also insist on one for the difference between the labels $X^\top w_\init$ predicted by the initialization $w_\init$ versus under the labels $X^\top w$ predicted by the SoS variable $w$:
\begin{equation}
\cAres''(X,y; w_{\init})=
\left\{
\begin{aligned}
&\Normop{\frac{1}{n}XX^\top-\Id}
   \le \chi \sqrt{\frac{d+\log(1/\beta)}{n}},\\
&y=X^\top w+z,
z=z_1 + z_2, \norm{z_1+z_2}^2\le 2n\\
&\|z_1\|^2 \le \eta n \log(1/\eta)+\log(1/\beta), \\
&\|z_2\|_\infty^2 \le \chi \log(1/\eta)+\log(1/\beta)/\eta n\\
& X^\top (w-w_\init)=b_1+b_2, \\
& \|b_1\|^2 \le \eta n \log(1/\eta), \|b_2\|_\infty^2 \le \chi \log(1/\eta)
\end{aligned}
\right\}.
\end{equation}
Similarly, for the second program, we replace $w_\init$ with $\hat{w}_1$ to get $\cAres''(X,y; \hat{w}_1)$.

For the first program, we add a constraint for the definition of the least square estimator:
\begin{equation}\label{eq:ls-constraint}
\cA_\LS(X,y,w,w_{\LS})=\Set{w_{\LS}=\frac{1}{n}Xy, \norm{w-w_{\LS}}^2\le \frac{\chi(d+\log(1/\beta))}{n}}\,,
\end{equation}
and for the second program, we will use the constraint for the definition of the (simplified) posterior mean from \cref{eq:fine-grained-post-constraint}. 

Formally, we have the following guarantee showing that the posterior mean can be estimated to within the desired level of error:

\begin{lemma}\torestate{\label{lem:robust-regression-least-square}
    Let $\Sigma = \sigma^2\Id_d$ for known parameter $\sigma = \omega(1/\sqrt{n})$. There is a polynomial-time algorithm that, given $\eta$-corrupted samples $(X,y)$ generated as in \cref{def:bayesian-regression-under-gaussian-prior}, and given rough estimate $w_\init$ satisfying $\norm{w_\init - \wnull} \le O(1)$, outputs $\hat{w}$ such that with probability at least $1 - \beta$, we have
    \begin{equation*}
        \Norm{\hat{w}-\E[\wnull \mid \Xnull,\ynull]}^2\leq O\Paren{\eta\log(1/\eta)\sqrt{\frac{d+\log(1/\beta)}{n}}+\eta^2\log^2(1/\eta)}
    \end{equation*} 
    provided $n\geq (d+\log(1/\beta))/\eta$.
}
\end{lemma}

\noindent The full algorithm combining Stages 1 and 2 is given in \cref{algo:private-least-square-regression} below:

\begin{algorithmbox}[Robust algorithm for approximating posterior mean estimator]\label{algo:private-least-square-regression}
    \mbox{}\\
    \textbf{Input:} $\eta$ corrupted matrix $X_{\mathrm{input}}\in \R^{n\times m},y_{\mathrm{input}}\in \R^n$  \\
    \textbf{Output:} a polynomial time estimator $\hat{w}\in \R^d$
\begin{enumerate}
\item Run \cref{algo:constant-error-regression}, for obtaining estimator $w_\init$.
  \item Solve the degree-$O(1)$ SoS relaxation of the constraints $\cA(X,y,\xi,w;\Xinput,\yinput,w_\init)=\cA_{\mathrm{corr}}(X,y,\xi; X_\mathrm{input}, y_\mathrm{input})\cup \cAres''(X,y,w;w_\init)\cup \cA_\LS(X,y,w_{\LS})$, where $t = O(1)$,
    obtaining a pseudo-expectation operator $\tE$. \label{line:program1}
  \item Round the pseudo-distribution for $\hat{w}_1=\tE[w]$. \label{line:program1_program}
  \item Solve the degree-$O(1)$ SoS relaxation of the constraints $\cA(X,y,\xi,w;\Xinput,\yinput,\hat{w}_1)=\cA_{\mathrm{corr}}(X,y,\xi; X_\mathrm{input}, y_\mathrm{input})\cup \cAres''(X,y,w;\hat{w}_1)$, where $t = O(1)$, 
    obtaining a pseudo-expectation operator $\tE$. 
    \label{line:program2}
    \item Round the pseudo-distribution for $\hat{w}_2=\tE[(n+\frac{1}{\sigma^2})^{-1}X(y-X\hat{w}_1)]$ and return $(n+\frac{1}{\sigma^2})^{-1}\hat{w}_1+\hat{w}_2$.   \label{line:program2_program}
\end{enumerate}
\end{algorithmbox}

\noindent First, we verify feasibility of the relevant programs. 

\begin{lemma}\label{lem:program-satisfiability-fine-regression}
    The program constraints in Steps~\ref{line:program1} and \ref{line:program2} of \cref{algo:private-least-square-regression} are satisfied by the ground truth with probability at least $1-3\beta$.
\end{lemma}
 \begin{proof}
    The corruption constraints $\cA_{\rm corr}$ are satisfied by definition. 
    For the resilience constraints, by \cref{thm:covariance-concentration}, with probability at least $1-\beta$, the bound on $\Normop{\frac{1}{n}\Xnull\Xnull^\top-\Id}$ holds.
    On the other hand, the vector $\ynull- (\Xnull)^\top w$ is i.i.d Gaussian with variance bounded by $1$. 
    We can let $z_1$ represent the $\eta n$ entries in $y$ with largest absolute values.
    By \cref{thm:gaussian-order-stat-sharp}, we have
    \begin{equation*}
        \norm{z_2}_\infty^2\leq \chi\log(1/\eta)+\frac{\chi\log(1/\beta)}{\eta n}\,.
    \end{equation*}
    By \cref{thm:maximum-subset-sum-gaussian-correct}, we have 
    \begin{equation*}
        \norm{z_1}^2 \leq \chi\eta n\log(1/\eta)+\chi\log(1/\beta)\,.
    \end{equation*}
    
    Fix $u\seteq w^* - w_{\rm init}$ and let $a\defeq \Xnull^\top (\wnull-w_\init)\in\R^n$.
    Let $S\subseteq[n]$ be the indices of the $k=\eta n$ largest values of $|a_i|$ (breaking ties arbitrarily),
    and define
    \[
    b_1 \defeq a\odot \mathbf{1}_S,
    \qquad
    b_2 \defeq a\odot \mathbf{1}_{S^c}.
    \]
    Then $a=b_1+b_2$ by construction.
    
    By \cref{lem:sparse-spectral-norm}, there exists a universal
    constant $C>0$ such that, with probability at least $1-\beta$ over $X$,
    we have \(\|X\|_{\mathrm{sp}}(\eta n)\leq L\) where $L = C(
    \sqrt d\;+\;\sqrt{\eta n\log\frac{e}{\eta }}\;+\;\sqrt{\log\frac1\beta}
    )$.
    Under the condition that $n\geq \frac{d+\log(1/\beta)}{\eta}$, we have $L\lesssim \sqrt{\eta n \log(1/\eta)}$. 
    
    For the bound of $b_1$, note
    \[
    \|b_1\|_2
    =
    \|a_S\|_2
    =
    \|X_S^\top u\|_2
    \le
    \|X_S\|_{\op}\,\|u\|_2
    \le
    \|X\|_{\mathrm{sp}}(k)\,\|u\|_2
    \le
    L\|u\|_2,
    \]
    which implies $\|b_1\|_2^2\le L^2\|u\|_2^2$. 
    
    For the bound of $b_2$, let $|a|_{(1)}\ge \cdots \ge |a|_{(n)}$ be the order statistics of
    $\{|a_i|\}_{i=1}^n$.  By definition of $S$, for every $i\notin S$ we have $|b_2(i)|=|a_i|\le |a|_{(k+1)}\le |a|_{(k)}$.
    Moreover,
    \[
    |a|_{(k)}^2
    \le
    \frac{1}{k}\sum_{j=1}^k |a|_{(j)}^2
    =
    \frac{1}{k}\|a_S\|_2^2
    =
    \frac{1}{\eta n}\|b_1\|_2^2.
    \]
    Therefore,
    \[
    \|b_2\|_\infty^2
    =
    \max_{i\notin S}|a_i|^2
    \le
    |a|_{(k)}^2
    \le
    \frac{1}{k}\|b_1\|_2^2
    \le
    \frac{L^2}{\eta n}\|u\|_2^2,
    \]
    where we used the bound of $b_1$ in the last step.
    Therefore, the program constraints are satisfied.
\end{proof}

\noindent Next we show that $\hat{w}_1$ produced by first program in Stage 2 achieves error $O(\sqrt{\eta\log 1/\eta} + \sqrt{\frac{d + \log1/\beta}{n}})$:

\begin{lemma}\label{lem:regression-sos-identifiability-certifying}
        Consider the program in Step~\ref{line:program1_program} from \cref{algo:private-least-square-regression}. There is a degree-$O(1)$ sum-of-squares proof that \begin{align*}\cA(X^{(1)},y^{(1)},\xi^{(1)},w^{(1)};\Xinput,\yinput)\,, \ \cA(X^{(2)},y^{(2)},\xi^{(2)},w^{(2)};\Xinput,\yinput)\\
        \sststile{}{} \norm{w^{(1)}-w^{(2)}}^2 \leq O\Paren{\eta\log(1/\eta)+\frac{d+\log(1/\beta)}{n}}\,.
    \end{align*}
    By standard rounding, this implies that
    \begin{equation}
        \norm{\hat{w}_1 - w^*}_2 \le  O\Bigl(\sqrt{\eta\log 1/\eta}+\sqrt{\frac{d+\log(1/\beta)}{n}}\Bigr)\,.
    \end{equation} 
\end{lemma}
\begin{proof}
    For $j\in \{1,2\}$, define $y^{(j)}_1=y^{(j)}-(X^{(j)})^\top w_{\init}$ and $w^{(j)}_1=w^{(j)}-w_{\init}$.
    Then we have $\norm{w^{(j)}_1}^2\leq O(1)$.
    We aim to give sum-of-squares proof that 
    \begin{equation*}
        \norm{w_{\LS}^{(1)}-w_{\LS}^{(2)}}^2\leq O\Bigl(\eta \log(1/\eta)+\frac{d+\log(1/\beta)}{n}\Bigr)\,.
    \end{equation*}
 
    Let $v=1-(1-\xi^{(1)})\odot (1-\xi^{(2)})$, noting that
    \begin{align*}
         X^{(1)}(1-v) &=X^{(2)}(1-v)\\
        (1-v)\odot y_1^{(1)} &=(1-v)\odot y_1^{(2)}\,.
    \end{align*}
    We perform the split
    \begin{equation*}
        X^{(1)}y_1^{(1)}-X^{(2)}y_1^{(2)}=X^{(1)}(y_1^{(1)}\odot v)-X^{(2)}(y_1^{(2)}\odot v)
    \end{equation*} 
    and bound these two terms separately, using the spectral norm bound on $X$. 
    In particular, we have
    \begin{equation*}
    \norm{X^{(1)}(y_1^{(1)}\odot v)}^2=\normop{X^{(1)}}^2 \cdot \norm{y_1^{(1)}\odot v}^2\leq \chi n\cdot \norm{y_1^{(1)}\odot v}^2\,.
    \end{equation*}
    By the short-flat decomposition for $y^{(1)}$,
    \begin{equation*}
        \norm{y_1^{(1)}\odot v}^2\leq \chi^2\eta n\log(1/\eta)+\log(1/\beta)\,.
    \end{equation*}
    As a result, we have \begin{equation*}
      \norm{X^{(j)}(y_1^{(j)}\odot v)}^2 \leq O(\eta n^2 \log(1/\eta)+\log(1/\beta) n)
    \end{equation*}
    for $j = 1,2$, and thus $\norm{w_{\LS}^{(1)}-w_{\LS}^{(2)}}^2\leq \eta\log(1/\eta)+\log(1/\beta)/{n}$.
    On the other hand, we have the program constraints that $\norm{w^{(j)}-w^{(j)}_{\LS}}^2\leq O(\frac{d+\log(1/\beta)}{n})$, from which the bound on $\norm{w^{(1)}-w^{(2)}}$ follows by triangle inequality.
\end{proof}

\noindent Finally, we show that the last program in \cref{algo:private-least-square-regression} results in a sufficiently accurate approximation to the posterior mean.
 
\begin{proof}[Proof of \cref{lem:robust-regression-least-square}]
    Let $y'=y-X^\top \hat{w}_1$, and $w'=w-\hat{w}_1$.
    Let $\ynull^\prime=\ynull-\Xnull^\top \hat{w}_1$ and $\wnull^\prime=\wnull-\hat{w}_1$.
    By the guarantee on $\hat{w}_1$ from \cref{lem:regression-sos-identifiability-certifying}, we have $\norm{w'}^2\leq O\Paren{\eta \log(1/\eta)+\frac{d+\log(1/\beta)}{n}}$. 

    Now we give sum-of-squares proof that for any unit vector $u$ 
    \begin{equation*}
        \iprod{Xy'-\Xnull \ynull^\prime,u}^2\leq O\Paren{\eta^2 \log(1/\eta)+\eta\sqrt{\frac{d+\log(1/\beta)}{n}}}\,.
    \end{equation*}
    Let $r$ be the indicator for the set of corrupted nodes. Let $v=1-(1-\xi)\odot (1-r)$, noting that
    \begin{align*}
         X((1-v)\odot y) &=\Xnull((1-v)\odot \ynullp)\\
        (1-v)\odot y &=(1-v)\odot \ynullp\,.
    \end{align*}
    As before, we perform the split
    \begin{equation*}
        Xy'-\Xnull \ynullp=X(y'\odot v)-\Xnull(\ynullp\odot v) \label{eq:split_triangle}
    \end{equation*} 
    and bound these two terms separately, using the spectral norm bound on $X$. In particular, for any unit vector $u$,
    \begin{equation*}
        \iprod{X(y'\odot v),u}^2\leq \norm{Xu\odot v}^2\norm{y'\odot v}^2\,.
    \end{equation*}
    Now similar to our proofs for empirical mean estimation, we have a sum-of-squares proof that 
    \begin{align*}
        \frac{1}{n}\norm{Xu \odot v}^2 &\leq O\Bigl(\eta \log(1/\eta)+\sqrt{\frac{d+\log(1/\beta)}{n}}\Bigr) \\
        \norm{y'\odot v}^2 & \leq (\eta n\log(1/\eta)+\log(1/\beta))\,. 
    \end{align*}
    Combining the above, we have a sum-of-squares proof that
    \begin{equation}
        \frac{1}{n}\iprod{X(y'\odot v),u}^2\leq O\Bigl(\eta \log(1/\eta)+\sqrt{\frac{d+\log(1/\beta)}{n}}\Bigr)\cdot (\eta n\log(1/\eta)+\log(1/\beta))\,, \label{eq:residbound}
    \end{equation}
    as well as the same bound for $\frac{1}{n}\iprod{\Xnull(\ynullp\odot v),u}^2$.
    Therefore, by \cref{eq:split_triangle}, for any pseudo-distribution $\tE$ satisfying the constraints and for any unit vector $u$, we have that $\tilde{\E}\Bigl[\frac{1}{n}\iprod{Xy'-\Xnull \ynullp,u}^2\Bigr]$ is bounded by the right-hand side of \cref{eq:residbound}.
    Therefore by rounding the pseudo-distribution, we conclude that
    \begin{equation}\label{eq:fine-estimation-error-regression}
        \Bigl\|\frac{1}{n}\Paren{\tE[Xy']-\Xnull \ynullp}\Bigr\|^2\leq O\Paren{\eta \log(1/\eta)+\sqrt{\frac{d+\log(1/\beta)}{n}}}\cdot \paren{\eta \log(1/\eta)+\log(1/\beta)/n}\,.
    \end{equation}

    Now we show that $(n+\frac{1}{\sigma^2})^{-1}\Xnull\ynull'+\Paren{1+\frac{1}{n\sigma^2}}^{-1}\hat{w}_1$ is close to the true posterior mean $(\Xnull\Xnull^\top+\frac{1}{\sigma^2}\Id)^{-1}\Xnull\ynull$. 
    By definition, we have $\ynull'=\ynull-\Xnull^\top\hat{w}_1$. 
    Now we have
    \begin{align*}
        \MoveEqLeft (n+\frac{1}{\sigma^2})^{-1}\Xnull \ynullp \\
        &= \Paren{\Xnull\Xnull^\top+\frac{1}{\sigma^2}\Id}^{-1}\Xnull (\ynull-\Xnull^\top \hat{w}_1)+\Paren{(n+\frac{1}{\sigma^2})^{-1}-\Paren{\Xnull\Xnull^\top+\frac{1}{\sigma^2}\Id}^{-1}}\Xnull(\ynull-\Xnull^\top \hat{w}_1)\\
        &= \Paren{\Xnull\Xnull^\top+\frac{1}{\sigma^2}\Id}^{-1}\Xnull \ynull-\hat{w}_1+\Paren{\Xnull\Xnull^\top+\frac{1}{\sigma^2}\Id}^{-1}\frac{\hat{w}_1}{\sigma^2}\\
        &\qquad +\Paren{(n+\frac{1}{\sigma^2})^{-1}-\Paren{\Xnull\Xnull^\top+\frac{1}{\sigma^2}\Id}^{-1}}\Xnull(\ynull-\Xnull^\top \hat{w}_1)\,.
    \end{align*}

    First, 
    \begin{equation*}
        \Norm{\Paren{\Xnull\Xnull^\top+\frac{1}{\sigma^2}\Id}^{-1}\frac{\hat{w}_1}{\sigma^2} - \frac{\hat{w}_1}{n\sigma^2+1}} \leq \Normop{\Paren{\Xnull\Xnull^\top+\frac{1}{\sigma^2}\Id}^{-1}-\frac{\sigma^2}{n\sigma^2+1}}\frac{\norm{\hat{w}_1}}{\sigma^2}
    \end{equation*}
    Moreover, we have 
    \begin{equation*}
         \Normop{\Paren{\Xnull\Xnull^\top+\frac{1}{\sigma^2}\Id}^{-1}-\frac{\sigma^2}{n\sigma^2+1}\Id}\leq \frac{1}{n} \Normop{\frac{1}{n}\Xnull\Xnull^\top-\Id}\leq \frac{1}{n} \sqrt{\frac{d+\log(1/\beta)}{n}}\,.
    \end{equation*}
    Combining, we obtain
    \begin{equation*}
    \Norm{\Paren{\Xnull\Xnull^\top+\frac{1}{\sigma^2}\Id}^{-1}\frac{\hat{w}_1}{\sigma^2} - \frac{\hat{w}_1}{n\sigma^2+1}}
    \leq \sqrt{\frac{d+\log(1/\beta)}{n}}\cdot \frac{\norm{\hat{w}_1}}{n\sigma^2}\lesssim \sqrt{\frac{d+\log(1/\beta)}{n}}\cdot \frac{1}{n\sigma}\lesssim \frac{d+\log(1/\beta)}{n}\,.
    \end{equation*}

    Finally we want to bound the norm of $\Paren{(n+\frac{1}{\sigma^2})^{-1}-\Paren{\Xnull\Xnull^\top+\frac{1}{\sigma^2}\Id}^{-1}}\Xnull(\ynull-\Xnull^\top \hat{w}_1)$.
    In particular, with probability at least $1-\beta$, for some universal constant $C$, we have
    \begin{equation*}
        \Normop{(n+\frac{1}{\sigma^2})^{-1}\Id-\Paren{\Xnull\Xnull^\top+\frac{1}{\sigma^2}\Id}^{-1}}\leq \frac{1}{n}\Normop{\frac{1}{n}\Xnull\Xnull^\top-\Id}\leq \frac{1}{n}\sqrt{\frac{d+\log(1/\beta)}{n}}\,.
    \end{equation*}

    On the other hand, we have
    \begin{equation*}
        \norm{\frac{1}{n}\Xnull(\ynull-\Xnull^\top \hat{w}_1)}=\norm{\frac{1}{n}\Xnull\Xnull^\top(\wnull-\hat{w}_1)}+\Norm{\frac{1}{n}\Xnull z_0}\leq O\Paren{\sqrt{\eta\log 1/\eta}+\sqrt{\frac{d+\log(1/\beta)}{n}}}\,.
    \end{equation*}
    Therefore, we have 
    \begin{multline*}
         \Norm{\Paren{(n+\frac{1}{\sigma^2})^{-1}\Id-\Paren{\Xnull\Xnull^\top+\frac{1}{\sigma^2}\Id}^{-1}}\Xnull(\ynull-\Xnull^\top \hat{w}_1)} \\
         \leq O\Paren{\sqrt{\eta\log 1/\eta}+\sqrt{\frac{d+\log(1/\beta)}{n}}}\cdot \Paren{\sqrt{\frac{d+\log(1/\beta)}{n}}}\,.
    \end{multline*}
    Putting everything together, we get
    \begin{align*}
        \MoveEqLeft\Norm{(n+\frac{1}{\sigma^2})^{-1}\Xnull \ynullp +\Paren{1+\frac{1}{n\sigma^2}}^{-1}\hat{w}_1-(\Xnull\Xnull^\top+\frac{1}{\sigma^2}\Id)^{-1}\Xnull\ynull} \\ 
        &\leq  \Norm{\Paren{(n+\frac{1}{\sigma^2})^{-1}\Id-\Paren{\Xnull\Xnull^\top+\frac{1}{\sigma^2}\Id}^{-1}}\Xnull(\ynull-\Xnull^\top \hat{w}_1)}+  \Norm{\Paren{\Xnull\Xnull^\top+\frac{1}{\sigma^2}\Id}^{-1}\frac{\hat{w}_1}{\sigma^2} - \frac{\hat{w}_1}{n\sigma^2+1}}\\
        &\leq  O\Paren{\sqrt{\eta\log 1/\eta}+\sqrt{\frac{d+\log(1/\beta)}{n}}}\cdot \Paren{\sqrt{\frac{d+\log(1/\beta)}{n}}}\,.
    \end{align*}
    Combining with \cref{eq:fine-estimation-error-regression}, we conclude the proof.
\end{proof}

\subsubsection{Putting everything together}

To obtain a private algorithm, it remains to plug the above robustness guarantees into the robustness-to-privacy reduction. We implement this in \cref{sec:privatize} to prove Theorem~\ref{thm:private-approximation-least-square}. Finally, by combining this with the guarantee in the critical prior case, we immediately deduce \cref{thm:bayesian-error-regression}.

\begin{proof}[Proof of \cref{thm:bayesian-error-regression}]
    When $\sigma=o(1/\sqrt{n})$, the claimed error rate is achieved by sampling from the prior distribution.
    When $\sigma=\Omega(1/\sqrt{n})$, we can combine \cref{thm:private-bayesian-linear-regression}, \cref{thm:private-approximation-least-square} and \cref{fact:error-decomposition}.
    In particular, by \cref{thm:private-bayesian-linear-regression}, \cref{thm:private-approximation-least-square}, when $n\geq d\log(\sigma\sqrt{d})/\e$ with probability $1-\beta$, we can achieve error rate $\alpha=O(\eta'+\sqrt{\eta'\sqrt{\frac{d+\log(1/\beta)}{n}}})$ where $\eta'=\eta+\frac{d+\log(1/\beta)}{n\e}$. 
\end{proof}

\input{content/regression-lb}

%% file: content/regression-lb.tex
\subsection{Hardness of robust Bayesian regression}\label{sec:robust-lr-lb}

We now give a computational lower bound for robust posterior mean estimation in the regime
\(
\eta\log(1/\eta)\lesssim \alpha \le \eta d^{1/4-o(1)}.
\)
As in the rest of this subsection, it is convenient to work first with the simplified posterior statistic
\[
w_{\rm simp}(X,y)\defeq \Paren{n+\frac{1}{\sigma^2}}^{-1}Xy.
\]
In the critical-prior regime \(\sigma^2=\Theta(1/n)\), \cref{lem:Lipshitz-least-square} implies that
\(w_{\rm simp}(X^\circ,y^\circ)\) differs from the exact posterior mean by
\(O((d+\log(1/\beta))/n)\).

\begin{definition}[Distinguishing mixtures of linear regressions]
\label{def:variance-matched-xcorruption-test}
Fix \(\eta\in(0,1/2)\). 
Let $K$ be some sufficiently large constant.
Fix a signal parameter \(\delta>0\), and define
\[
s \defeq \sqrt{1+K^2\alpha^2},
\qquad
a \defeq -\frac{K\alpha}{\eta s^2}.
\]
We consider the following one-sample hypothesis testing problem on
\((x_i,y_i)\in\R^d\times\R\).

\begin{itemize}
    \item \textbf{Null distribution \(Q\):}
    \(
    x_i = g_i,
    y_i = sz_i,
    \)
    where \(g_i\sim \mathcal N(0,\Id_d)\), \(z_i\sim \mathcal N(0,1)\), and \(g_i,z_i\) are independent.

    \item \textbf{Planted \(P\):}
    first draw
    \(
    u \sim \mathrm{Unif}(\mathbb S^{d-1}),
    B \sim \mathrm{Ber}(q).
    \)
    Conditional on \((u,B)\), draw
    \(
    g_i\sim \mathcal N(0,\Id_d),
    z_i\sim \mathcal N(0,1),
    \)
    independently, and set
    \(
    y_i \defeq K\alpha \langle g_i,u\rangle + z_i.
    \)
    If \(B=0\), let
    \(
    x_i \defeq g_i.
    \)
    If \(B=1\), let
    \(
    x_i \defeq g_i + ay_iu\,.
    \)
\end{itemize}
\end{definition}

\noindent We record some basic properties of the distinguishing problem:

\begin{lemma}
\label{lem:basic-properties-xcorruption}
Assume
\(
\alpha^2 \lesssim \eta.
\)
Then the following hold.

\begin{enumerate}
    \item Under \(P\), conditional on \(u\), the observed sample is a \(q\)-corruption of a clean
    Gaussian-design regression sample with regressor \(w^\star=K\alpha u\).

    \item Under both \(Q\) and \(P\), the marginal law of \(y_i\) is exactly \(\mathcal N(0,s^2)\).

    \item Under \(Q\),
    \(
    \E_Q[x_iy_i]=0,
    \)
    while under \(P\), conditional on \(u\),
    \(
    \E_P[x_iy_i\mid u]=0.
    \)

    \item The null law \(Q\) is a \(\eta\)-corruption of a clean Gaussian-design regression sample with
    regressor \(w^\star=0\).
\end{enumerate}
\end{lemma}

\noindent We show that a robust estimator for the exact posterior mean can be used to solve the distinguishing problem.

\begin{lemma}[Reduction from distinguishing to robust posterior mean estimation]
\label{lem:reduction-xcorruption}
Assume \(\sigma^2=\Theta(1/n)\), and define
\[
w_{\rm simp}(X,y)\defeq \Paren{n+\frac1{\sigma^2}}^{-1}Xy.
\]
Consider the testing problem in \cref{def:variance-matched-xcorruption-test}.
Assume that
\begin{equation}
\label{eq:reduction-side-conditions}
n \ge C_1\frac{d+\log(1/\beta)}{\alpha}
\;+\;
C_2\frac{\eta(d+\log(1/\beta))}{\alpha^2}
\;+\;
C_3\frac{\log(1/\beta)}{\eta}
\end{equation}
for sufficiently large universal constants \(C_1,C_2,C_3\).

Suppose there exists a polynomial-time estimator \(\mathcal A\) with the following guarantee:
for every \(\eta\)-corruption \((X,y)\) of a clean sample \((X^\circ,y^\circ)\) from
\cref{def:bayesian-regression-under-gaussian-prior} with \(\|\wnull\|_2\le \alpha\), with
probability at least \(9/10\),
\[
\Norm{\mathcal A(X,y)-\E[\wnull\mid X^\circ,y^\circ]}_2 \le \frac{\alpha}{20}.
\]
Then there is a polynomial-time test that distinguishes \(Q\) from \(P\) in
\cref{def:variance-matched-xcorruption-test} with probability at least
\(1-O(\beta)-\exp(-\Omega(\eta n))\).
\end{lemma}

\noindent In particular, we give the following test:
we compute
\(
w_{\rm obs}\defeq w_{\rm simp}(X,y),
\widehat w\defeq \mathcal A(X,y)
\),
and output \emph{planted} if
\(
\Norm{\widehat w-w_{\rm obs}}_2 \ge 2\alpha,
\)
and \emph{null} otherwise.

Finally, we give a lemma showing that the hypothesis testing problem is low-degree hard in the regime that we are interested in. 
\begin{lemma}
\label{lem:ldlr-xcorruption}
Assume \(\alpha\leq \min\Set{\eta d^{1/4}/D,\eta^2}\) and \(D\le d^{10^{-4}}\).
Consider the testing problem in \cref{def:variance-matched-xcorruption-test}, and let
\(\Adv_{\le D}\) denote the degree-\(D\) advantage.
Then there exists a universal constant \(c>0\) such that whenever
\(
n \le c\,\frac{d \eta^2}{D^2\alpha^4},
\)
we have
\(
\Adv_{\le D}=O(1).
\)
\end{lemma}

\noindent From these two results, we can conclude an information-computation gap for robust posterior mean estimation in regression.

\begin{theorem}[Information-computation gap for robust posterior mean estimation in regression]
\label{thm:robust-regression-gap-moderately-dense}
There exist absolute constants \(c_0,C_0>0\) such that the following holds for all sufficiently large
\(d\).
Assume \(\alpha\leq \min\Set{\eta d^{1/4}/D,\eta^2}\).
Suppose
    \[
    C_0\frac{d+\log(1/\beta)}{\alpha}
    \;+\;
    C_0\frac{\eta(d+\log(1/\beta))}{\alpha^2}
    \;+\;
    C_0\frac{\log(1/\beta)}{\eta}
    \;\le\;
    n
    \;\le\;
    c_0 \frac{d\eta^2}{D^2\alpha^4}\,.
    \]
Assuming the low-degree conjecture in the hypothesis testing problem
\cref{def:variance-matched-xcorruption-test}.

Then any polynomial-time estimator that robustly estimates the exact posterior mean to error
    \(\alpha/20\) with probability at least \(9/10\) on every \(\eta\)-corruption of clean samples from
    \cref{def:bayesian-regression-under-gaussian-prior} with \(\|\wnull\|_2\le \alpha\) yields a
    polynomial-time distinguisher for the testing problem.
\end{theorem}

\begin{proof}
This follows from combining \cref{lem:reduction-xcorruption} and \cref{lem:ldlr-xcorruption}.
\end{proof}

\subsubsection{Proof of \cref{lem:basic-properties-xcorruption}}
\begin{proof}[Proof of \cref{lem:basic-properties-xcorruption}]
Let $\delta=K\alpha$.
For (i), define the clean sample
\[
x_i^\circ \defeq g_i,
\qquad
y_i^\circ \defeq \delta\langle g_i,u\rangle + z_i.
\]
Then \((X^\circ,Y^\circ)\) is exactly one clean sample from
\cref{def:bayesian-regression-under-gaussian-prior} with regressor \(w^\star=\delta u\).
Under \(P\), if \(B=0\) we keep the sample unchanged, and if \(B=1\) we replace \(x_i^\circ\) by
\(x_i^\circ+ay_i^\circ u\) while leaving \(y_i^\circ\) unchanged.  This proves (i).

For (ii), under \(Q\) the claim is immediate.  Under \(P\), conditional on \(u\), in both branches
we have
\[
y_i = \delta\langle g_i,u\rangle + z_i.
\]
Since \(\langle g_i,u\rangle\sim \mathcal N(0,1)\) and is independent of \(z_i\sim \mathcal N(0,1)\),
it follows that \(y_i\sim \mathcal N(0,1+\alpha^2)=\mathcal N(0,s^2)\).

For (iii), under \(Q\) we have independence and centering, so \(\E_Q[x_iy_i]=0\).
Under \(P\), conditional on \(u\), in the clean branch \(B=0\),
\[
\E[x_iy_i\mid u,B=0]
=
\E\bigl[x_i(\delta\langle G,u\rangle+Z)\mid u\bigr]
=
\delta u.
\]
In the outlier branch \(B=1\),
\[
\E[x_iy_i\mid u,B=1]
=
\E[(g_i+auy_i)y_i\mid u]
=
\E[g_iy_i\mid u] + a\E[y_i^2\mid u]u
=
\delta u + as^2u.
\]
Averaging over \(B\sim \Ber(q)\),
\[
\E_P[x_iy_i\mid u]
=
(1-q)\delta u + q(\delta+as^2)u
=
\bigl(\delta + qas^2\bigr)u
=
0,
\]
because \(a=-\delta/(qs^2)\).

For (iv), let \(\varphi_\sigma\) denote the density of \(\mathcal N(0,\sigma^2)\), and define
\[
r_0(y)\defeq \frac{\varphi_s(y)-(1-\eta)\varphi_1(y)}{\eta}.
\]
Since
\[
\frac{\varphi_s(y)}{\varphi_1(y)}
=
\frac{1}{s}\exp\Paren{\frac{\alpha^2}{2(1+\alpha^2)}y^2},
\]
the ratio is minimized at \(y=0\), so
\[
\inf_y \frac{\varphi_s(y)}{\varphi_1(y)} = \frac{1}{s}.
\]
Also,
\[
1-\frac{1}{\sqrt{1+t}} \le \frac{t}{2}
\qquad (t\ge 0),
\]
hence
\[
1-\frac1s \le \frac{\alpha^2}{2}\le \eta.
\]
Therefore
\[
\varphi_s(y)\ge (1-\eta)\varphi_1(y)
\qquad\text{for all }y,
\]
so \(r_0(y)\ge 0\). Moreover,
\[
\int_\R r_0(y)\,dy
=
\frac{1-(1-\eta)}{\eta}
=
1,
\]
so \(r_0\) is a valid density.

Now start from a clean \(w^\star=0\) sample
\(
x_i^\circ = g_i,
\) and \(
y_i^\circ = z_i
\),
with \(g_i\sim \mathcal N(0,\Id_d)\), \(z_i\sim \mathcal N(0,1)\) independent.  Keep \(x_i^\circ\) unchanged.
With probability \(1-\eta\), keep \(y_i^\circ\); with probability \(\eta\), replace it by an independent draw
\(R\sim r_0\).  The observed response then has density
\[
(1-\eta)\varphi_1 + \eta r_0 = \varphi_s,
\]
so \(y_i\sim \mathcal N(0,s^2)\), and therefore the resulting observed pair \((x_i,y_i)\) has exactly law \(Q\).
\end{proof}

\subsubsection{Reducing hypothesis testing to robust posterior mean estimation.}

In this section, we give the proof of \cref{lem:reduction-xcorruption}
\begin{proof}[Proof of \cref{lem:reduction-xcorruption}]
Let
\[
\lambda \defeq \Paren{n+\frac1{\sigma^2}}^{-1},
\qquad
\nu \defeq n\lambda = \frac{n\sigma^2}{1+n\sigma^2}.
\]
Since \(\sigma^2=\Theta(1/n)\), we have \(\nu=\Theta(1)\).

We first consider the case under null distribution.
By \cref{lem:basic-properties-xcorruption}, \(Q\) can be realized as a \(\eta\)-corruption of a clean
\(w^\star=0\) sample \((X^\circ,y^\circ)\), where \(X^\circ_i=G_i\) and \(y^\circ_i=Z_i\).
Let \(C_i\sim \Ber(\eta)\) denote the indicators of the corrupted responses.  
A Chernoff bound then yields
\[
\Pr\Bigl[\sum_{i=1}^n C_i \le \eta n\Bigr] \ge 1-\exp(-\Omega(\eta n)).
\]
On this event, the robust guarantee for \(\mathcal A\) applies.  Also, by \cref{lem:Lipshitz-least-square} and the first term
in \eqref{eq:reduction-side-conditions},
\[
\Norm{\E[\wnull\mid X^\circ,y^\circ]-w_{\rm simp}(X^\circ,y^\circ)}_2 \le \frac{\alpha}{20}
\]
except with probability at most \(\beta\).  Hence, except with probability \(O(\beta)\),
\begin{equation}
\label{eq:null-clean-close}
\Norm{\widehat w-w_{\rm simp}(X^\circ,y^\circ)}_2 \le \frac{\alpha}{10}.
\end{equation}

Now compare the observed and clean simplified statistics.  Since only the responses are modified,
\[
w_{\rm obs}-w_{\rm simp}(X^\circ,y^\circ)
=
\lambda \sum_{i=1}^n X_i^\circ \Delta_i,
\qquad
\Delta_i \defeq C_i(R_i-Z_i).
\]
The replacement density \(r_0\) has variance
\[
\mathrm{Var}(R_i)=\frac{s^2-(1-\eta)}{\eta}=1+\frac{K^2\alpha^2}{\eta}\le O(1),
\]
so \(\Delta_i\) are centered and subgaussian with variance proxy \(O(\eta)\).  Thus
\[
\sum_{i=1}^n \Delta_i^2 \le C \eta n
\]
with probability at least \(1-\exp(-\Omega(\eta n))\).  Conditional on the \(\Delta_i\)'s, the vector
\(\sum_i X_i^\circ \Delta_i\) is Gaussian with covariance \((\sum_i \Delta_i^2)\Id_d\).  Therefore,
with probability at least \(1-\beta-\exp(-\Omega(\eta n))\),
\[
\Norm{w_{\rm obs}-w_{\rm simp}(X^\circ,y^\circ)}_2
\le
C\lambda \sqrt{\eta n(d+\log(1/\beta))}
\le
C\sqrt{\frac{\eta(d+\log(1/\beta))}{n}}
\le
\alpha,
\]
where the last step uses the second term in \eqref{eq:reduction-side-conditions}.
Combining with \eqref{eq:null-clean-close},
\[
\Norm{\widehat w-w_{\rm obs}}_2 \le \frac{11}{10}\alpha < 2\alpha
\]
except with probability \(O(\beta)+\exp(-\Omega(\eta n))\).  Thus the test outputs \emph{null}.

We then consider the planted distribution.
By construction, the planted law is realized by first drawing a clean sample
\[
X_i^\circ = G_i,
\qquad
y_i^\circ = \alpha\langle G_i,u\rangle + Z_i
\]
with regressor \(w^\star=\alpha u\), then drawing i.i.d.\ indicators \(B_i\sim \Ber(q)\), and finally
setting
\[
(X_i,y_i)=
\begin{cases}
(X_i^\circ,y_i^\circ), & B_i=0,\\
(X_i^\circ + a y_i^\circ u,\ y_i^\circ), & B_i=1.
\end{cases}
\]
Again
\[
\Pr\Bigl[\sum_{i=1}^n B_i \le \eta n\Bigr]\ge 1-\exp(-\Omega(\eta n)).
\]
On this event the robust guarantee for \(\mathcal A\) applies, and by \cref{lem:Lipshitz-least-square} together with the first
term in \eqref{eq:reduction-side-conditions},
\begin{equation}
\label{eq:planted-clean-close}
\Norm{\widehat w-w_{\rm simp}(X^\circ,y^\circ)}_2 \le \frac{\alpha}{10}
\end{equation}
except with probability \(O(\beta)\).

Now compare the observed and clean simplified statistics.  Since only \(X\) changes,
\[
w_{\rm obs}-w_{\rm simp}(X^\circ,y^\circ)
=
\lambda \sum_{i=1}^n B_i\,a\,(y_i^\circ)^2\,u.
\]
Taking expectations conditional on \(u\),
\[
\E[w_{\rm obs}-w_{\rm simp}(X^\circ,y^\circ)\mid u]
=
\lambda n \eta a s^2 u
=
-\nu K\alpha\,u.
\]
Also, \(y_i^\circ\sim \mathcal N(0,s^2)\), so Bernstein's inequality gives
\[
\left|
\sum_{i=1}^n B_i\bigl((y_i^\circ)^2-s^2\bigr)
\right|
\le
C s^2 \sqrt{\eta n\log(1/\beta)}
\]
with probability at least \(1-\beta-\exp(-\Omega(\eta n))\).  Therefore,
\[
\Norm{\Paren{w_{\rm obs}-w_{\rm simp}(X^\circ,y^\circ)}+\nu\alpha u}_2
\le
C\lambda |a| s^2 \sqrt{\eta n\log(1/\beta)}
=
C\nu K\alpha\sqrt{\frac{\log(1/\beta)}{\eta n}}.
\]
By the third term in \eqref{eq:reduction-side-conditions}, the right-hand side is at most
\(\nu K\alpha/4\).  
Hence, except with probability \(\beta+\exp(-\Omega(\eta n))\),
\[
\Norm{w_{\rm obs}-w_{\rm simp}(X^\circ,y^\circ)}_2 \ge \frac{3\nu K}{4}\alpha.
\]
Combining with \eqref{eq:planted-clean-close},
\[
\Norm{\widehat w-w_{\rm obs}}_2
\ge
\Norm{w_{\rm obs}-w_{\rm simp}(X^\circ,y^\circ)}_2
-
\Norm{\widehat w-w_{\rm simp}(X^\circ,y^\circ)}_2
\ge
\frac{3\nu K}{4}\alpha- \frac{\alpha}{10}.
\]
Since \(\nu=\Theta(1)\), choosing \(K\) sufficiently large gives
\[
\frac{3\nu K}{4}\alpha - \frac{\alpha}{10} \ge 2\alpha.
\]
Thus the test outputs \emph{planted} except with probability \(O(\beta)+\exp(-\Omega(\eta n))\).
\end{proof}

\subsubsection{Proof of \cref{lem:ldlr-xcorruption}}
It is easier to prove low-degree lower bounds when the null distribution is standard Gaussian.
We define the fixed linear map
\[
T:\R^{d+1}\to\R^{d+1},
\qquad
T(x,y)\defeq (x,y/s).
\]
Since \(T\) is an invertible linear map and preserves total polynomial degree, the degree-\(D\)
advantage is unchanged by applying \(T\) samplewise. We will henceforth work with the data transformed by $T$, for which the likelihood ratio admits the following expansion.

\begin{lemma}
\label{lem:whitened-xcorruption}
Let $\theta \defeq \frac{K\alpha}{s}$.
Under \(T\), the null becomes standard Gaussian on \(\R^{d+1}\):
\(
T_\# Q = \mathcal N(0,\Id_{d+1}).
\)
Under the planted law, conditional on \(u\), if we write
\(
t \defeq \langle x,u\rangle,
y \defeq Y/s,
\)
then the one-sample likelihood ratio with respect to the whitened null has the expansion
\begin{equation}
\label{eq:whitened-one-sample-lr}
\ell_u(x,y)
=
1+\sum_{k\ge 1} h_k(\langle x,u\rangle)\,\psi_k(y),
\end{equation}
where
\begin{equation}
\label{eq:psi-k}
\psi_k(y)
=
\theta^k\Paren{(1-q)h_k(y)+\eta\,h_k\Paren{\Paren{1-\frac1\eta}y}}.
\end{equation}
In particular,
\(
\psi_1(y)\equiv 0.
\)
\end{lemma}

\begin{proof}
Under the null, \(X\sim \mathcal N(0,\Id_d)\) and \(Y/s\sim \mathcal N(0,1)\) independently, so
\(T_\#Q=\mathcal N(0,\Id_{d+1})\).
Let $\delta=K\alpha$.

Now condition on the planted direction \(u\), and let \(g\defeq \langle G,u\rangle\).
Then
\[
y=\frac{Y}{s}
=
\frac{\delta g + Z}{s}
=
\theta g + \sqrt{1-\theta^2}\,\xi,
\]
where \(\xi\sim\mathcal N(0,1)\).  Hence \(y\sim \mathcal N(0,1)\).

In the clean branch \(B=0\), we have \(x=G\), so \(t=g\).  Therefore
\[
t\mid y,\,B=0 \sim \mathcal N(\theta y,\,1-\theta^2).
\]
In the outlier branch \(B=1\), we have
\[
x = G + aYu = G + as\,y\,u,
\]
so
\[
t = g + by,
\qquad
b\defeq as = -\frac{\delta}{\eta s} = -\frac{\theta}{\eta}.
\]
Hence
\[
t\mid y,\,B=1 \sim \mathcal N\Paren{\theta\Paren{1-\frac1\eta}y,\,1-\theta^2}.
\]

By Fact~4.7, if \(T\sim \mathcal N(\rho x,1-\rho^2)\), then
\[
\E[h_k(T)] = \rho^k h_k(x).
\]
Applying this with \(\rho=\theta\) in the two branches gives
\[
\E[h_k(t)\mid y,B=0] = \theta^k h_k(y),
\qquad
\E[h_k(t)\mid y,B=1]
=
\theta^k h_k\Paren{\Paren{1-\frac1\eta}y}.
\]
Averaging over \(B\sim \Ber(\eta)\) yields \eqref{eq:psi-k}, we immediately have
\begin{align}
\psi_1(y)
&=
\theta\Paren{(1-\eta)h_1(y)+\eta h_1\Paren{\Paren{1-\frac1\eta}y}}
=
0\,.\qedhere
\end{align}
\end{proof}

\noindent In preparation for the proof of \cref{lem:ldlr-xcorruption}, we prove the following two helper lemmas.

\begin{lemma}
\label{lem:scaled-hermite-bound}
Let \(Y\sim \mathcal N(0,1)\), let \(|c|\ge 1\), and let \(k\ge 0\). Then
\(
\E[h_k(cY)^2] \le 2(4c^2)^k.
\)
\end{lemma}

\begin{proof}
By Mehler's identity, for every \(|r|<1\),
\[
\sum_{k\ge 0} r^k h_k(x)^2
=
\frac{1}{\sqrt{1-r^2}}\exp\Paren{\frac{r x^2}{1+r}}.
\]
Substituting \(x=cY\) and taking expectation over \(Y\sim \mathcal N(0,1)\) gives
\[
\sum_{k\ge 0} r^k \E[h_k(cY)^2]
=
\frac{1}{\sqrt{1-r^2}}
\cdot
\frac{1}{\sqrt{1-\frac{2rc^2}{1+r}}}.
\]
Set
\[
r = \frac{1}{4c^2}.
\]
Since \(|c|\ge 1\), we have \(0<r<1\), and
\[
1+r-2rc^2 = 1+\frac{1}{4c^2}-\frac12 \ge \frac12,
\qquad
1-r\ge \frac34.
\]
Hence the right-hand side is at most \(2\).  Since all coefficients are nonnegative,
\[
r^k \E[h_k(cY)^2]\le 2,
\]
and therefore
\[
\E[h_k(cY)^2]\le 2r^{-k}=2(4c^2)^k.
\qedhere
\]
\end{proof}

\begin{lemma}
\label{lem:psi-coefficient-bound}
Let
\(
a_k \defeq \|\psi_k\|_{L_2(\mathcal N(0,1))}^2.
\)
Then
\(
a_1=0,
\)
and for every \(k\ge 2\),
\(
a_k \le \eta^2 \Paren{\frac{C\theta^2}{\eta^2}}^k
\)
for a universal constant \(C>0\).
\end{lemma}

\begin{proof}
The identity \(a_1=0\) follows from \(\psi_1\equiv 0\).
For \(k\ge 2\), by \eqref{eq:psi-k},
\[
a_k
\le
2\theta^{2k}\Bigl((1-\eta)^2\|h_k\|_2^2 + \eta^2\,\E\Bigl[h_k\Bigl(\Paren{1-\frac1\eta}Y\Bigr)^2\Bigr]\Bigr).
\]
Since \(\|h_k\|_2=1\) and \(|1-1/\eta|\ge 1\), \cref{lem:scaled-hermite-bound} gives
\[
a_k
\le
2\theta^{2k}\Paren{1 + 2\eta^2 \Bigl(4\Paren{1-\frac1\eta}^2\Bigr)^k}.
\]
Because we have \(|1-1/\eta|\le 1/\eta\), so
\(
a_k \le \eta^2 \Paren{\frac{C\theta^2}{\eta^2}}^k
\)
after adjusting the constant \(C\).
\end{proof}

\begin{proof}[Proof of \cref{lem:ldlr-xcorruption}]
By \cref{lem:whitened-xcorruption}, it suffices to analyze the problem after rescaling via $T$.
For \(k\ge 2\), define
\[
\phi_{u,k}(x,y)\defeq h_k(\langle x,u\rangle)\psi_k(y).
\]
Then the one-sample likelihood ratio is
\[
\ell_u(x,y)=1+\sum_{k\ge 2}\phi_{u,k}(x,y).
\]

Let \(u,v\in \mathbb S^{d-1}\), and let \(\gamma\defeq \langle u,v\rangle\).
Under the whitened null, \(x\) and \(y\) are independent standard Gaussians, so by orthogonality,
\begin{equation}
\label{eq:phi-orthogonality}
\langle \phi_{u,k},\phi_{v,\ell}\rangle_{L_2(Q)}
=
\mathbf 1_{\{k=\ell\}}\,\gamma^k a_k.
\end{equation}
Also, for independent \(u,v\sim \mathrm{Unif}(\mathbb S^{d-1})\),
\begin{equation}
\label{eq:overlap-moment-here}
\E[|\gamma|^m] \le \Paren{\frac{Cm}{d}}^{m/2}
\end{equation}
for all integers \(m\ge 1\).
Let
\(
K\defeq \lfloor D/2\rfloor.
\)
Since each \(\phi_{u,k}\) has total degree \(2k\), the \(n\)-sample degree-\(D\) projection keeps
exactly those terms for which the sum of the chosen hidden degrees is at most \(K\).  Expanding the
\(n\)-sample likelihood ratio and using \eqref{eq:phi-orthogonality}, we obtain
\begin{equation}
\label{eq:ldlr-master}
\Adv_{\le D}^2
\le
1+\sum_{\ell\ge 1}\binom{n}{\ell}
\sum_{\substack{k_1,\dots,k_\ell\ge 2\\k_1+\cdots+k_\ell\le K}}
\Paren{\prod_{j=1}^\ell a_{k_j}}
\E\bigl[|\gamma|^{k_1+\cdots+k_\ell}\bigr].
\end{equation}

Set
\(
S \defeq n\alpha^4/d\eta^2.
\)
Since \(\alpha^2\le \eta\), we have \(s=\Theta(1)\) and hence \(\theta=\Theta(\alpha)\).  Therefore
\cref{lem:psi-coefficient-bound} implies
\begin{equation}
\label{eq:ak-bound-delta}
a_k \le \eta^2 \Paren{\frac{C\delta^2}{\eta^2}}^k
\qquad (k\ge 2).
\end{equation}

We split the sum in \eqref{eq:ldlr-master} into the \(\ell=1\) part and the \(\ell\ge 2\) part.

We first bound the contribution from terms where \(\ell=1\).
For \(m\ge 2\), the \(\ell=1\) term is
\(
n a_m \E[|\gamma|^m].
\)
For \(m=2\), using \eqref{eq:ak-bound-delta} and \eqref{eq:overlap-moment-here},
\(
n a_2 \E[\gamma^2] \le C \frac{n\alpha^4}{d\eta^2} = C S.
\)

Now fix \(m\ge 4\).  Using \eqref{eq:ak-bound-delta} and \eqref{eq:overlap-moment-here},
\begin{align*}
n a_m \E[|\gamma|^m]
&\le
n \eta^2 \Paren{\frac{C\delta^2}{\eta^2}}^m \Paren{\frac{Cm}{d}}^{m/2} \\
&=
S\cdot
\Paren{\frac{\delta}{\eta}}^{2m-4}
\cdot
(Cm)^{m/2}\,
d^{1-m/2}.
\end{align*}
When \(m\le K\le d^{10^{-4}}\), we have $(\delta/\eta)^4 \cdot m/d\leq O(1)$.
Therefore
\[
n a_m \E[|\gamma|^m]
\le
S\cdot \Paren{\Bigl(\frac{\delta}{q}\Bigr)^4 \cdot \frac{m}{d}}^{m/2-1}
\le
S\cdot d^{-0.001}\,.
\]
for every \(m\ge 4\).  Summing over \(m=4,\dots,K\), we conclude that the entire \(\ell=1\)
contribution is \(O(S)\).

For the contribution from terms where $\ell \geq 2$.
Fix \(m\in\{4,\dots,K\}\).  The number of compositions of \(m\) into \(\ell\) parts each at least \(2\) is
\[
\binom{m-\ell-1}{\ell-1}\le 2^m.
\]
Using this, together with \eqref{eq:ak-bound-delta} and \eqref{eq:overlap-moment-here}, the
contribution of all terms with total hidden degree \(m\) and \(\ell\ge 2\) is at most
\begin{align*}
&\E[|\gamma|^m]\sum_{\ell=2}^{\lfloor m/2\rfloor}
\binom{n}{\ell}\binom{m-\ell-1}{\ell-1}
\eta^{2\ell}\Paren{\frac{C\delta^2}{\eta^2}}^m \\
&\le
\Paren{\frac{Cm}{d}}^{m/2}
\cdot
2^m
\sum_{\ell=2}^{\lfloor m/2\rfloor}
(n \eta^2)^\ell
\Paren{\frac{C\delta^2}{\eta^2}}^m.
\end{align*}
Since \(n\eta^2\ge 1\) in the regime of interest, and \(\ell\le m/2\),
\(
\sum_{\ell=2}^{\lfloor m/2\rfloor}(n\eta^2)^\ell
\le
m(n\eta^2)^{m/2}.
\)
Hence the total \(\ell\ge 2\) contribution is bounded by
\[
m(Cm)^m \Paren{\frac{n\delta^4}{d\eta^2}}^{m/2}
=
m(Cm)^m S^{m/2}.
\]
If \(S\le c/K^2\) for a sufficiently small universal constant \(c>0\), then
\[
m(Cm)^m S^{m/2}\le 2^{-m}
\qquad\text{for all }m\le K.
\]
Therefore the sum over all \(m\ge 4\) and \(\ell\ge 2\) is \(O(1)\).
Combining the \(\ell=1\) and \(\ell\ge 2\) estimates, we obtain
\(
\Adv_{\le D}^2 \le 1 + O(S) + O(1).
\)
If
\(
n \le c\,\frac{d\eta^2}{D^2\delta^4},
\)
then \(S\le c/D^2\), and hence \(\Adv_{\le D}=O(1)\).
\end{proof}

\subsection{Hardness of private Bayesian regression}\label{sec:private-lr-lb}

Now we can apply the privacy-to-robustness reduction of~\cite{Georgiev2022PrivacyIR} to obtain the claimed information-computation gap for private Bayesian regression, \cref{thm:privacy-lower-bound-regression}.

\restatetheorem{thm:privacy-lower-bound-regression}

\begin{remark}\label{remark:gap}
Within the low-degree framework, we give a computational threshold for error rate $\alpha_{\mathrm{comp}}=\Paren{\frac{d\log^2(1/\beta)}{n^3 \e^2}}^{1/4}$.
In comparison, by the black-box privacy reduction of~\cite{Georgiev2022PrivacyIR} combined with our exponential-time robust algorithm, we know that exponential time algorithms can achieve error rate $\alpha\leq \tilde{O}\Paren{\frac{d+\log(1/\beta)}{n\e}}$.
Therefore, the error achievable by exponential time algorithm is given by $\alpha_{\mathrm{IT}}=\tilde{O}(d/n\epsilon)$.
As a result, even for constant failure probability $0.001$, there is a separation between $\alpha_{\mathrm{comp}}$ and $\alpha_{\mathrm{IT}}$ when 
$\epsilon^2\geq d^3/n$. 
\end{remark}

\begin{proof}[Proof of \cref{thm:privacy-lower-bound-regression}]
We argue by contradiction and apply the privacy-to-robustness reduction from
\cref{thm:privacy-robustness-reduction}.

Let $\mu_{\post} \coloneqq \E[\wnull \mid \Xnull, \ynull]$ denote the posterior mean for Bayesian linear regression under the stated model.
Assume there exists a polynomial-time $\epsilon$-DP algorithm $\mathcal{A}$ which, given $n$ i.i.d.\ samples,
outputs an estimate $\widehat{\theta}$ satisfying
\[
\Pr\bigl[\ \|\widehat{\theta}-\mu_{\post}\|\le O(\alpha)\ \bigr]\;\ge\;1-\beta.
\]
Applying \cref{thm:privacy-robustness-reduction} to the target estimator $\mu_{\post}$, we obtain a
polynomial-time $\eta$-robust algorithm $\mathcal{A}_{\mathrm{rob}}$ that achieves error $\alpha$ with probability
at least $1-\beta$ against an $\eta$-fraction of corruptions, where
\(
\eta \;=\; O\left(\log(1/\beta)/{n\epsilon}\right).
\)
By \cref{lem:reduction-xcorruption}, if there exists an efficient
$\eta$-robust regression algorithm with error $\alpha$ using
\(
n \;=\; \tilde{O}\left(\frac{d\,\eta^2}{\alpha^4 D}\right)
\)
samples, then one can efficiently solve the distinguishing problem $\Pi$
from \cref{def:variance-matched-xcorruption-test} in as many samples. But \cref{lem:ldlr-xcorruption} implies that no degree-$D$ estimator can do this in that many samples. Hence, for any such efficient
$\epsilon$-DP algorithm $\mathcal{A}$, it must be that
$n = \tilde{\Omega}\left(\frac{d\,\eta^2}{\alpha^4 D}\right)$.

Substituting $\eta = O(\log(1/\beta)/(n\epsilon))$ and rearranging gives $n = \tilde{\Omega}(\frac{d}{\alpha^4 D}\cdot \frac{\log^2(1/\beta)}{n^2\epsilon^2})$, from which 
the claimed lower bound follows.
\end{proof}

%% file: content/appendix-mean-estimation.tex
\section{Deferred Proofs From \cref{sec:gaussian-posterior-mean}}

\subsection{From empirical mean to posterior mean}
\label{app:bayesian-me-from-emp-me}

In this section we fill in the details for how to deduce \cref{thm:mean_estimation_upper_informal} (private posterior mean estimation) from \cref{thm:anisotropic-private-est} (private empirical mean estimation).

\begin{proof}[Proof of the Inefficient Upper Bound in \Cref{thm:mean_estimation_upper_informal}]
    We will apply the inefficient algorithm in \Cref{thm:anisotropic-private-est} with
    \[
    R = \sqrt{\Bigr(\|\Sigma\|_\text{op} + \frac{1}{n}\Bigl)\Bigl(d + \log \frac{1}{\beta}\Bigr)}
    \]
    so that the antecedent $\|\bar x\|_2 \le R$ holds with probability $1-O(\beta)$. Then if we require
    \[
        n \ge \frac{d \log \frac{R}{\alpha}}{\varepsilon} = \tilde \Theta\left(\frac{d \log \frac{\|\Sigma\|_\text{op}}{\alpha}}{\varepsilon}\right)
    \]
    and apply the rescaling reduction, we get the desired error
    \[
        \alpha \le \tilde O\left(\frac{\text{tr}(\Lambda) + \|\Lambda\|_\text{op} \log \frac{1}{\beta}}{\varepsilon n}\right) = \tilde O\left(\frac{\text{tr}\bigl(I + \frac{1}{n} \Sigma^{-1}\bigr)^{-1} +\bigl \|\bigl(I + \frac{1}{n} \Sigma^{-1}\bigr)^{-1}\bigr\|_\text{op} \log \frac{1}{\beta}}{\varepsilon n}\right) \qedhere
    \]
\end{proof}

\begin{proof}[Proof of the Efficient Upper Bound in \Cref{thm:mean_estimation_upper_informal}]
    Along the same lines as the previous proof, we apply the efficient algorithm in \Cref{thm:anisotropic-private-est} with the same $R$. Then assuming $n \ge \tilde \Theta\left(\frac{d \log \|\Sigma\|_\text{op}}{\varepsilon}\right)$ we have
    \begin{align*}
        \alpha &\le \tilde O\left(\left(\frac{\text{tr}(\Lambda^{4/3}) + \|\Lambda^{4/3}\|_\text{op} \log \frac{1}{\beta}}{n \varepsilon^{2/3}}\right)^{3/4} + \frac{\text{tr}(\Lambda)+ \left\|\Lambda\right\|_\text{op} \log \frac{1}{\beta}}{\varepsilon n}\right)
    \end{align*}
    as desired.
\end{proof}

\subsection{From isotropic to anisotropic via bucketing}
\label{app:aniso-bucketing-argument}

Here we will deduce \Cref{thm:anisotropic-private-est} from the special case of isotropic prior that was proven in Section~\ref{sec:gaussian-posterior-mean}.

\begin{proof}[Proof of the Inefficient Variant of \Cref{thm:anisotropic-private-est}]
    Our plan is to break $\Lambda^2$ apart by eigenvalue into buckets of for $\sigma^2_i =  2^{-i} \cdot \|\Lambda^2\|_\text{op}$, and run the isotropic estimator of \cref{thm:anisotropic-private-est} with privacy parameter $\frac{\varepsilon}{M+1}$ on the first $M+1$ buckets and then on all of the remaining buckets output $0$. Formally, let $\alpha_i$ be the error on bucket $\sigma^2_i$, $\Lambda^2_i$ be the covariance matrix restricted to bucket $i$, and $d_i$ be the size of bucket $i$.
    \begin{equation*}
        \alpha \le \sum_{i=0}^\infty \alpha_i \le \sum_{i=1}^{M} \alpha_i + 2^{-M/2} \cdot d \cdot \left\|\Lambda^2\right\|_\text{op} \cdot \frac{d + \log \frac{1}{\beta}}{\varepsilon n}
    \end{equation*}
    Let's ignore the second term for now (and set $M$ large enough later to make it negligible). The first term, using a QM-AM inequality, is at most (up to constant factors)
    \begin{align*}
        \sum_{i=0}^M \alpha_i &\le \sum_{i=1}^M \sigma_i (M+1) \cdot \frac{d_i + \log \frac{1}{\beta}}{\varepsilon n} \\
        &\le (M+1) \sum_{i=1}^M \sigma_i \cdot \frac{d_i + \log \frac{1}{\beta}}{\varepsilon n} \\
        &\le  (M+1) \cdot \frac{\text{tr}(\Lambda) +\left\|\Lambda\right\|_\text{op} \log \frac{1}{\beta}}{\varepsilon n}
    \end{align*}
    where we used the inequalities $\sum_{i=0}^M \sigma_i^2 \le 2 \|\Lambda^2\|_\text{op}$ and $\sum_{i=0}^M \sigma_i \le 2 \|\Lambda\|_\text{op}$. Finally, to bound the remaining buckets it suffices to set $M=\Omega(\text{polylog}(n, d, \frac{1}{\varepsilon}, \log \frac{1}{\beta}))$, completing the argument.
\end{proof}

\begin{proof}[Proof of the Efficient Variant of \Cref{thm:anisotropic-private-est}]
    We use a similar approach as our proof for the efficient case above. Analogous to before, 
    \begin{equation*}
        \alpha \le \sum_{i=0}^\infty \alpha_i \le \sum_{i=1}^{M} \alpha_i + 2^{-M/2} \cdot d \cdot \left\|\Lambda^2\right\|_\text{op} \cdot \left(\left(\frac{d + \log \frac{1}{\beta}}{n \varepsilon^{2/3}}\right)^{3/4} + \frac{d + \log \frac{1}{\beta}}{\varepsilon n}\right)\,.
    \end{equation*}
    Let's ignore the second term for now (and set $M$ large enough later to make it negligible). The first term, using the concavity of $f(x)=x^{3/4}$, is at most (up to constant factors)
    \begin{align*}
        \sum_{i=0}^M \alpha_i &\le \sum_{i=1}^M \left[\sigma_i \left((M+1)^{2/3}\frac{d_i + \log \frac{1}{\beta}}{n \varepsilon^{2/3}}\right)^{3/4} + \sigma_i (M+1) \cdot \frac{d_i + \log \frac{1}{\beta}}{\varepsilon n}\right] \\
        &\le M^{1/4} \sqrt {M+1} \cdot \left(\sum_{i=1}^M \sigma_i^{4/3} \cdot \frac{d_i + \log \frac{1}{\beta}}{n \varepsilon^{2/3}}\right)^{3/4} + (M+1) \sum_{i=1}^M \sigma_i \cdot \frac{d_i + \log \frac{1}{\beta}}{\varepsilon n} \\
        &\le M^{1/4} \sqrt{M+1} \cdot \left(\frac{\text{tr}(\Lambda^{4/3}) + \|\Lambda^{4/3}\|_\text{op} \log \frac{1}{\beta}}{n \varepsilon^{2/3}}\right)^{3/4} + (M+1) \cdot \frac{\text{tr}(\Lambda) +\left\|\Lambda\right\|_\text{op} \log \frac{1}{\beta}}{\varepsilon n}
    \end{align*}
    where we used the inequalities $\sum_{i=0}^M \sigma_i^{4/3} \le 2 \|\Lambda^{4/3}\|_\text{op}$ and $\sum_{i=0}^M \sigma_i \le 2 \|\Lambda\|_\text{op}$. Finally, to bound the remaining buckets it suffices to set $M=\Omega(\text{polylog}(n, d, \frac{1}{\varepsilon}, \log \frac{1}{\beta}))$, completing the argument.
\end{proof}

\subsection{Proof of \cref{lem:stability-Gaussian}: resilience bound for feasibility}
\label{sec:mean_concentration}

In this section we prove the resilience bound needed to show that the ground truth is feasible for the main program $\cA$ in \cref{sec:apply_kmz}:

\restatelemma{lem:stability-Gaussian}
\begin{proof}
Fix $b\in[0,1]^n$ satisfying $\E_i b_i \ge 1-\eta$. Let $v\in \mathbb{S}^{d-1}$ be arbitrary.
Using $\bar{x}=\E_i x_i$, we have
\begin{align*}
\E_i b_i \langle x_i-\bar{x},v\rangle
&=
\E_i b_i \langle x_i,v\rangle
-
(\E_i b_i)\,\E_j \langle x_j,v\rangle \\
&=
\E_i \big(b_i-(\E_j b_j)\big)\langle x_i,v\rangle.
\end{align*}
Decompose $x_i=\mu+(x_i-\mu)$. Since $\E_i\big(b_i-(\E_j b_j)\big)=0$, the $\mu$ term cancels and
\[
\E_i b_i \langle x_i-\bar{x},v\rangle
=
\E_i \big(b_i-(\E_j b_j)\big)\langle x_i-\mu,v\rangle.
\]
Now note that $u = \E_i (b_i - (\E_j b_j)) (x_i - \mu)$ is distributed as a Gaussian centered at zero and variance $\E_i (b_i - (\E_j b_j))^2 I_d$ (in Loewner order). Now note that we can write
\begin{align*}
\E_i (b_i - \E_j b_j)^2 
&= 
\E_i b_i^2 - (\E_i b_i)^2 \\
&\le 1 - (1-\eta)^2 \\
&= 2 \eta - \eta^2 \\
&\le 2\eta \mper
\end{align*}
Therefore, from the concentration of the $\ell_2$ norm of a Gaussian centered at $0$ with variance at most $O(\eta) I_d$,  namely $u$,we conclude that with probability at $1-\beta$, for all $v \in \mathbb{S}^{d-1}$ we have that
\begin{equation*}
\abs{\E_i b_i \iprod{x_i - \bar{x}, v}} \le O\Paren{\sqrt{\frac{\eta (d + \log(1/\beta))}{n}}} \mper
\end{equation*}
Now in order to extend this to all $b$, it suffices to take a union bound over all integral points $b \in \{0,1\}^n$ such that $\E_i b_i \ge 1-\eta$, since those can be the extreme points. Taking a union bound over the ${n \choose\eta n} \le (e/\eta)^{\eta n}\le \exp(O(\eta \log(1/\eta) n))$ many such $b$'s taking a union bound over them gives us that with probability $1-\beta$, for all vectors $v \in \mathbb{S}^{d-1}$ and $b \in [0,1]^n$ such that $\E_i b_i \ge 1-\eta$, we have 
\begin{equation*}
\abs{\E_i b_i \iprod{x_i - \bar{x}, v}} \le O\Paren{\sqrt{\frac{\eta (d + \log(1/\beta))}{n}} 
+ \eta \sqrt{\log(1/\eta)}
} \mper
\end{equation*}
as desired.

Now we prove the second part of the lemma. Fix $b \in \set{0,1}^n$ satisfying $\E_i b_i \ge 1-\eta$. We can write
\begin{equation*}
\Abs{\E_i b_i \iprod{x_i - \bar{x}, v}^2 - 1} \le \Abs{\E_i \iprod{x_i - \bar{x}, v}^2 - 1} + \E_i (1-b_i) \iprod{x_i - \bar{x}, v}^2 \mper 
\end{equation*}
We bound each term separately.
For the first term we can write
\begin{align*}
\Abs{\E_i \iprod{x_i - \bar{x}, v}^2 - 1}
&= \Abs{\E_i \iprod{x_i - \mu, v}^2 - \iprod{\bar{x} - \mu, v}^2 -1} \\
&\le 
\Abs{\E_i \iprod{x_i - \mu, v}^2 - 1} +\iprod{\bar{x} - \mu, v}^2 \\
&\le 
\Normop{\E_i (x_i - \mu)\transpose{(x_i - \mu)} - I_d} + \Snormt{\bar{x} - \mu} \mper
\end{align*}
From standard concentration bounds of a standard Gaussian we know that with probability $1-\beta$,
\begin{equation*}
\Normop{\E_i (x_i - \mu)\transpose{(x_i - \mu)} - I_d} \le O\Paren{\sqrt{\frac{d+ \log(1/\beta)}{n}} + \frac{d + \log(1/\beta)}{n}}, \quad \Snormt{\bar{x} - \mu} \le O\Paren{\frac{d + \log(1/\beta)}{n}} \mper
\end{equation*}
Therefore, with probability $1-\beta$, for all $v \in \mathbb{S}^{d-1}$ we obtain
\begin{equation*}
\Abs{\E_i \iprod{x_i - \bar{x}, v}^2 - 1} \le 
O\Paren{\sqrt{\frac{d+ \log(1/\beta)}{n}}
+
\frac{d + \log(1/\beta)}{n}
} \mper
\end{equation*}
Now we bound the second term. Adding and subtracting $\mu$, we can write
\begin{align*}
\E_i (1- b_i) \iprod{x_i - \bar{x}, v}^2 &\le
2 \E_i (1-b_i) \iprod{x_i - \mu, v}^2 + 2 \E_i (1-b_i) \iprod{\mu - \bar{x},v}^2 \mper
\end{align*}
The second term we can show is bounded from above by $O\Paren{\frac{\eta (d + \log(1/\beta))}{n}}$, with probability $1-\beta$, by applying standard bounds on concentration of a Gaussian and the assumption that $\E_i b_i \ge 1-\eta$. It remains to bound the first term. Again applying the standard concentration bound for the empirical covariance we can write that with probability $1-\beta$
\begin{equation*}
\E_i (1-b_i) \iprod{x_i - \mu, v}^2 \le \eta \cdot \Brac{1 + O\Paren{\sqrt{\frac{d+ \log(1/\beta)}{n}}
+
\frac{d + \log(1/\beta)}{n}}} \mper
\end{equation*}
Overall for a fixed $b$ we obtain that with probability $1-\beta$,
\begin{equation*}
\Abs{\E_i b_i \iprod{x_i - \bar{x}, v}^2 - 1} \le 
O\Paren{\sqrt{\frac{d+ \log(1/\beta)}{n}}
+
\frac{d + \log(1/\beta)}{n} + \eta} \mper
\end{equation*}
Taking a union bound over the choice of $\beta$ gives us the desired bound.

\end{proof}

\input{content/kmz-lemma}

%% file: content/kmz-lemma.tex
\subsection{Isotropic estimation lemma}
\label{sec:isotropic-estimation-lemma}
Here we will state and prove the following lemma, which is an adaption of Lemma~4.1 in \cite{kothari2022polynomial} and which is a crucial ingredient in the proof in \cref{sec:apply_kmz}. As we will only require the lemma in the case where $V(\mu_0, v) = 1$, that's the setting where we state and prove here as well. We will follow their strategy exactly, the only difference is that we adapt it to the setting where we do not necessarily have the resilience guarantees with $O(\eta)$ error and instead might have weaker as is the case in the $n \ll d / \eta^2$ regime for Gaussian mean estimation.

\begin{lemma}
\label{lem:technicallem-unit-variance}
Let \(x_1,\dots,x_n\in\mathbb{R}^d\), and fix a reference mean \(\mu_0\in\mathbb{R}^d\). Assume the following \emph{robust moment bounds}: for every \(v\in S\) and every weight vector \(a=(a_1,\dots,a_n)\in[0,1]^n\) with \(\sum_{i=1}^n a_i \ge (1-\eta)n\),
\begin{equation}
\left|\frac{1}{n}\sum_{i=1}^n a_i  \langle x_i-\mu_0,\ v\rangle\right| \ \le\ \alpha_1,
\qquad
\left|\frac{1}{n}\sum_{i=1}^n a_i\big(\langle x_i-\mu_0,\ v\rangle^2-1\big)\right| \ \le\ \alpha_2.
\label{eq:robust-assumptions-unit}
\end{equation}

Let \(y_1,\dots,y_n\) be an \(\eta\)-corruption of \(x_1,\dots,x_n\). Let \(  \pE\) be a pseudo-expectation of sufficiently high degree over formal variables \(x'_1,\dots,x'_n\in\mathbb{R}^d\) and \(w_1,\dots,w_n\), and let \(\mu' := \frac{1}{n}\sum_{i=1}^n x'_i\). Suppose:
\begin{enumerate}
\item \(  \pE[w_i^2=w_i]\) for all \(i\in[n]\),
\item \(  \pE\left[\sum_{i=1}^n w_i\right] = (1-\eta)n\),
\item \(  \pE[  w_i x'_i = w_i y_i  ]\) for all \(i\in[n]\),
\item for every \(v\in S\), \(\displaystyle   \pE \left[\frac{1}{n}\sum_{i=1}^n \langle x'_i-\mu',\ v\rangle^2\right] \le 1+\alpha_0\).
\end{enumerate}
Define \(w'_i := w_i  \mathbf{1}\{x_i=y_i\}\) and set \(a_i :=   \pE[w'_i]\) for \(i\in[n]\). Then \(a_i\in[0,1]\) and \(\frac{1}{n}\sum_i a_i \ge 1-2\eta\). For every \(v\in S\), writing \(\widehat{\mu} :=   \pE[\mu']\), the following hold:
\begin{align}
  \pE\big[\langle \mu'-\mu_0,\ v\rangle^2\big]
&\ \le\ \frac{4\eta(\alpha_0+\alpha_2)+2\alpha_1^2}{1-4\eta}
\ \ \le\ 8\eta(\alpha_0+\alpha_2)+4\alpha_1^2 \quad (\eta\le\tfrac18),
\label{eq:key-second-moment-unit}
\\[6pt]
\Big|\langle \widehat{\mu}-\mu_0,\ v\rangle\Big|
&\ \le\ \alpha_1\ +\ \sqrt{  2\eta\Big(\alpha_0+\alpha_2+  \pE\big[\langle \mu'-\mu_0,\ v\rangle^2\big]\Big)}.
\label{eq:key-bias-unit}
\end{align}
In particular, combining \eqref{eq:key-second-moment-unit} and \eqref{eq:key-bias-unit} gives the bound
\begin{equation}
\Big|\langle \widehat{\mu}-\mu_0,\ v\rangle\Big|
\ \le\ \alpha_1\ +\ \sqrt{  C_1  \eta(\alpha_0+\alpha_2)\ +\ C_2  \eta  \alpha_1^2\ +\ C_3  \eta^2(\alpha_0+\alpha_2)  }
\ \le\ C\Big(\alpha_1+\sqrt{\eta(\alpha_0+\alpha_2)}+\eta\Big),
\label{eq:clean-rate-unit}
\end{equation}
for universal constants \(C_1,C_2,C_3,C>0\). Thus the rate
\[
\alpha_1\ +\ \sqrt{  \eta(\alpha_0+\alpha_2)\ +\ \eta^2  }
\]
is valid up to universal constant factors.
\end{lemma}

\begin{proof}[Proof of Lemma~\ref{lem:technicallem-unit-variance}]
Define \(w'_i := w_i  \mathbf{1}\{x_i=y_i\}\). Then \(  \pE[{w'_i}^2=w'_i]\), \(  \pE[w'_i x'_i = w'_i x_i]\), and \(\frac{1}{n}\sum_{i=1}^n   \pE[1-w'_i] \le 2\eta\). Set \(a_i:=  \pE[w'_i]\); then \(a_i\in[0,1]\) and \(\frac{1}{n}\sum_i a_i\ge 1-2\eta\).

\medskip
\noindent\emph{Step 1: Decompose the estimation error.}
For any \(v\in S\),
\begin{align*}
\langle \widehat{\mu}-\mu_0,\ v\rangle
&=   \pE\left[\frac{1}{n}\sum_{i=1}^n \langle x'_i-\mu_0,\ v\rangle\right]
\\
&=   \pE\left[\frac{1}{n}\sum_{i=1}^n (1-w'_i)  \langle x'_i-\mu_0,\ v\rangle\right]
  +   \frac{1}{n}\sum_{i=1}^n   \pE[w'_i]\ \langle x_i-\mu_0,\ v\rangle.
\end{align*}
Hence,
\begin{equation}
\Big|\langle \widehat{\mu}-\mu_0,\ v\rangle\Big|
\ \le\ \underbrace{\left|  \pE\left[\frac{1}{n}\sum_{i=1}^n (1-w'_i)  \langle x'_i-\mu_0,\ v\rangle\right]\right|}_{\text{(A)}}  +  \underbrace{\left|\frac{1}{n}\sum_{i=1}^n a_i \langle x_i-\mu_0,\ v\rangle\right|}_{\le\ \alpha_1\ \text{ by \eqref{eq:robust-assumptions-unit}}}.
\label{eq:first-decomp}
\end{equation}
Applying Cauchy--Schwarz to (A) and \({w'_i}^2=w'_i\), we get
\begin{align}
\text{(A)}^2
&\ \le\   \pE \left[\left(\frac{1}{n}\sum_{i=1}^n (1-w'_i)  \langle x'_i-\mu_0,\ v\rangle\right)^{ 2}\right]
\ \le\   \pE \left[\frac{1}{n}\sum_{i=1}^n (1-w'_i)\right]\cdot
  \pE \left[\frac{1}{n}\sum_{i=1}^n (1-w'_i)  \langle x'_i-\mu_0,\ v\rangle^2\right]
\nonumber\\
&\ \le\ 2\eta\cdot   \pE \left[\frac{1}{n}\sum_{i=1}^n (1-w'_i)  \langle x'_i-\mu_0,\ v\rangle^2\right].
\label{eq:A-CS}
\end{align}

\noindent Substituting \(w'_i x'_i = w'_i x_i\),
\begin{align}
  \pE \left[\frac{1}{n}\sum_{i=1}^n (1-w'_i)  \langle x'_i-\mu_0,\ v\rangle^2\right]
&=   \pE \left[\frac{1}{n}\sum_{i=1}^n \langle x'_i-\mu_0,\ v\rangle^2\right]
  -   \frac{1}{n}\sum_{i=1}^n a_i  \langle x_i-\mu_0,\ v\rangle^2
\nonumber\\
&\le\ \underbrace{  \pE \left[\frac{1}{n}\sum_{i=1}^n \langle x'_i-\mu',\ v\rangle^2\right]}_{\le\ 1+\alpha_0}
  +     \pE \left[\langle \mu'-\mu_0,\ v\rangle^2\right]
  -   \underbrace{\frac{1}{n}\sum_{i=1}^n a_i  \langle x_i-\mu_0,\ v\rangle^2}_{\ge\ 1-\alpha_2}
\nonumber\\
&\le\ \alpha_0+\alpha_2\ +\   \pE \left[\langle \mu'-\mu_0,\ v\rangle^2\right].
\label{eq:error-mass-second-moment}
\end{align}
Combining \eqref{eq:A-CS} and \eqref{eq:error-mass-second-moment},
\begin{equation}
\text{(A)}^2\ \le\ 2\eta\Big(\alpha_0+\alpha_2+  \pE\big[\langle \mu'-\mu_0,\ v\rangle^2\big]\Big).
\label{eq:A-square-final}
\end{equation}
Inserting \eqref{eq:A-square-final} into \eqref{eq:first-decomp} yields \eqref{eq:key-bias-unit}.

\medskip
\noindent Now we bound \(  \pE[\langle \mu'-\mu_0,\ v\rangle^2]\).
Decompose
\[
\langle \mu'-\mu_0,\ v\rangle
= \frac{1}{n}\sum_{i=1}^n (1-w'_i)  \langle x'_i-\mu_0,\ v\rangle
  +   \frac{1}{n}\sum_{i=1}^n a_i  \langle x_i-\mu_0,\ v\rangle.
\]
By \((x+y)^2\le 2x^2+2y^2\), \eqref{eq:A-square-final}, and \eqref{eq:robust-assumptions-unit},
\begin{align*}
  \pE\big[\langle \mu'-\mu_0,\ v\rangle^2\big]
&\le 2  \text{(A)}^2 + 2\left(\frac{1}{n}\sum_{i=1}^n a_i  \langle x_i-\mu_0,\ v\rangle\right)^{ 2}
\\
&\le 4\eta\Big(\alpha_0+\alpha_2+  \pE\big[\langle \mu'-\mu_0,\ v\rangle^2\big]\Big) + 2\alpha_1^2.
\end{align*}
Rearranging gives \eqref{eq:key-second-moment-unit}. Substituting \eqref{eq:key-second-moment-unit} into \eqref{eq:key-bias-unit} and simplifying gives us \eqref{eq:clean-rate-unit}.
\end{proof}

%% file: content/appendix-frequentist-implications.tex
\section{Frequentist Implications}\label{appendix:frequentist-implications}

In this section we prove the frequentist consequences of our main results, as stated in \cref{thm:mean_estimation_upper_informal_frequentist} and \cref{thm:lr-freq}.

\subsection{Frequentist mean estimation}

\begin{proof}[Proof of \Cref{{thm:mean_estimation_upper_informal_frequentist}}]
      Except with probability $\frac{\beta}{3}$, the sample mean $\bar{x}$ satisfies
    \begin{equation*}
        \|\bar{x}\|_2 
        \le O\Bigl(R + \sqrt{\text{Tr}(\Lambda^2) + \log 1/\beta}\Bigr).
    \end{equation*}
    Condition on the above. By \Cref{thm:anisotropic-private-est}, whenever 
    \begin{equation*}
        n 
        \ge \tilde \Omega\left(\frac{d}{\varepsilon} \log\Bigl(\frac{R + \|\Lambda\|_F}{\alpha}\Bigr)\right)\,,
    \end{equation*}
    where the $\tilde \Omega$ hides factors of $\log \log \frac{1}{\beta}$, except with probability $\frac{\beta}{3}$ our inefficient estimator outputs $\hat \mu$ such that
    \begin{equation*}
        \|\bar x - \hat \mu\| 
        \le \tilde O\biggl(\frac{\tr(\Lambda) + \norm{\Lambda}_{\sf op} \log(1/\beta)}{\epsilon n}\biggr)\,,
    \end{equation*}
    and our efficient estimator outputs $\hat \mu$ such that
    \begin{equation*}
        \|\bar x - \hat \mu\| 
        \le \tilde O\biggl(\frac{\norm{\Lambda}_{4/3} + \norm{\Lambda}_{\sf op} \log^{3/4} (1/\beta)}{\epsilon^{1/2} n^{3/4}}+\frac{\tr(\Lambda) + \norm{\Lambda}_{\sf op} \log(1/\beta)}{\epsilon n}\biggr)
    \end{equation*}
    Finally, except with probability $\frac{\beta}{3}$, %
    \begin{equation*}
        \|\bar x - \mu\|_2 \le \sqrt{\frac{\|\Lambda\|_F^2 + \|\Lambda\|_\text{op}^2 \log \frac{3}{\beta}}{n}} = (1+o(1)) \sqrt{\frac{\|\Lambda\|_F^2 + \|\Lambda\|_{\text{op}}^2 \log \frac{1}{\beta}}{n}}.
    \end{equation*}
    where the $o(1)$ holds in regimes where the effective dimension $\|\Lambda\|_F/\|\Lambda\|_\text{op} \to \infty$ or $\beta \to 0$. The theorem follows by a union bound over all three events and a triangle inequality.
\end{proof}

\subsection{Frequentist linear regression}
In this part, we give the proof of \Cref{thm:lr-freq}.
We first give a lemma for the error rate achieved by ordinary least square estimator. 
\begin{lemma}\label{lem:ols_concentration_frequentist}
Let
\(
X := [x_1,\ldots,x_n] \in \mathbb{R}^{d\times n}\),
\(
y = X^\top w + \xi\),
\(\xi \sim \mathcal{N}(0,\Id_n).
\),
and let
\[
w_{\mathrm{OLS}} := (XX^\top)^{-1}Xy.
\]
For any $\beta \in (0,1/2)$, define 
\[
r_\beta := C\sqrt{\frac{d+\log(4/\beta)}{n}},
\]
where $C>0$ is a sufficiently large absolute constant. If $r_\beta \le 1/2$, then with
probability at least $1-\beta$,
\[
\|w_{\mathrm{OLS}} - w\|_2^2
\le
\frac{1}{(1-r_\beta)n}
\left(
d + 2\sqrt{d\log\frac{4}{\beta}} + 2\log\frac{4}{\beta}
\right).
\]
Consequently, if $r_\beta = o(1)$ then for some universal constant $c$,
\[
\|w_{\mathrm{OLS}} - w\|_2
\le
(1+o(1))\sqrt{\frac{d+c\log(1/\beta)}{n}}.
\]
\end{lemma}

\begin{proof}
Consider the random event that
\[
\mathcal{E}_{\mathrm{cov}}
:=
\left\{
\left\|
\frac{1}{n}XX^\top - \Id
\right\|_{\sf op}
\le r_\beta
\right\}.
\]
By Theorem C.1, after enlarging $C$ if necessary,
\(
\Pr(\mathcal{E}_{\mathrm{cov}}) \ge 1-\beta/2.
\)
On $\mathcal{E}_{\mathrm{cov}}$, in particular $XX^\top$ is invertible and
\(
(1-r_\beta)n\Id \preceq XX^\top \preceq (1+r_\beta)n\Id.
\)

Also,
\(
w_{\mathrm{OLS}} - w
=
(XX^\top)^{-1}X\xi.
\)
Conditional on $X$, this is a centered Gaussian in $\mathbb{R}^d$ with covariance
\(
\Sigma_X := (XX^\top)^{-1}.
\)
Therefore, by the standard Gaussian tail bound, conditional on $X$, with
probability at least $1-\beta/2$,
\[
\|w_{\mathrm{OLS}} - w\|_2^2
\le
\tr(\Sigma_X)
+
2\sqrt{\tr(\Sigma_X^2)\log\frac{2}{\beta}}
+
2\|\Sigma_X\|_{\sf op}\log\frac{2}{\beta}.
\]

On $\mathcal{E}_{\mathrm{cov}}$ we have
\[
\tr(\Sigma_X) \le \frac{d}{(1-r_\beta)n},
\qquad
\tr(\Sigma_X^2) \le \frac{d}{(1-r_\beta)^2n^2},
\qquad
\|\Sigma_X\|_{\sf op} \le \frac{1}{(1-r_\beta)n}.
\]
Substituting these bounds yields that on $\mathcal{E}_{\mathrm{cov}}$, with conditional
probability at least $1-\delta/2$,
\[
\|w_{\mathrm{OLS}} - w\|_2^2
\le
\frac{1}{(1-r_\beta)n}
\left(
d + 2\sqrt{d\log\frac{2}{\beta}} + 2\log\frac{2}{\beta}
\right).
\]
A union bound proves the first claim, and the second follows immediately from
$r_\beta=o(1)$.
\end{proof}
\begin{proof}[Proof of \Cref{thm:lr-freq}]
    Following similar proofs as \cref{thm:private-approximation-least-square} and \cref{lem:exponential-bayesian-error-regression}, with $\eta = 0$ and the radius bound $R$, whenever the assumptions of \Cref{thm:lr-freq} hold,
    with probability at least $1-\frac{\beta}{3}$ our inefficient estimator outputs $\hat w$
    such that
    \[
    \|\hat w - w_{\mathrm{OLS}}\|_2
    \le
    \widetilde O\left(
    \frac{d + \log(3/\beta)}{\epsilon n}
    \right)
    =
    \widetilde O\left(
    \frac{d + \log(1/\beta)}{\epsilon n}
    \right),
    \]
    and except with probability $\frac{\beta}{3}$ our efficient estimator outputs $\hat w$
    such that
    \[
    \|\hat w - w_{\mathrm{OLS}}\|_2
    \le
    \widetilde O\left(
    \frac{(d + \log(3/\beta))^{3/4}}{\epsilon^{1/2}n^{3/4}}
    +
    \frac{d + \log(3/\beta)}{\epsilon n}
    \right)
    =
    \widetilde O\left(
    \frac{(d + \log(1/\beta))^{3/4}}{\epsilon^{1/2}n^{3/4}}
    +
    \frac{d + \log(1/\beta)}{\epsilon n}
    \right).
    \]

    Finally, applying \Cref{lem:ols_concentration_frequentist} with
    $\delta = \beta/3$, except with probability $\beta/3$,
    \[
    \|w_{\mathrm{OLS}} - w\|_2^2
    \le
    \frac{1}{(1-r_\beta)n}
    \left(
    d + 2\sqrt{d\log\frac{12}{\beta}}
    + 2\log\frac{12}{\beta}
    \right),
    \]
    where
    \[
    r_\beta := C\sqrt{\frac{d+\log(12/\beta)}{n}} = o(1).
    \]
    In particular, if $\log(1/\beta)=o(d)$, then
    \[
    \|w_{\mathrm{OLS}} - w\|_2
    \le
    (1+o(1))\sqrt{\frac{d+\log(1/\beta)}{n}}.
    \]

    The theorem follows by a union bound over these three events and the triangle inequality
    between $w$, $w_{\mathrm{OLS}}$, and $\hat w$.
\end{proof}

%% file: content/appendix_LR.tex
\section{Deferred Proofs From \cref{sec:linear-regression}}
\label{app:regression}

\subsection{Concentration inequalities for feasibility}\label{sec:concentration}

Here we prove concentration inequalities that were used in the proof that the ground truth is feasible for the program $\cA$ in \cref{algo:posterior-mean-regression}.

\begin{theorem}\label{thm:covariance-concentration}
    Let $x_1,x_2,\ldots,x_n$ be i.i.d sampled from Gaussian distribution $N(0,\Sigma)$. 
    Then with probability at least $1-\beta$, we have 
    \begin{equation*}
        \Normop{\frac{1}{n}\sum_{i\in [n]} x_ix_i^\top-\Sigma}\leq \sqrt{\frac{\Tr(\Sigma)+\log(1/\beta)\normop{\Sigma}}{n}}\,.
    \end{equation*}
\end{theorem}

\begin{theorem}[Maximal subset of Gaussian coordinates]\label{thm:maximum-subset-sum-gaussian-correct}
Let $v = (v_1,\dots,v_d) \sim N(0, \Id_d)$, let $\eta \in (0,1)$, and let $\beta \in (0,1)$. 
Set $m := \eta d$. Then with probability at least $1-\beta$ (over the draw of $v$),
\begin{equation}\label{eq:max-subset-gaussian}
    \max_{\substack{S \subseteq [d]\\ |S| = m}} \;\sum_{i \in S} v_i^2 
    \;\le\;
    C \Bigl( \eta d \log \frac{e}{\eta} + \log \frac{1}{\beta} \Bigr),
\end{equation}
where $C>0$ is an absolute constant.
In words: simultaneously for every subset $S$ of coordinates of size $\eta d$, the sum of squares of those coordinates is at most the right-hand side.
\end{theorem}

\begin{proof}
Fix any subset $S \subseteq [d]$ of size $m = \eta d$. Since $v \sim N(0,I_d)$, the random variable
\[
    X_S := \sum_{i \in S} v_i^2
\]
has the $\chi^2$ distribution with $m$ degrees of freedom, i.e.\ $X_S \sim \chi^2_m$.
We will use the standard Laurent--Massart tail bound for the chi-square distribution: if $X \sim \chi^2_m$, then for every $x>0$,
\begin{equation}\label{eq:laurent-massart}
    \Pr\bigl( X \ge m + 2\sqrt{mx} + 2x \bigr) \le e^{-x}.
\end{equation}

Apply \eqref{eq:laurent-massart} to $X_S$ for a parameter $x>0$ to be chosen later:
\begin{equation}\label{eq:for-fixed-S}
    \Pr\bigl( X_S \ge m + 2\sqrt{mx} + 2x \bigr) \le e^{-x}.
\end{equation}

We want the upper bound to hold \emph{simultaneously} for every subset $S$ of size $m$. 
The total number of such subsets is
\begin{equation}\label{eq:count-subsets}
    \binom{d}{m} \le \Bigl(\frac{e d}{m}\Bigr)^m = \Bigl(\frac{e}{\eta}\Bigr)^{\eta d}.
\end{equation}
By the union bound, combining \eqref{eq:for-fixed-S} and \eqref{eq:count-subsets}, we obtain
\[
    \Pr\Bigl( \exists S \subseteq [d], |S|=m : X_S \ge m + 2\sqrt{mx} + 2x \Bigr)
    \;\le\;
    \binom{d}{m} e^{-x}
    \;\le\;
    \Bigl(\frac{e}{\eta}\Bigr)^{\eta d} e^{-x}.
\]
Choose
\begin{equation}\label{eq:choice-of-x}
    x := \eta d \log \frac{e}{\eta} + \log \frac{1}{\beta}.
\end{equation}
Then
\[
    \Bigl(\frac{e}{\eta}\Bigr)^{\eta d} e^{-x}
    =
    \exp\bigl( \eta d \log(e/\eta) \bigr)\;
    \exp\bigl( - \eta d \log(e/\eta) - \log(1/\beta) \bigr)
    =
    \beta.
\]
Hence, with probability at least $1-\beta$,
\begin{equation}\label{eq:all-subsets-good}
    \forall S \subseteq [d], |S|=m:\quad
    \sum_{i \in S} v_i^2
    \;\le\;
    m + 2\sqrt{m x} + 2x,
\end{equation}
where $x$ is given by \eqref{eq:choice-of-x}. 

It remains to simplify the right-hand side.
Recall that $m = \eta d$ and
\[
    x = \eta d \log\frac{e}{\eta} + \log\frac{1}{\beta}.
\]
Then
\[
    m + 2x \;\le\; 3\eta d \log\frac{e}{\eta} + 2 \log\frac{1}{\beta},
\]
and
\[
    2\sqrt{mx}
    = 2 \sqrt{\eta d \bigl( \eta d \log(e/\eta) + \log(1/\beta) \bigr)}
    \;\le\;
    \eta d \log\frac{e}{\eta} + \log\frac{1}{\beta},
\]
after enlarging constants (this is a standard inequality of the form $2\sqrt{ab} \le a+b$).
Therefore, for a universal constant $C>0$,
\[
    m + 2\sqrt{mx} + 2x
    \;\le\;
    C \Bigl( \eta d \log \frac{e}{\eta} + \log \frac{1}{\beta} \Bigr).
\]
Plugging this back into \eqref{eq:all-subsets-good} gives exactly \eqref{eq:max-subset-gaussian}.
\end{proof}

\begin{theorem}[Order statistic of a Gaussian vector with sharp $\beta$-dependence]\label{thm:gaussian-order-stat-sharp}
Let $v = (v_1,\dots,v_d) \sim N(0, \Id_d)$, and let $\eta, \beta \in (0,1)$.
Order the squared coordinates in nonincreasing order:
\[
    v_{(1)}^2 \ge v_{(2)}^2 \ge \dots \ge v_{(d)}^2.
\]
Then, with probability at least $1 - \beta$,
\begin{equation}\label{eq:gaussian-order-sharp}
    v_{(\lceil \eta d \rceil)}^2 
    \;\le\; 
    2 \log \frac{e}{\eta}
    \;+\;
    \frac{2}{\eta d} \log \frac{1}{\beta}.
\end{equation}
Equivalently: with probability at least $1-\beta$, \emph{at most} $\eta d$ coordinates of $v$ have squared value exceeding the right-hand side of \eqref{eq:gaussian-order-sharp}.
\end{theorem}

\begin{proof}
Fix $\eta \in (0,1)$ and set $m := \lceil \eta d \rceil$.
Let $T>0$ be a threshold to be chosen.
We want to bound the probability that there exist $m$ coordinates of $v$ all having squared value at least $T$.

For a single coordinate $v_i \sim N(0,1)$ we have
\[
    \Pr(v_i^2 \ge T) = \Pr(|v_i| \ge \sqrt{T}) \le e^{-T/2}.
\]
Now fix a subset $S \subseteq [d]$ with $|S| = m$. By independence,
\[
    \Pr\bigl( v_i^2 \ge T \text{ for all } i \in S \bigr)
    \;\le\;
    \bigl( e^{-T/2} \bigr)^{m}
    = e^{-(T/2) m}.
\]

There are $\binom{d}{m}$ such subsets $S$. By the union bound,
\begin{equation}\label{eq:union}
    \Pr\Bigl( \exists S \subseteq [d], |S|=m : v_i^2 \ge T \ \forall i \in S \Bigr)
    \;\le\;
    \binom{d}{m} \, e^{-(T/2) m}.
\end{equation}
We will upper bound the binomial coefficient by the standard entropy-like bound
\[
    \binom{d}{m} \le \Bigl( \frac{e d}{m} \Bigr)^m.
\]
Since $m = \lceil \eta d \rceil \ge \eta d$, we have
\[
    \binom{d}{m} \le \Bigl( \frac{e d}{\eta d} \Bigr)^m = \Bigl( \frac{e}{\eta} \Bigr)^m.
\]
Plugging into \eqref{eq:union} gives
\[
    \Pr\Bigl( \exists S: |S|=m, \ v_i^2 \ge T \ \forall i \in S \Bigr)
    \;\le\;
    \Bigl( \frac{e}{\eta} \Bigr)^m \, e^{-(T/2) m}
    = \exp\Bigl( m \bigl( \log \tfrac{e}{\eta} - \tfrac{T}{2} \bigr) \Bigr).
\]

We want the right-hand side to be at most $\beta$. Since $m \ge \eta d$, it suffices to require
\[
    \exp\Bigl( \eta d \bigl( \log \tfrac{e}{\eta} - \tfrac{T}{2} \bigr) \Bigr) \le \beta.
\]
Taking logs,
\[
    \eta d \bigl( \log \tfrac{e}{\eta} - \tfrac{T}{2} \bigr) \le \log \beta
    \quad\Longleftrightarrow\quad
    - \tfrac{T}{2} \le \frac{1}{\eta d} \log \beta - \log \tfrac{e}{\eta}.
\]
Multiply by $-1$ (and flip the inequality):
\[
    \frac{T}{2} \ge \log \tfrac{e}{\eta} + \frac{1}{\eta d} \log \frac{1}{\beta}.
\]
So it is enough to choose
\[
    T := 2 \log \frac{e}{\eta} + \frac{2}{\eta d} \log \frac{1}{\beta}.
\]
With this choice of $T$, the probability in \eqref{eq:union} is at most $\beta$.

But the event in \eqref{eq:union} is exactly: 
there exist $m$ coordinates all $\ge T$. 
Then the complement is given by at most $m-1 \le \eta d$ coordinates exceed $T$. 
Equivalently, in that complement event we have
\[
    v_{(m)}^2 \le T,
\]
i.e.
\[
    v_{(\lceil \eta d \rceil)}^2 \le 2 \log \frac{e}{\eta} + \frac{2}{\eta d} \log \frac{1}{\beta},
\]
with probability at least $1-\beta$, which is exactly \eqref{eq:gaussian-order-sharp}.
\end{proof}

\noindent Finally, we use the largest sparse singular values of the Gaussian submatrices.
\begin{lemma}[Sparse spectral norm of Gaussian submatrices]\label{lem:sparse-spectral-norm}
Let $X\in\R^{d\times n}$ have i.i.d.\ columns $x_1,\dots,x_n\sim\mathcal N(0,I_d)$.
For an index set $S\subseteq[n]$, let $X_S\in\R^{d\times |S|}$ denote the submatrix formed by the
columns indexed by $S$.  
For $k\in[n]$, recall the definition of $k$-sparse spectral norm(\cref{def:sparse-spectral-norm})
\[
\|X\|_{\mathrm{sp}}(k)\;\defeq\;\max_{\substack{S\subseteq[n]\\ |S|\le k}}\ \|X_S\|_{\op}\,.
\]
Then there exists a universal constant $C>0$ such that for every $k\in[n]$ and $\beta\in(0,1)$,
with probability at least $1-\beta$,
\[
\|X\|_{\mathrm{sp}}(k)
\;\le\;
C\left(
\sqrt d\;+\;\sqrt{k\log\frac{en}{k}}\;+\;\sqrt{\log\frac1\beta}
\right).
\]
Equivalently, on the same event, for all $S\subseteq[n]$ with $|S|\le k$,
\[
\|X_S\|_{\op}^2
\;\le\;
C^2\left(
d\;+\;k\log\frac{en}{k}\;+\;\log\frac1\beta
\right).
\]
\end{lemma}

\begin{proof}
Fix $s\in\{1,\dots,k\}$ and a set $S\subseteq[n]$ with $|S|=s$.  By a standard operator-norm tail
bound for a $d\times s$ Gaussian matrix, for all $t\ge 0$,
\begin{equation}\label{eq:gaussian-op-tail}
\Pr\left[\|X_S\|_{\op}\ge \sqrt d+\sqrt s+t\right]\le 2e^{-t^2/2}.
\end{equation}
There are at most $\binom{n}{s}\le (en/s)^s$ such subsets.  Set
\[
t_s \;\defeq\; \sqrt{2\left(s\log\frac{en}{s}+\log\frac{2k}{\beta}\right)}.
\]
By a union bound and \eqref{eq:gaussian-op-tail},
\[
\Pr\left[\max_{|S|=s}\|X_S\|_{\op}\ge \sqrt d+\sqrt s+t_s\right]
\le
\binom{n}{s}\cdot 2e^{-t_s^2/2}
\le
\frac{\beta}{k}.
\]
Taking a union bound over $s=1,\dots,k$ yields that with probability at least $1-\beta$, for all
$S$ with $|S|\le k$,
\[
\|X_S\|_{\op}
\le
\sqrt d+\sqrt{|S|}+t_{|S|}
\le
\sqrt d+\sqrt k+\sqrt{2k\log\frac{en}{k}}+\sqrt{2\log\frac{2k}{\beta}}.
\]
Absorbing $\sqrt k$ and $\sqrt{\log(2k/\beta)}$ into the displayed expression (by enlarging the
universal constant $C$) gives the stated bound.
\end{proof}

\subsection{Proof of \cref{lem:exponential-bayesian-error-regression}: Inefficient algorithm for robust and private Bayesian regression}
\label{sec:inefficientreg}

In this section, we first give an inefficient robust algorithm for Bayesian regression.
For simplicity, we base our algorithm on our previous robust and efficient linear regression estimator.
As in empirical mean estimation, we can exploit better resilience bound for approximating the posterior mean estimator(or equivalently ridge regression estimator) robustly.
Finally we use the black-box robustness to privacy reduction\cite{hopkins2023robustness, asi2023robustness} for developing a differentially private algorithm with optimal error guarantees.

\subsubsection{Inefficient algorithm for robust regression}
\label{app:exponential-regression}

\begin{theorem}\label{thm:robust-inefficient-linear-regression}
 Let $\Sigma=\sigma^2 \Id_d$ for known parameter $\sigma > 0$.
 Let $\beta \in [0,1]$.
 Then there is a computationally inefficient algorithm that, given $\eta$-corrupted samples $(\Xinput,\yinput)$ generated as in \cref{def:bayesian-regression-under-gaussian-prior}, outputs an estimator $\hat{w}$ such that with probability at least $1-\beta$,
 \begin{equation*}
     \norm{\hat{w}-\E[w^*\mid \Xnull,\ynull]}_2 \leq O\biggl(\eta\log(1/\eta)\biggr) \,,
 \end{equation*}
 where $(\Xnull,\ynull)$ are the uncorrupted samples, under the condition that $n\geq (d+\log(1/\beta))/\eta$..
 \end{theorem}

Throughout, $(\Xnull,\ynull)$ denotes the uncorrupted sample, and $(\Xinput,\yinput)$ denotes the
$\eta$-corrupted sample.  Let $I^\star\subseteq[n]$ be the (unknown) set of uncorrupted indices, so
$|I^\star|\ge (1-\eta)n$.

Assume $\Sigma=\sigma^2 \Id_d$ and define
\[
\lambda \defeq \frac{1}{n\sigma^2},\qquad
A^\star \defeq \frac{1}{n}\Xnull\Xnull^\top,\qquad
b^\star \defeq \frac{1}{n}\Xnull\ynull.
\]
Then the posterior mean is
\[
w_{\rm post}\defeq \E[w\mid \Xnull,\ynull] = (A^\star+\lambda \Id_d)^{-1} b^\star.
\]
For any $u\in\R^d$ define the (clean) residual vector
\[
r^\star(u)\defeq \ynull-\Xnull^\top u\in\R^n,
\qquad
g^\star(u)\defeq \frac{1}{n}\Xnull\, r^\star(u)\in\R^d.
\]
Note $g^\star(u)=b^\star-A^\star u$, and hence the following identity holds for every $u$:
\begin{equation}\label{eq:posterior-identity}
w_{\rm post}
=
u + (A^\star+\lambda \Id_d)^{-1}\bigl(g^\star(u)-\lambda u\bigr).
\end{equation}
Define failure probabilities $\beta_1=\beta_2=\beta_3=\beta/3$.

For a sequence $\{z_i\}_{i=1}^n\subset\R^d$ and $\eta\in(0,1)$ define the \emph{resilience}
\[
\mathrm{Res}_\eta(z)\defeq \sup_{\substack{S\subseteq[n]\\ |S|\le \eta n}}
\Bigl\|\frac{1}{n}\sum_{i\in S} z_i\Bigr\|_2.
\]

We define our algorithm as \cref{algo:two-stage-robust-posterior-mean}.
\begin{algorithmbox}[Two-stage robust posterior mean estimation]\label{algo:two-stage-robust-posterior-mean}
    \mbox{}\\
    \textbf{Input:} $\eta$-corrupted design matrix $\Xinput\in\R^{d\times n}$ (columns $x_i$),
    responses $\yinput\in\R^n$, corruption fraction $\eta\in(0,1/2)$, failure probability $\beta\in(0,1)$,
    prior variance $\sigma>0$ (so $\Sigma=\sigma^2 I_d$). \\
    \textbf{Output:} An estimator $\hat w\in\R^d$ of the posterior mean $w_{\rm post}=\E[w\mid \Xnull,\ynull]$.
\begin{enumerate}
    \item Assume $n \gtrsim (d+\log(1/\beta))/\eta$.
    Set $\beta_1=\beta_2=\beta_3=\beta/3$ and $\lambda \gets (n\sigma^2)^{-1}$.
    
    \item (\textbf{Stage 1: rough robust regression})
    Run the polynomial time robust regression algorithm from \cref{thm:bayesian-error-regression} to obtain an estimator $\tilde w$. 
    
    \item (\textbf{Stage 2: residualization})
    Compute residuals $r \gets \yinput - \Xinput^\top \tilde w$ and vectors
    $v \gets \Xinput r$.
    
    \item Set
    \[
        \delta_0 \gets C_1\Big(\sqrt{\tfrac{d+\log(1/\beta_1)}{n}}+\eta\Big),
    \]
    and
    \[
        \tau \gets 2\eta\,\delta_0
        + C_2(1+\delta_0)\Big(
        \sqrt{\eta\log\tfrac1\eta}\sqrt{\tfrac{d+\log(1/\beta_2)}{n}}
        + \eta\log\tfrac1\eta
        \Big).
    \]
    
    \item Run \cref{algo:completion-mean} on $(\{v_i\}_{i=1}^n,\eta,\tau)$ to obtain $\hat g$.
    
    \item Return
    \[
        \hat w \gets \tilde w + \frac{1}{1+\lambda}\big(\hat g-\lambda \tilde w\big).
    \]
\end{enumerate}
\end{algorithmbox}

\begin{algorithmbox}[Robust mean estimation by replacement]\label{algo:completion-mean}
    \mbox{}\\
    \textbf{Input:} Vectors $v_1,\dots,v_n\in\R^d$, corruption fraction $\eta\in(0,1/2)$,
    resilience bound $\tau>0$. \\
    \textbf{Output:} A robust estimate $\hat g$ of the clean mean $\frac1n\sum_{i=1}^n v_i^\star$.
\begin{enumerate}
    \item Find vectors $v'_1,\dots,v'_n\in\R^d$ and an index set $J\subseteq[n]$ with $|J|\le \eta n$
    such that:
    \begin{enumerate}
        \item $v'_i = v_i$ for all $i\notin J$;
        \item
        \[
            \sup_{\substack{S\subseteq[n]\\ |S|\le \eta n}}
            \left\|\frac{1}{n}\sum_{i\in S} v'_i\right\|_2 \le \tau .
        \]
    \end{enumerate}
    
    \item Output
    \[
        \hat g \gets \frac{1}{n}\sum_{i=1}^n v'_i .
    \]
\end{enumerate}
\end{algorithmbox}

\begin{lemma}[Certificate $\Rightarrow$ closeness to the clean mean]\label{lem:completion-close}
Let $\{z_i\}$ be $\eta$-corrupted from an unknown clean sequence $\{z_i^\star\}$, i.e.\ they differ on
at most $\eta n$ indices.  Suppose there exists $\tau$ such that
$\mathrm{Res}_{2\eta}(z^\star)\le\tau$ and
\cref{algo:completion-mean} outputs a replacement $z'$ satisfying
$\mathrm{Res}_{2\eta}(z')\le\tau$.  Then
\[
\Bigl\|\frac{1}{n}\sum_{i=1}^n z'_i - \frac{1}{n}\sum_{i=1}^n z_i^\star\Bigr\|_2 \le 2\tau.
\]
\end{lemma}
\begin{proof}
Let $\Delta\defeq\{i: z'_i\neq z_i^\star\}$.  Since $z$ differs from $z^\star$ on at most $\eta n$
indices and $z'$ differs from $z$ on at most $\eta n$ indices, we have $|\Delta|\le 2\eta n$.
Then
\[
\frac{1}{n}\sum_{i=1}^n (z'_i-z_i^\star)=\frac{1}{n}\sum_{i\in\Delta}(z'_i-z_i^\star),
\]
so by triangle inequality,
\[
\Bigl\|\frac{1}{n}\sum_{i=1}^n (z'_i-z_i^\star)\Bigr\|_2
\le
\Bigl\|\frac{1}{n}\sum_{i\in\Delta} z'_i\Bigr\|_2
+
\Bigl\|\frac{1}{n}\sum_{i\in\Delta} z_i^\star\Bigr\|_2
\le
\mathrm{Res}_{2\eta}(z')+\mathrm{Res}_{2\eta}(z^\star)\le 2\tau.
\]
\end{proof}

Now we provide the resilience bound.
Define the (clean) residuals at $\tilde w$:
\[
r_i^\star \defeq \ynull(i) - \langle \Xnull(\cdot,i),\tilde w\rangle
= \langle \Xnull(\cdot,i), w-\tilde w\rangle + \xi_i,
\qquad \xi_i\sim\mathcal{N}(0,1),
\]
and the clean vectors
\[
v_i^\star \defeq \Xnull(\cdot,i)\, r_i^\star \in \R^d.
\]
Write $\Delta\defeq w-\tilde w$ and note $\E[v_i^\star\mid \Delta]=\Delta$.

Define the clean residuals at $\tilde w$ by
\[
r_i^\star \defeq \ynull(i) - \langle \Xnull(\cdot,i),\tilde w\rangle
= \langle \Xnull(\cdot,i), w-\tilde w\rangle + \xi_i,
\qquad \xi_i\sim\mathcal{N}(0,1),
\]
and the clean vectors
\[
v_i^\star \defeq \Xnull(\cdot,i)\, r_i^\star \in \R^d.
\]
Write $\Delta\defeq w-\tilde w$.
Note that for every $i$,
\[
v_i^\star
=
\Xnull(\cdot,i)\Xnull(\cdot,i)^\top \Delta
+
\Xnull(\cdot,i)\,\xi_i.
\]

\begin{lemma}[Small-set resilience bound for $\{v_i^\star\}$]\label{lem:resilience-vstar}
Assume $\|\Delta\|_2\le \delta$, and let $k\defeq \lceil 2\eta n\rceil$.
Under the assumption that
\(
n\gtrsim \frac{d+\log(1/\beta_2)}{\eta},
\)
there exists an absolute constant $C>0$ such that with probability at least $1-\beta_2$ over the clean draw $(\Xnull,\ynull)$,

\[
\mathrm{Res}_{2\eta}(v^\star)
\le
C(1+\delta)\left(
\sqrt{\eta\log\frac{e}{\eta}}\sqrt{\frac{d+\log(1/\beta_2)}{n}}
+
\eta\log\frac{e}{\eta}
\right).
\]
\end{lemma}

\begin{proof}
Let $x_i\defeq \Xnull(\cdot,i)$, and for $S\subseteq[n]$ let $\Xnull_S$ be the submatrix formed by the
columns indexed by $S$. Also let $\xi_S\in\R^{|S|}$ denote the restriction of $(\xi_1,\dots,\xi_n)$ to $S$.

For every $S\subseteq[n]$ with $|S|\le k$, using
\[
v_i^\star = x_i x_i^\top \Delta + x_i\xi_i,
\]
we have
\[
\frac1n\sum_{i\in S} v_i^\star
=
\frac1n\,\Xnull_S \Xnull_S^\top \Delta
+
\frac1n\,\Xnull_S \xi_S.
\]
Therefore,
\[
\Bigl\|\frac1n\sum_{i\in S} v_i^\star\Bigr\|_2
\le
\frac1n \|\Xnull_S\|_{\op}^2\,\|\Delta\|_2
+
\frac1n \|\Xnull_S\|_{\op}\,\|\xi_S\|_2.
\]
Taking the supremum over all $|S|\le k$ yields
\begin{equation}\label{eq:resilience-decomp}
\mathrm{Res}_{2\eta}(v^\star)
\le
\frac{\delta}{n}\,\|\Xnull\|_{\mathrm{sp}}(k)^2
+
\frac{1}{n}\,\|\Xnull\|_{\mathrm{sp}}(k)\,
\max_{|S|\le k}\|\xi_S\|_2.
\end{equation}

We now bound the two factors on the right-hand side.

First, by Lemma~\ref{lem:sparse-spectral-norm}, with probability at least $1-\beta_2/2$,
\[
\|\Xnull\|_{\mathrm{sp}}(k)
\le
C_1\left(
\sqrt d + \sqrt{k\log\frac{en}{k}} + \sqrt{\log\frac{2}{\beta_2}}
\right).
\]
Hence on this event,
\begin{equation}\label{eq:sp2-bound}
\frac{\delta}{n}\,\|\Xnull\|_{\mathrm{sp}}(k)^2
\le
C\delta\left(
\frac{d+\log(1/\beta_2)}{n}
+
\frac{k}{n}\log\frac{en}{k}
\right).
\end{equation}

Second, since $(\xi_1,\dots,\xi_n)\sim N(0,I_n)$ and the coordinates are nonnegative after squaring,
\[
\max_{|S|\le k}\|\xi_S\|_2^2
=
\max_{|S|=k}\sum_{i\in S}\xi_i^2.
\]
Applying Theorem~\ref{thm:maximum-subset-sum-gaussian-correct} in dimension $n$ with subset size $k$,
we get that with probability at least $1-\beta_2/2$,
\begin{equation}\label{eq:xi-topk}
\max_{|S|\le k}\|\xi_S\|_2^2
\le
C_2\left(
k\log\frac{en}{k}
+
\log\frac{2}{\beta_2}
\right).
\end{equation}

On the intersection of the events \eqref{eq:sp2-bound} and \eqref{eq:xi-topk}, which has probability at least $1-\beta_2$, the second term in \eqref{eq:resilience-decomp} is bounded by
\begin{align*}
\frac{1}{n}\,\|\Xnull\|_{\mathrm{sp}}(k)\,\max_{|S|\le k}\|\xi_S\|_2
&\le
\frac{C}{n}
\left(
\sqrt d + \sqrt{k\log\frac{en}{k}} + \sqrt{\log\frac{1}{\beta_2}}
\right)
\left(
\sqrt{k\log\frac{en}{k}} + \sqrt{\log\frac{1}{\beta_2}}
\right) \\
&\le
C\left(
\sqrt{\frac{d+\log(1/\beta_2)}{n}}\,
\sqrt{\frac{k\log(en/k)+\log(1/\beta_2)}{n}}
+
\frac{k\log(en/k)+\log(1/\beta_2)}{n}
\right).
\end{align*}
Since $k=\lceil 2\eta n\rceil$, we have $k/n\asymp \eta$ and
\[
\log\frac{en}{k}\asymp \log\frac{e}{\eta}.
\]
Substituting these into the two bounds above and combining with \eqref{eq:resilience-decomp} proves
\[
\mathrm{Res}_{2\eta}(v^\star)
\le
C\delta\left(
\frac{d+\log(1/\beta_2)}{n}
+
\eta\log\frac{e}{\eta}
\right)
+
C\left(
\sqrt{\eta\log\frac{e}{\eta}}\sqrt{\frac{d+\log(1/\beta_2)}{n}}
+
\eta\log\frac{e}{\eta}
+
\frac{\log(1/\beta_2)}{n}
\right).
\]
Finally, under the assumption that
\[
n\gtrsim \frac{d+\log(1/\beta_2)}{\eta},
\]
the terms $\delta(d+\log(1/\beta_2))/n$ and $\log(1/\beta_2)/n$ are absorbed into
\[
C(1+\delta)\,\eta\log\frac{e}{\eta},
\]
which gives the final error bound.
\end{proof}
\begin{proof}[Proof of \cref{thm:robust-inefficient-linear-regression}]
Conditioned on $(\Xnull,\ynull)$, the posterior is Gaussian with covariance
\[
\mathrm{Cov}(w^*\mid \Xnull,\ynull)=(\Xnull\Xnull^\top+\sigma^{-2}\Id_d)^{-1}
= \frac{1}{n}(A^\star+\lambda \Id_d)^{-1}.
\]
On the event $\lambda_{\min}(A^\star)\ge 1/2$, we have
$(A^\star+\lambda \Id_d)^{-1}\preceq 2\Id_d$, hence
$\mathrm{Cov}(w\mid \Xnull,\ynull)\preceq \frac{2}{n}\Id_d$.
Therefore, with probability at least $1-\beta_3$ (possibly after adjusting constants),
\begin{equation}\label{eq:w-minus-wpost}
\|\wnull-w_{\rm post}\|_2 \le C_5\sqrt{\frac{d+\log(1/\beta_3)}{n}}.
\end{equation}

Apply Lemma~\ref{lem:resilience-vstar} with $\delta\defeq \|\Delta\|_2$ (given the rough estimator we have $\delta \le C_1(\sqrt{(d+\log(1/\beta_1))/n}+\eta)$).
Let $\tau$ be the RHS of Lemma~\ref{lem:resilience-vstar}.  The observed vectors
$v_i=x_i r_i$ are $\eta$-corrupted from $v_i^\star$, hence the completion-by-certificate procedure
(Definition~\ref{algo:completion-mean}) produces $\hat g$ satisfying
\begin{equation}\label{eq:g-est}
\|\hat g-g^\star(\tilde w)\|_2 \le 2\tau
\end{equation}
by Lemma~\ref{lem:completion-close}, with probability at least $1-\beta_2$.

Now we relate $\hat w$ to $w_{\rm post}$ and control the error from $A^\star-\Id_d$ .
Using \eqref{eq:posterior-identity} with $u=\tilde w$,
\[
w_{\rm post}=\tilde w + (A^\star+\lambda\Id_d)^{-1}\bigl(g^\star(\tilde w)-\lambda\tilde w\bigr).
\]
Recall $\hat w=\tilde w+\frac{1}{1+\lambda}(\hat g-\lambda\tilde w)$.
Subtracting yields
\begin{align*}
\hat w-w_{\rm post}
&=
\frac{1}{1+\lambda}(\hat g-g^\star(\tilde w))
+
\Big(\frac{1}{1+\lambda}\Id_d-(A^\star+\lambda\Id_d)^{-1}\Big)\bigl(g^\star(\tilde w)-\lambda\tilde w\bigr).
\end{align*}
Using the resolvent identity
\[
(A^\star+\lambda\Id_d)^{-1}-\frac{1}{1+\lambda}\Id_d
=
(A^\star+\lambda\Id_d)^{-1}(\Id_d-A^\star)\frac{1}{1+\lambda}\Id_d
\]
and $\frac{1}{1+\lambda}\le 1$, we get
\begin{equation}\label{eq:basic-decomp}
\|\hat w-w_{\rm post}\|_2
\le
\|\hat g-g^\star(\tilde w)\|_2
+
\|A^\star-\Id_d\|_{\rm op}\cdot \|w_{\rm post}-\tilde w\|_2,
\end{equation}
where we used that
$(A^\star+\lambda\Id_d)^{-1}(g^\star(\tilde w)-\lambda\tilde w)=w_{\rm post}-\tilde w$
by \eqref{eq:posterior-identity}.

Next,
\[
\|w_{\rm post}-\tilde w\|_2 \le \|w_{\rm post}-w\|_2+\|w-\tilde w\|_2.
\]
On the intersection of events of rough estimation error guarantee and \eqref{eq:w-minus-wpost},
\begin{equation}\label{eq:wpost-tildew}
\|w_{\rm post}-\tilde w\|_2
\le
C_6\Big(\sqrt{\frac{d+\log(1/\beta)}{n}}+\eta\Big).
\end{equation}
By the concentration of Wishart matrix, we have
\begin{equation}\label{eq:Astar-I}
\|A^\star-\Id_d\|_{\rm op}\le C_4\sqrt{\frac{d+\log(1/\beta)}{n}}.
\end{equation}
Combining \eqref{eq:wpost-tildew} and \eqref{eq:Astar-I},
\[
\|A^\star-\Id_d\|_{\rm op}\cdot \|w_{\rm post}-\tilde w\|_2
\le
C_7\left(\frac{d+\log(1/\beta)}{n}+\eta\sqrt{\frac{d+\log(1/\beta)}{n}}\right).
\]
Under the assumed regime $n\gtrsim (d+\log(1/\beta))/\eta$, we have
$\frac{d+\log(1/\beta)}{n}\lesssim \eta$ and $\eta\sqrt{\frac{d+\log(1/\beta)}{n}}\lesssim \eta$.
Hence,
\begin{equation}\label{eq:A-approx-cost}
\|A^\star-\Id_d\|_{\rm op}\cdot \|w_{\rm post}-\tilde w\|_2 \lesssim \eta.
\end{equation}

Finally,
plugging \eqref{eq:g-est} and \eqref{eq:A-approx-cost} into \eqref{eq:basic-decomp},
\[
\|\hat w-w_{\rm post}\|_2
\le
2\tau + O(\eta).
\]
Finally, using $\delta=\|w^*-\tilde w\|_2\lesssim \sqrt{\frac{d+\log(1/\beta)}{n}}+\eta$ inside
Lemma~\ref{lem:resilience-vstar} and absorbing the lower-order $O(\eta)$ term into
$\eta\log(1/\eta)$ (for $\eta\le 1/2$), we obtain
\[
\|\hat w-w_{\rm post}\|_2
\le
O\!\left(
\sqrt{\eta\log\frac1\eta}\sqrt{\frac{d+\log(1/\beta)}{n}}
+
\eta\log\frac1\eta
\right)
\]
with probability at least $1-\beta$ after a union bound over the events used above.
\end{proof}

\subsubsection{Robustness-to-privacy reduction}\label{sec:privacy-reduction-linear-regression}
Now using the reduction from privacy to robustness for exponential time algorithms~\cite{hopkins2023robustness}, we obtain \cref{lem:exponential-bayesian-error-regression}, restated here for convenience:

\restatetheorem{lem:exponential-bayesian-error-regression}

\begin{proof}
By \cref{thm:robust-inefficient-linear-regression}, we have an inefficient algorithm which achieves error rate $\eta\log(1/\eta)$ when $n\geq \tilde{\Omega}(d/\eta)$.
Using the blackbox privacy-robustness reduction (\cite[Lemma 2.1]{hopkins2023robustness}), we have a private algorithm which achieves error $\alpha$ when 
\begin{equation*}
    n\geq \tilde{\Omega}\Paren{\frac{d+\log(1/\beta)}{\alpha\epsilon}+\frac{d\log(\sigma\sqrt{d})}{\epsilon}+\frac{(d+\log(1/\beta))\eta}{\alpha^2}}\,.
\end{equation*}
under the condition that $\alpha\geq \eta \log(1/\eta)$. 
Now taking $\alpha=\eta \log(1/\eta)+\frac{d+\log(1/\beta)}{n\e} \log\Paren{\frac{d+\log(1/\beta)}{n\e}}$, we have the claim of \cref{lem:exponential-bayesian-error-regression}.
\end{proof}

\subsection{Proof of \cref{lem:robust-regression-constant-error}: Robust regression with constant error rate}\label{sec:robust-regression-constant-error}

Here we provide the details of the rough estimation procedure for robust regression to constant error in the weak prior regime. The result is stated below for convenience:

\restatelemma{lem:robust-regression-constant-error}

\noindent Recall the definition of the program $\cA = \cA_{\rm corr} \cup \cAres'$, where $\cA_{\rm corr}$ is defined in \cref{eq:corruption_constraints} and $\cAres'$ is defined in \cref{eq:resilience_constraints_rough}. Formally, our algorithm is described in \cref{algo:constant-error-regression}.
\begin{algorithmbox}[Robust algorithm for regression with constant error]\label{algo:constant-error-regression}
    \mbox{}\\
    \textbf{Input:} $\eta$ corrupted matrix $X_{\text{input}}\in \R^{n\times m},y_{\text{input}}\in \R^n$  \\
    \textbf{Output:} a polynomial time estimator $\hat{w}\in \R^d$
\begin{enumerate}
  \item Solve the degree-$O(1)$ SoS relaxation of the constraints $\cA(X,y,\xi,w;\Xinput,\yinput)=\cA_{\mathrm{corr}}(X,y,\xi)\cup \cAres'(X,y,w)$, where $t = O(1)$,
    obtaining a pseudo-expectation operator.
  \item Round the pseudo-distribution using $\tE[w]$.
\end{enumerate}
\end{algorithmbox}

\begin{proof}[Proof of \cref{lem:robust-regression-constant-error}]
    First, as proved by \cref{thm:gaussian-order-stat-sharp}, the constraints in the sum-of-squares program are satisfied with probability at least $1-\beta$.
    Next, we give the sum-of-squares identifiability proof. 
    Let $r\in \Set{0,1}^n$ be the set of corruptions. Let $v = 1 - (1 - \xi)\odot (1 - r)$ denote the indicator for the points which were corrupted and/or were marked by $\xi$ as corrupted.
    Then we have 
    \begin{equation*}
        (y-X^\top w) \odot (1-v) = (\ynull-\Xnull^\top w) \odot (1 - v)\,,
    \end{equation*}
    so
    \begin{equation*}
       ( \ynull-\Xnull^\top w)-(y-X^\top w)=
       (\ynull-\Xnull^\top w)\odot v - (y-X^\top w)\odot v\,.
    \end{equation*}
    We can bound the second term on the right-hand side using the short-flat decomposition of $y - X^\top w$:
    \begin{equation*}
     \norm{(y-X^\top w)\odot v}^2\leq \norm{z_1}^2+2\eta n\norm{z_2}_\infty^2\leq \eta n \log(1/\eta)+\log(1/\beta)\,.
    \end{equation*}
    On the other hand, we have 
    \begin{equation*}
        \ynull-\Xnull^\top w=\ynull-\Xnull^\top \wnull+\Xnull^\top(\wnull-w)
    \end{equation*}
    As the ground truth is feasible, $\ynull - \Xnull^\top w^*$ also admits a short-flat decomposition, so 
    \begin{equation*}
        \norm{(\ynull-\Xnull^\top \wnull)\odot v}^2 \leq \eta n\log(1/\eta)+\log(1/\beta)\,.
    \end{equation*}

    Therefore it remains to bound
    $\norm{\paren{\Xnull^\top (w-\wnull)}\odot v}^2$.
    We have the sum-of-squares upper bound
    \begin{equation*}
        \norm{\paren{\Xnull^\top (w-\wnull)}\odot v}^2=\sum_{i\in [n]} v_i \iprod{\Xnull(\cdot,i),w-\wnull}^2\leq \Norm{w-\wnull}^2\,,
    \end{equation*}
    where the second step follows by the selector lemma (\cref{lem:sos-subset-sum}).
    By triangle inequality, we then have
    \begin{equation*}
        \Norm{\ynull-\Xnull^\top w}^2\leq 2\eta n \log(1/\eta)\norm{w-\wnull}^2+2n\,.
    \end{equation*}
     We then have 
    \begin{equation*}
    \Norm{\Xnull^\top w-\Xnull^\top \wnull}^2\leq 2\eta n\log(1/\eta)\norm{w-\wnull}^2+4n \,,
    \end{equation*}
    Since we have $\Normop{\frac{1}{n}\Xnull\Xnull^\top-\Id}\leq \sqrt{\frac{d+\log(1/\beta)}{n}}$, the pseudo-distribution satisfies
    \begin{equation*}
        \Norm{w-\wnull}^2\leq O(1)\,.
    \end{equation*}
    By standard rounding, we obtain an estimator $\tilde{\E}[w]$ which achieves constant error rate.
\end{proof}

\subsection{Privatizing robust regression algorithms under weak prior}
\label{sec:privatize}

\subsubsection{Privatizing the rough estimator from Stage 1}
In this part, we privatize our algorithms for \cref{algo:constant-error-regression}.
We prove the following theorem. 
\begin{theorem}\label{thm:private-regression-constant-error}
There is a polynomial-time $\varepsilon$-DP algorithm that, given an $\eta$-corrupted sample $(X_{\mathrm{in}},y_{\mathrm{in}})$, outputs $\widehat w$ with
\[
\E\bigl[\|\widehat w-\wnull\|_2^2\bigr]\ \le\ O(1)
\]
as long as
\(
n \ \ge\ \widetilde{\Omega}\!\left(
\frac{d\,\log(\sigma\sqrt d)}{\varepsilon}
\right).
\)
Moreover, via concentration of $\wpost$ around the true Bayesian posterior mean $\E[\wnull\mid \Xnull,\ynull]$, the same bound holds (up to lower-order terms) for $\E\|\widehat w-\E[\wnull\mid \Xnull,\ynull]\|_2^2$.
\end{theorem}
\begin{definition}[(\(\alpha,\tau,T\))-certifiable  regressor ]\label{def:cert2}
Fix $R>0$ and $\alpha,\tau>0$. For an input $(X_{\mathrm{in}},y_{\mathrm{in}})$, a vector $\widetilde w\in\R^d$
is \emph{$(\alpha,\tau,T)$-certifiable} if there exists a degree-$O(1)$ linear functional $L$ (a pseudoexpectation) over indeterminates $(X,y,\xi,w)$ such that:
\begin{enumerate}[leftmargin=2.1em]
\item \textbf{Closeness/boundedness:}
$\norm{L[w]-\widetilde w}_2\le \alpha$ and $\norm{L[w]}_2\le 2R+\tau T$.
\item \textbf{Approximate feasibility (degree-6):} $L$ $\tau$-approximately satisfies the union of constraint families
\[
\mathcal A \;=\; \underbrace{\mathcal A_{\mathrm{corr}}}_{\text{corruption}}\ \cup\ 
\underbrace{\mathcal A_{\mathrm{res}}}_{\text{design/residual stability}}
\]
concretely:
\begin{itemize}[leftmargin=1.5em]
\item $\mathcal A_{\mathrm{corr}}$: $\xi\odot\xi=\xi$, $\|\xi\|_1\le \eta n$ and $(1-\xi(i))X(i,\cdot)=(1-\xi(i))X_{\mathrm{in}}(i,\cdot)$, $(1-\xi(i))y(i)=(1-\xi(i))y_{\mathrm{in}}(i)$ for all $i$.
\item $\mathcal A_{\mathrm{res}}$: degree-6 SoS surrogates for
$\big\|\frac{1}{n}XX^\top-I\big\|_{\op}\lesssim \sqrt{(d+\log(1/\beta))/n}$ and a decomposition $y-X^\top w=z_1+z_2$ with
$\|z_1\|_2^2\lesssim \eta n\log(1/\eta)+\log(1/\beta)$ and $\|z_2\|_\infty\lesssim \log(1/\eta)$.
\end{itemize}
\end{enumerate}
We choose $R=\Theta(\sigma\sqrt d)$.
\end{definition}

\begin{lemma}[Low-score sets localize near $w$]\label{lem:low-score-volume-regression-constant}
Let $\eta_\star=R$.
Let
\[
\alpha(\eta)^2\ \lesssim\ \eta^2\log\!\tfrac{1}{\eta}\ +\ \eta\cdot\frac{d+\log(1/\beta)}{n}\,.
\]
Then the following hold simultaneously. 
\begin{enumerate}[leftmargin=2.1em]
\item \textbf{Completeness (inner ball).} If $\norm{\widetilde w-\wpost}_2\le \alpha(\eta)$ then $\mathcal{S}(\widetilde w;\cdot)\le \eta n$.
\item \textbf{Soundness (outer ball).} If $\mathcal{S}(\widetilde w;\cdot)\le \eta_\star n$ then $\norm{\widetilde w-\wpost}_2\lesssim \alpha(\eta_\star)$.
\end{enumerate}
\end{lemma}
\begin{proof}[Proof of \cref{thm:private-regression-constant-error}]
It is easy to verify the quasi-convexity of the score function. 
To show that the score function is well-defined for $\tilde{w}\in \ball{d}{2R+n\tau+\alpha}$, we note that we can set $y=0$ and $X=0$. 
Then we can set $w$ be the projection of $\tilde{w}$ 
into the ball of radius $2R+\tau T$.
Now we apply the reduction in \cref{thm:efficient-reduction-meta-pure}.
Particularly, this leads to sample complexity
\[
n\ \gtrsim\ 
\max_{\eta\le\eta'\le 1}\frac{d\log\!\big((R+n\tau)/\alpha(\eta)\big)+\log(1/\beta_\eta)}{\varepsilon}.
\]
which concludes the proof.
\end{proof}

\subsubsection{Privatizing the refined estimator from Stage 2}\label{sec:private-least-square}
Finally we prove the following theorem for privately approximating the least square estimator. 
\begin{theorem}
There is a polynomial-time $\varepsilon$-DP algorithm that, given an $\eta$-corrupted sample $(X_{\mathrm{in}},y_{\mathrm{in}})$, outputs $\widehat w$ with
\[
\E\bigl[\|\widehat w-\wnull\|_2^2\bigr]\ \le\ O(\alpha^2)
\]
as long as
\(
n \ \ge\ \widetilde{\Omega}\left(
\frac{d\,\log(\sigma\sqrt d)}{\varepsilon}
\right).
\)
Moreover, via concentration of $\wpost$ around the true Bayesian posterior mean $\E[\wnull\mid \Xnull,\ynull]$, the same bound holds (up to lower-order terms) for $\E\|\widehat w-\E[\wnull\mid \Xnull,\ynull]\|_2^2$.
\end{theorem}
In this section, we privatize our \cref{algo:posterior-mean-regression}.
We make the following definition for certifiable posterior mean for regression.
\begin{definition}[(\(\alpha,\tau,T\))-certifiable posterior mean for regression]
Fix $R>0$ and $\alpha,\tau>0$. For an input $(X_{\mathrm{in}},y_{\mathrm{in}})$, a vector $\widetilde w\in\R^d$
is \emph{$(\alpha,\tau,T)$-certifiable} if there exists a degree-$O(1)$ linear functional $L$ (a pseudoexpectation) over indeterminates $(X,y,\xi,w,\wpost,M)$ such that:
\begin{enumerate}[leftmargin=2.1em]
\item \textbf{Closeness/boundedness:}
$\norm{L[\wpost]-\widetilde w}_2\le \alpha$ and $\norm{L[\wpost]}_2\le 2R+\tau T$.
\item \textbf{Approximate feasibility (degree-6):} $L$ $\tau$-approximately satisfies the union of constraint families
\[
\mathcal A \;=\; \underbrace{\mathcal A_{\mathrm{corr}}}_{\text{corruption}}\ \cup\ 
\underbrace{\mathcal A_{\mathrm{res}}''}_{\text{design/residual stability}}\ \cup\ 
\underbrace{\mathcal A_{\mathrm{LS}}}_{w_{\LS}=\frac{1}{n}Xy}\!,
\]
concretely:
\begin{itemize}[leftmargin=1.5em]
\item $\mathcal A_{\mathrm{corr}}$: $\xi\odot\xi=\xi$, $\|\xi\|_1\le \eta n$ and $(1-\xi(i))X(i,\cdot)=(1-\xi(i))X_{\mathrm{in}}(i,\cdot)$, $(1-\xi(i))y(i)=(1-\xi(i))y_{\mathrm{in}}(i)$ for all $i$.
\item $\mathcal A_{\mathrm{res}}''$: degree-6 SoS surrogates for
$\big\|\frac{1}{n}XX^\top-I\big\|_{\op}\lesssim \sqrt{(d+\log(1/\beta))/n}$ and a decomposition $y-X^\top w=z_1+z_2$ with
$\|z_1\|_2^2\lesssim \eta n\log(1/\eta)+\log(1/\beta)$ and $\|z_2\|_\infty^2\lesssim \log(1/\eta)+\log(1/\beta)/\eta n$, also we have $X^\top(w^*- w_\init)=b_1+b_2$ such that $\|b_1\|_2^2\lesssim \eta n\log(1/\eta)+\log(1/\beta)$ and $\|b_2\|_\infty^2\lesssim \log(1/\eta)+\log(1/\beta)/\eta n$.
\item $\mathcal A_{\mathrm{LS}}$: $w_\LS=\frac{1}{n}Xy$.
\end{itemize}
\end{enumerate}
We choose $R=\Theta(\sigma\sqrt d)$.
\end{definition}

\restatetheorem{thm:private-approximation-least-square}

\begin{proof}
Apply the robustness-to-privacy reduction to the score $\mathcal{S}$ on the domain $\ball{d}{\,2R+n\tau+\alpha\,}$.
Let $V_\eta$ denote the volume of $\{\widetilde w:\ \mathcal{S}(\widetilde w)\le \eta n\}$ inside this domain.
By \cref{lem:vol}:
\begin{itemize}
    \item \emph{Local regime} $\eta\le \eta'\le \eta_\star$: low-score sets are sandwiched between $\ell_2$-balls around $\wpost$ with radii $\alpha(\eta)$ and $\alpha(\eta')$, hence
    \[
    \log\!\frac{V_{\eta'}}{V_\eta}\ \lesssim\ d\cdot \log\!\frac{\alpha(\eta')}{\alpha(\eta)}.
    \]
    \item \emph{Global regime} $\eta_\star\le \eta'\le 1$: the score-$\le \eta_\star n$ set lies in an $O(\alpha(\eta_\star))$-ball while the ambient radius is $O(R+n\tau)$, so
    \[
    \log\!\frac{V_{\eta'}}{V_\eta}\ \lesssim\ d\cdot \log\!\frac{R+n\tau}{\alpha(\eta_\star)}.
    \]
\end{itemize}

Finally it is easy to verify that the score function is quasi-convex and can be efficiently evaluated.

The reduction in \cref{thm:efficient-reduction-meta-pure} then yields a sufficient condition
\[
n\ \gtrsim\ 
\max_{\eta\le\eta'\le\eta_\star}\frac{d\log(\alpha(\eta')/\alpha(\eta))+\log(1/\beta_\eta)}{\varepsilon \eta'}
\;+\;
\max_{\eta_\star\le\eta'\le 1}\frac{d\log\!\big((R+n\tau)/\alpha(\eta_\star)\big)+\log(1/\beta_\eta)}{\varepsilon}.
\]
Optimizing over $\eta$ using $\alpha(\eta)^2\asymp \eta\cdot\frac{d+\log(1/\beta)}{n}$ in the small-$\eta$ regime yields the two local terms
$\tfrac{d+\log(1/\beta)}{\alpha^{4/3}\varepsilon^{2/3}}$ and $\tfrac{d+\log(1/\beta)}{\alpha\varepsilon}$.
For the global term, take $R=\Theta(\sigma\sqrt d)$ so that the contribution is $\tfrac{d\log(\sigma\sqrt d)}{\varepsilon}$ (polylog factors suppressed).
\end{proof}